\pdfoutput=1
\documentclass{article}
\usepackage[utf8]{inputenc}
\usepackage{amsmath}
\usepackage{amsfonts}
\usepackage{amssymb}
\usepackage{amsthm}
\usepackage{tikz}
\usetikzlibrary{external}
\usetikzlibrary{decorations.markings}
\usepackage{xparse}
\usepackage[breaklinks]{hyperref}
\usepackage{cleveref}
\usepackage{todonotes}
\usetikzlibrary{arrows,calc}
\usepackage{soul}
\usepackage[T1]{fontenc}
\usepackage[english]{babel}
\usepackage{csquotes}

\usepackage[style=alphabetic,sorting=none,doi=false,isbn=false,url=false,maxbibnames=20,maxcitenames=2]{biblatex}
\renewbibmacro{in:}{%
	\ifentrytype{article}
	{}
	{\printtext{\bibstring{in}\intitlepunct}}}
\newbibmacro{string+doi}[1]{\iffieldundef{doi}{#1}{\href{https://dx.doi.org/\thefield{doi}}{#1}}}
\DeclareFieldFormat{title}{\usebibmacro{string+doi}{\mkbibemph{#1}}}
\DeclareFieldFormat[article]{title}{\usebibmacro{string+doi}{\mkbibquote{#1}}}
\DeclareFieldFormat[incollection]{title}{\usebibmacro{string+doi}{\mkbibquote{#1}}}

\bibliography{refs}

\usepackage{graphicx} 
\graphicspath{{figures/}}

\newtheorem{lemma}{Lemma}
\newtheorem{corollary}{Corollary}
\newtheorem{theorem}{Theorem}
\newtheorem{definition}{Definition}
\newtheorem{prop}{Proposition}

\newcommand\ii{\mathrm i}

\usepackage{xcolor}
\newcommand{\jk}[1]{{\color{red} #1}}

\definecolor{cbORANGE}{RGB}{213, 94, 0}
\definecolor{cbBLUE}{RGB}{0, 114, 178}
\definecolor{cbGREEN}{RGB}{0,158,115}

\definecolor{DARKRED}{RGB}{181,0,0}
\definecolor{urlpurple}{RGB}{103,0,181}

\hypersetup{
	colorlinks=true,
	linkcolor=DARKRED,
	citecolor=cbBLUE,
	urlcolor=urlpurple
}

\NewDocumentCommand{\rotateReset}{s m}
{
	\IfBooleanTF #1
	{
		\pgftransformreset
		\pgftransformscale{#2}
	}
	{
		\pgftransformresetnontranslations
		\pgftransformscale{#2}
	}
	
}

\newcommand{\rotateFlat}[3]
{
	\rotateReset{#3}
	\pgftransformcm{1}{0}{0.3}{0.35}{\pgfpoint{#1cm}{#2cm}}
}
\newcommand{\rotateSide}[3]
{
	\rotateReset{#3}
	\pgftransformcm{0.3}{0.35}{0}{1}{\pgfpoint{#1cm}{#2cm}}
}

\DeclareDocumentCommand{\oddface}{s o o s}
{
	\begin{scope}[very thick,decoration={
			markings,
			mark=at position 0.5 with {\arrow{>}}},
		shift = {(.5,.5)},
		line cap = round
		]
		\IfNoValueF{#2}{
			\pgftransformxshift{#2cm}}
		\IfNoValueF{#3}{
			\pgftransformyshift{#3cm}}
		\IfBooleanT{#1}{\pgftransformrotate{90};}
		\fill[blue,opacity=0.2] (-.5,-.5) rectangle (.5,.5);
		
		\IfBooleanTF #4
		{
			\begin{scope}[white!40!black]
				\draw[postaction={decorate}] (-.5,-.5)--(.5,-.5);
				\draw[postaction={decorate}] (.5,.5)--(.5,-.5);
				\draw[postaction={decorate}] (.5,.5)--(-.5,.5);
				\draw[postaction={decorate}] (-.5,-.5)--(-.5,.5);
				\draw (0,0) circle (.1cm);
			\end{scope}
		}
		{
			
			\draw[postaction={decorate}] (-.5,-.5)--(.5,-.5);
			\draw[postaction={decorate}] (.5,.5)--(.5,-.5);
			\draw[postaction={decorate}] (.5,.5)--(-.5,.5);
			\draw[postaction={decorate}] (-.5,-.5)--(-.5,.5);
			\draw (0,0) circle (.1cm);
		}
		
	\end{scope}
}

\DeclareDocumentCommand{\whiteoddface}{s o o s}
{
	\begin{scope}[very thick,decoration={
			markings,
			mark=at position 0.5 with {\arrow{>}}},
		shift = {(.5,.5)},
		line cap = round
		]
		\IfNoValueF{#2}{
			\pgftransformxshift{#2cm}}
		\IfNoValueF{#3}{
			\pgftransformyshift{#3cm}}
		\IfBooleanT{#1}{\pgftransformrotate{90};}
		\fill[white,opacity=0.3] (-.5,-.5) rectangle (.5,.5);
		
		\IfBooleanTF #4
		{
			\begin{scope}[white!40!black]
				\draw[postaction={decorate}] (-.5,-.5)--(.5,-.5);
				\draw[postaction={decorate}] (.5,.5)--(.5,-.5);
				\draw[postaction={decorate}] (.5,.5)--(-.5,.5);
				\draw[postaction={decorate}] (-.5,-.5)--(-.5,.5);
				\draw (0,0) circle (.1cm);
			\end{scope}
		}
		{
			
			\draw[postaction={decorate}] (-.5,-.5)--(.5,-.5);
			\draw[postaction={decorate}] (.5,.5)--(.5,-.5);
			\draw[postaction={decorate}] (.5,.5)--(-.5,.5);
			\draw[postaction={decorate}] (-.5,-.5)--(-.5,.5);
			\draw (0,0) circle (.1cm);
		}
		
	\end{scope}
}

\NewDocumentCommand{\evenface}{s o o s}
{
	\begin{scope}[very thick,decoration={
			markings,
			mark=at position 0.5 with {\arrow{>}}},
		shift = {(.5,.5)},
		line cap = round
		] 
		\IfNoValueF{#2}{
			\pgftransformxshift{#2cm}}
		\IfNoValueF{#3}{
			\pgftransformyshift{#3cm}}
		
		\fill[white,opacity=0.5] (-.5,-.5) rectangle (.5,.5);
		
		\IfBooleanTF{#1}{
			\IfBooleanTF #4
			{
				\begin{scope}[white!40!black]
					\draw[postaction={decorate}] (-.5,-.5)--(.5,-.5);
					\draw[postaction={decorate}] (.5,-.5)--(.5,.5);
					\draw[postaction={decorate}] (.5,.5)--(-.5,.5);
					\draw[postaction={decorate}] (-.5,.5)--(-.5,-.5);
				\end{scope}
			}
			{
				\draw[postaction={decorate}] (-.5,-.5)--(.5,-.5);
				\draw[postaction={decorate}] (.5,-.5)--(.5,.5);
				\draw[postaction={decorate}] (.5,.5)--(-.5,.5);
				\draw[postaction={decorate}] (-.5,.5)--(-.5,-.5);
			}
		}
		{
			\IfBooleanTF #4
			{
				\begin{scope}[white!40!black]
					\draw[postaction={decorate}] (.5,-.5)--(-.5,-.5);
					\draw[postaction={decorate}] (.5,.5)--(.5,-.5);
					\draw[postaction={decorate}] (-.5,.5)--(.5,.5);
					\draw[postaction={decorate}] (-.5,-.5)--(-.5,.5);
				\end{scope}
			}
			{
				\draw[postaction={decorate}] (.5,-.5)--(-.5,-.5);
				\draw[postaction={decorate}] (.5,.5)--(.5,-.5);
				\draw[postaction={decorate}] (-.5,.5)--(.5,.5);
				\draw[postaction={decorate}] (-.5,-.5)--(-.5,.5);
			}
		}

	\end{scope}
}

\NewDocumentCommand{\Kagome}{o o s}{
	\begin{scope}[shift={#1}, rotate=#2]
		\foreach \x in {0,60,120,180,240,300}{
			\draw (\x:1)--(\x+60:1);
		}
		\foreach \x in {0,60,120,180,240,300}{
			\node at (\x:1){};
		}
		
		\IfBooleanTF #3
		{  }
		{ \begin{scope}[shift={(1,0)}]
				\node at (60:1) {};
				\node at (-60:1) {};
				\draw (0,0)--(60:1);
				\draw (0,0)--(-60:1);
				\draw (60:1)--(120:1);
				\draw (-60:1)--(-120:1);
		\end{scope}}
	\end{scope}
}

\NewDocumentCommand{\sixfourthreefour}{o o}{
	\begin{scope}[shift={#1},rotate={#2}]
		\foreach \x in {0,60,120,180,240,300}{
			\node at (\x:1){};
			\draw (\x:1)--(\x+60:1);
		}
		
		\begin{scope}[shift={(120:1)}]
			\node at (0,1){};
			\node at (1,1){};
			\draw (0,0)--(0,1);
			\draw (1,0)--(1,1);
			\draw (0,1)--(1,1);
			\draw (1,1)--($(1,1)+(-30:1)$);
		\end{scope}
		
		\begin{scope}[rotate=-60,shift={(120:1)}]
			\node at (0,1){};
			\node at (1,1){};
			\draw (0,0)--(0,1);
			\draw (1,0)--(1,1);
			\draw (0,1)--(1,1);
			\draw (1,1)--($(1,1)+(-30:1)$);
		\end{scope}
		
		\begin{scope}[rotate=-120,shift={(120:1)}]
			\node at (0,1){};
			\node at (1,1){};
			\draw (0,0)--(0,1);
			\draw (1,0)--(1,1);
			\draw (0,1)--(1,1);
		\end{scope}
	\end{scope}
}

\NewDocumentCommand{\twelvethree}{o o}{
	\begin{scope}[shift={#1},rotate=#2]
		\foreach \x in {15,45,75,105,135,165,195,225,255,285,315,345}{
			\node at (\x:{0.5/sin(15)}){};
			\draw[] (\x:{0.5/sin(15)})--(\x+30:{0.5/sin(15)});
		}
		
		\begin{scope}[shift={(15:{0.5/sin(15)})}]
			\node at (60:1){};
			\draw[] (0,0)--(60:1);
			\draw[] (60:1)--(120:1);
		\end{scope}
		
		\begin{scope}[shift={(-15:{0.5/sin(15)})}]
			\node at (-60:1){};
			\draw[] (0,0)--(-60:1);
			\draw[] (-60:1)--(-120:1);
		\end{scope}
	\end{scope}
}

\NewDocumentCommand{\eightfour}{o o}{
	\begin{scope}[shift={#1},rotate=#2]
		\foreach \x in {22.5,67.5,112.5,157.5,202.5,247.5,292.5,337.5}{
			\node at (\x:{0.5/sin(22.5)}){};
			\draw (\x:{0.5/sin(22.5)})--(\x+45:{0.5/sin(22.5)});
			\begin{scope}[shift={(45:{1/tan(22.5)})}]
				\node at (\x:{0.5/sin(22.5)}){};
				\draw (\x:{0.5/sin(22.5)})--(\x+45:{0.5/sin(22.5)});
			\end{scope}
		}
		
		\begin{scope}[shift={(112.5:{0.5/sin(22.5)})}]
			\node at (0,1){};
			\draw (0,0)--(0,1);
			\draw (0,1)--(1,1);
		\end{scope}
		
		\begin{scope}[shift={(-22.5:{0.5/sin(22.5)})}]
			\node at (1,0){};
			\draw (1,0)--(1,1);
			\draw (0,0)--(1,0);
		\end{scope}
	\end{scope}
}

\NewDocumentCommand{\twelvesixfour}{o o}{
	\begin{scope}[shift={#1},rotate=#2]
		\foreach \x in {15,45,75,105,135,165,195,225,255,285,315,345}{
			\node at (\x:{0.5/sin(15)}){};
			\draw (\x:{0.5/sin(15)})--(\x+30:{0.5/sin(15)});
		}
		
		\begin{scope}[rotate=30,shift={(105:{0.5/sin(15)})}]
			\node at (1,1){};
			\node at (0,1){};
			\draw (0,0)--(0,1);
			\draw (1,0)--(1,1);
			\draw (0,1)--(1,1);
			\draw (1,1)--($(1,1)+(30:1)$);
			\node at ($(1,1)+(30:1)$){};
			\draw ($(1,1)+(30:1)$)--($(1,1)+(30:1)+(-30:1)$);
			\node at ($(1,1)+(30:1)+(-30:1)$){};
			\draw ($(1,1)+(30:1)+(-30:1)$)--($(1,1)+(30:1)+(-30:1)+(-90:1)$);
		\end{scope}
		
		\begin{scope}[rotate=-30,shift={(105:{0.5/sin(15)})}]
			\node at (1,1){};
			\node at (0,1){};
			\draw (0,0)--(0,1);
			\draw (1,0)--(1,1);
			\draw (0,1)--(1,1);
			\draw (1,1)--($(1,1)+(30:1)$);
			\node at ($(1,1)+(30:1)$){};
			\draw ($(1,1)+(30:1)$)--($(1,1)+(30:1)+(-30:1)$);
			\node at ($(1,1)+(30:1)+(-30:1)$){};
			\draw ($(1,1)+(30:1)+(-30:1)$)--($(1,1)+(30:1)+(-30:1)+(-90:1)$);
		\end{scope}
		
		\begin{scope}[rotate=-90,shift={(105:{0.5/sin(15)})}]
			\node at (1,1){};
			\node at (0,1){};
			\draw (0,0)--(0,1);
			\draw (1,0)--(1,1);
			\draw (0,1)--(1,1);
		\end{scope}
	\end{scope}
}

\usepackage{authblk}

\title{A Compact Fermion to Qubit Mapping Part 2: Alternative Lattice Geometries}
\author[1]{Charles Derby}
\author[2]{Joel Klassen}
\affil[1,2]{Phasecraft Ltd.}
\affil[1]{University College London}

\begin{document}

\maketitle

\begin{abstract}
In recent work \cite{Derby20} a novel fermion to qubit mapping -- called the compact encoding -- was introduced which outperforms all previous local mappings in both the qubit to mode ratio, and the locality of mapped operators. There the encoding was demonstrated for square and hexagonal lattices. Here we present an extension of that work by illustrating how to apply the compact encoding to other regular lattices. We give constructions for variants of the compact encoding on all regular tilings with maximum  degree $4$. These constructions yield edge operators with Pauli weight at most $3$ and use fewer than $1.67$ qubits per fermionic mode. Additionally we demonstrate how the compact encoding may be applied to a cubic lattice, yielding edge operators with Pauli weight no greater than $4$ and using approximately $2.5$ qubits per mode. In order to properly analyse the compact encoding on these lattices a more general group theoretic framework is required, which we elaborate upon in this work. We expect this framework to find use in the design of fermion to qubit mappings more generally.
\end{abstract}

\section{Introduction}

The simulation of fermionic systems is a promising application of quantum computing which is likely to find use in fields such as quantum chemistry, material science, and high energy physics.  
An indispensable part of such simulations is a mapping from the modes of some fermionic system to the qubits of a quantum computer via a fermion to qubit mapping. 

For the purposes of quantum simulation a desirable feature of any fermion to qubit mapping is that commonplace fermionic operators are mapped to qubit operators that act on a small number of qubits (are algebraically local) which are nearby in terms of the interaction hardware of the device (are geometrically local). This reduces the circuit depth required to perform unitary operations generated by these operators, such as simulating the time evolution of the system or performing VQE. Typically these fermionic operators will conform to some natural local structure of the system under consideration, the canonical example being a square lattice. 

Early instances of fermion to qubit mappings  \cite{JordanWigner,bravyi2002fermionic} are neither algebraically nor geometrically local. These mappings employ $M$ qubits to represent $M$ fermionic modes, but map local fermionic interactions to qubit operators whose Pauli weights scale with system size by at least a logarithmic factor (shown to be optimal for encodings of this kind \cite{jiang2020optimal}).


Subsequently a number of mappings have been developed to solve this problem. These mappings achieve Pauli weights for fermionic operators -- tailored to specific lattices -- which do not scale with the system size while also preserving the geometric locality of the fermionic system (presuming the available hardware has similar connectivity)
\cite{bravyi2002fermionic, verstraete2005mapping, whitfield2016local, jiang2018majorana, chen2018exact, steudtner2019square, setia2019superfast, chien2020}. All of these mappings are stabilizer codes which encode the fermionic system into an entangled subspace of a multi-qubit system. Accordingly, these encodings use a greater number of qubits than the number of modes being simulated. We refer to such mappings as \emph{fermionic encodings}. 

It is worth noting that none of these fermionic encodings on qubits are truly local. Levin an Wen argue in \cite{levin2003fermions} that for any spin system with fermionic excitations, pairs of creation operators are necessarily connected by non-trivial string like structures, even if acting in geometrically separated regions.


In recent work \cite{Derby20} we introduced a new fermionic encoding -- which we call the compact encoding -- that outperforms all other local encodings in the Pauli weights of the ``edge'' and ``vertex'' generators of the even fermionic algebra and also in the ratio of qubits to simulated fermionic modes. We demonstrated the compact encoding for square and hexagonal lattices. For these lattice geometries the encoding gives rise to generators of the even fermionic algebera with maximum Pauli weight 3, and a qubit to mode ratio of less than $1.5$. Contrast this with the Verstraete-Cirac (VC) encoding \cite{verstraete2005mapping} which encodes these generators with maximum weight 4 with a qubit to mode ratio of 2, the lowest out of prior encodings in both these metrics. We closed that work with the suggestion that the scheme presented there could be generalized to other lattice geometries. In this work we go into more detail how this can be done.  

We illustrate how the compact encoding may be applied to all uniform tilings with degree less than $4$. In these cases the generators of the even fermionic algebra are maximum weight $3$, and the qubit to mode ratios range from $<1.25$ to $<1.67$. Some readers may be particularly interested in the $4.8.8$ uniform tiling, in which a spinful fermionic system on a square lattice may be embedded. Additionally we demonstrate how the compact encoding may be applied to a cubic lattice. In this case the generators of the even fermionic algebra are maximum weight $4$, and the qubit to mode ratio is less than $2.5$. Again, this is in contrast to the VC encoding which has maximum weight $4$ and qubit to mode ratio of 3 for this geometry. To facilitate the analysis of these encodings, we explain the general group theoretic features of this encoding scheme. A surprising emergent property of some of the encodings presented here is the disparity between the size of the code space and the fermionic space to be encoded.
In the cases where the code space is larger, multiple species of Majorana particles emerge in the encoded system. We show how additional stabilizers may be defined to remove these species.

The paper is structured as follows. In Section \ref{subsec:Encodings_on_graphs} we describe a group theoretic framework for understanding local fermionic encodings, which we will make use of throughout. In Section \ref{LW encoding} we review the compact encoding on a square lattice. In Section \ref{sec:species} we discuss the emergence of particle species in local fermionic encodings, and give a bound on their number. In Section \ref{sec:General planar codes} we define a restricted family of compact encodings on planar graphs about which we are able to prove some useful properties, and in \ref{sec:other codes} we apply this restricted family to all uniform planar graphs with degree less than $4$. In Section \ref{sec:cubic} we describe how the compact encoding may be applied to a cubic lattice.


 

\section{Local Fermionic Encodings on Graphs}\label{subsec:Encodings_on_graphs}
The aim of any fermionic mapping is to represent  fermionic operators by qubit operators.  All such mappings are necessarily restricted to discrete sets of fermionic modes, since the qubit system is finite dimensional. An important feature of natural fermionic systems is parity superselection, which forbids observables that do not commute with fermion number parity. Thus, insofar as one is concerned with representing natural observables, it suffices for fermionic mappings to represent operators which preserve parity, so called even fermionic operators. 

The even fermionic operators form a group algebra $\mathbb{C}[M_E]$ where $M_E$ is the group of even products of Majorana operators, with factors $\pm1$ and $\pm \ii$:
\begin{align}
&M_E = \left\lbrace \pm \ii^{0/1} \prod_j \gamma_j^{b_{2j}} \bar{\gamma}_j^{b_{2j+1}}  \; | \; b \textrm{ is an even parity bit string}\right\rbrace \\
&\gamma_j = a_j + a^\dagger_j, \;\;
\overline{\gamma}_j = (a_j - a^\dagger_j)/\ii.
\end{align}
Here $a_j^{\dagger}$ ($a_j$) are the standard fermionic creation (annihilation resp.) operators. Thus a fermionic encoding must at minimum constitute a group representation $\tau$ of $M_E$:
\begin{equation}
\tau: M_E \rightarrow L(\mathcal{H})\;,\; \mathcal{H} \textrm{ is the Hilbert space of a system of qubits.}
\end{equation}

The group $M_E$, and by extension the group algebra $\mathbb{C}[M_E]$, are generated by the ``edge'' $E_{ij}$ and ``vertex'' $V_j$ operators
\begin{equation}\label{eq:EVops1}
\begin{split}
E_{ij} := -\ii \gamma_i\gamma_j, \;
V_j := -\ii \gamma_i\overline{\gamma}_j,
\end{split}
\end{equation}
and phases $\{\pm 1, \pm \ii\}$,
which satisfy the relations \footnote{In the context of a group the relations $[x, y]=0$ and $\{x, y\}=0$ indicate $xy=yx$ and $xy=-yx$ resp.} 
\begin{equation}\label{eq:EVconds}
\begin{gathered}
E_{ij}^2=\mathbb{I},\quad V_j^2=\mathbb{I},\quad E_{ij}^\dagger=E_{ij},\quad V_j^\dagger=V_j,\quad E_{ij}=-E_{ji}, \\
i\neq j\neq k\neq l \neq m \, : \quad [V_i,V_j]=0,\quad [E_{ij},V_k]=0,\quad [E_{ij},E_{kl}]=0, \\
\{E_{ij},V_j\}=0,\quad \{E_{ij},E_{jk}\}=0
\end{gathered}
\end{equation}
and for any ordered set of modes $L=(i_1, i_2, \dots, i_{|L|})$, with $i_1 = i_{|L|}$ so that $L$ constitutes a cyclic path, the relation
\begin{equation}\label{eq:Eloopcond}
\ii^{|L|}\prod_{x=1}^{|L|-1}E_{i_x,i_{x+1}}=\mathbb{I}.
\end{equation}
These relations completely fix the group structure of $M_E$ and thus the group representation $\sigma$ is completely specified by qubit operators satisfying these relations. The edge and vertex operators are self-inverse, hermitian, and only mutually commute or anticommute. This makes multi-qubit Pauli operators natural candidates for their representation. Indeed every existing fermionic mapping represents edge and vertex operators as Pauli operators -- in some cases projected onto a subspace. 

There are two notions of locality which a design of a local fermionic encoding may wish to pursue: algebraic locality, wherein the \emph{number} of fermionic modes or qubits an operator acts upon is bounded; and geometric locality, wherein the \emph{maximum distance} -- for some distance measure -- between modes or qubits an operator acts upon is bounded. 
The compact encoding, and generalizations presented here, aims to preserve the geometric locality of operators, while minimizing the algebraic locality on qubits. We consider a distance measure on a fermionic system represented by a connected graph -- which we call a fermionic graph -- whose vertices correspond to fermionic modes, with distance given by the minimal path between modes. In the examples we consider the graphs are embedded in real space, and we employ proximity in real space as the distance measure on the qubit system. To this end it suffices to encode the vertex and edge operators associated with the edges and vertices of the graph to local edge and vertex operators on qubits, all other operators may be decomposed into these edge and vertex operators in a way that inherits this locality. 

In general it is not feasible to assign local Pauli operators to the edges and vertices of the fermionic graph which satisfy all of relations \ref{eq:EVconds} and \ref{eq:Eloopcond}. This is because Relation \ref{eq:Eloopcond} has a highly non-local character. Instead the strategy is to assign local Pauli operators which satisfy relations \ref{eq:EVconds} and project via the stabilizer formalism into the subspace which respects Relation \ref{eq:Eloopcond}. In this case the stabilizers are the products of closed loops of edges in the fermionic graph. In this way, one is constructing a representation of a group structure on the graph, defined by relations \ref{eq:EVconds}, and quotienting out the subgroup corresponding to the cycle space of the graph.

Formally, given an undirected, connected graph $G=(\textbf{E}=\{\{v_i,v_j\} \},\textbf{V}=\{v_i\})$ we define the finitely presented group $M_G$ whose presentation comprises the vertices $\bf{V}$, the directed versions of the edges $ \textbf{E}_D:=\{e_{ij}=(v_i, v_j), e_{ji}=(v_j, v_i) \;|\; \forall \{v_i, v_j \} \in \textbf{E}\}$  and $\{\pm 1, \pm \ii\}$, along with the relations
\begin{equation}\label{eq:MGconds}
\begin{gathered}
e_{ij}^2=\mathbb{I},\quad v_j^2=\mathbb{I},\quad e_{ij}=-e_{ji}, \\
i\neq j\neq k\neq l \neq m \, : \quad [v_i,v_j]=0,\quad [e_{ij},v_k]=0,\quad [e_{ij},e_{kl}]=0, \\
\{e_{ij},v_j\}=0,\quad \{e_{ij},e_{jk}\}=0 
\end{gathered}
\end{equation}
In the notation of group presentations:
\begin{equation}
M_G := \langle \mathbf{V} \cup \mathbf{E}_D \cup \{\pm 1, \pm i\} \vert \textrm{Eqs. \ref{eq:MGconds}}\rangle
\end{equation}
We define the abelian normal subgroup of directed cycles $\mathcal{C}_G \triangleleft M_G$ (see Appendix \ref{app:CGProperties} for details)
\begin{equation}\label{eq:cycleGroup}
\mathcal{C}_G = \left\lbrace \ii^{|L|}\prod_{x=1}^{|L|-1}e_{i_x,i_{x+1}} \; \vert \; L=(i_1, i_2, \dots, i_{|L|}),  i_1 = i_{|L|}, e_{i_x, i_{x+1}} \in \textbf{E}_D \right \rbrace
\end{equation}
We note that $\mathcal{C}_G$ is isomorphic to the cycle space of $G$ (see Appendix \ref{app:CGProperties}), and its elements are invariant under a choice of first and last element, or total orientation.

In Appendix \ref{app:CGProperties} we show that given a fermionic system, and a corresponding connected fermionic graph $G$, the group of even Majorana operators $M_E$ is isomorphic to the quotient group $M_G/\mathcal{C}_G$. Thus if we can identify a  representation 
\begin{equation}
\sigma: M_G \rightarrow L(\mathcal{H})\;,\; \mathcal{H} \textrm{ is the Hilbert space of a system of qubits,}
\end{equation}
such that all elements of $\sigma(\mathcal{C}_G)$ have a common $+1$ eigenspace $\mathcal{U} \subseteq \mathcal{H}$, then we may construct a representation $\tau$ of $M_E$ by considering the projection of $\sigma$ into $\mathcal{U}$ ie 
\begin{equation}
\sigma_{\mathcal{U}} := \textrm{Proj}_{\mathcal{U}} \circ \sigma.
\end{equation}
 Noting that $\mathcal{C}_G \subseteq ker(\sigma_{\mathcal{U}})$ we are free to define the action of $\tau$
\begin{equation}
\tau(m \mathcal{C}_G) = \sigma_{\mathcal{U}}(m) \;,\; \forall m \in M_G
\end{equation}
which constitutes a faithful representation of $M_E$. If we choose a representation $\sigma$ which maps into the multi qubit Pauli group
\begin{equation}
\{\pm 1, \pm \ii \} \times_i^N \{I_i, X_i, Y_i, Z_i\},
\end{equation}
then since $\mathcal{C}_G$ is abelian, and provided that $-1 \not \in \sigma(\mathcal{C}_G)$, $\mathcal{U}$ automatically exists \cite{gottesman1997} and corresponds to a stabilizer code space of the stabilizer group $S:=\sigma(\mathcal{C}_G).$

To summarize, given a connected fermionic graph $G$, corresponding to some fermionic system, to construct a local fermionic encoding it suffices to specify a local mapping $\sigma$ of the edges and vertices of the graph to multi-qubit Pauli operators, satisfying the relations \ref{eq:MGconds}, such that no element of $\mathcal{C}_G$ is mapped to $-1$. An example is the Jordan Wigner encoding, where the fermionic graph is a line, there are no cycles and so $\mathcal{C}_G$ is trivial, and $\sigma(e_{ij})=X_iY_j$, $\sigma(v_i)=Z_i$.

Throughout the text, we use the tilde superscript to denote the representation of an operator, ie $\tilde{e}_{ij}:=\sigma(e_{ij})$ and $\tilde{E}_{ij}:=\tau(E_{ij})$. In cases where we set $\tilde{E}_{ij}$ or $\tilde{V}_{ij}$ equal to a Pauli operator, it is implicit that this Pauli operator is projected into the subspace $\mathcal{U}$. Thus in the case where $e_{ij}$ is defined for some edge $(i,j)$, $\tilde{E}_{ij}$ and $\tilde{e}_{ij}$ may be used interchangeably (similarly for $\tilde{V}_i$ and $\tilde{v}_i$), as was the case in \cite{Derby20}. Nevertheless we wish to emphasize here the conceptual difference between $e_{ij}$, $v_{i}$ and $E_{ij}$, $V_{i}$. The former are elements of an abstract group $M_G$ corresponding only to edges and vertices of a particular graph, while the latter are elements of the group of even Majorana monomials $M_E$, which has no particular graph structure.

\subsection{Counting Stabilizers}

Using the construction described in the previous section, one specifies a representation $\sigma$ of $M_G$ which prescribes how the even Majorana monomials $M_E$ of a fermionic system are encoded into a stabilizer code space $\mathcal{U}$. However the logical operators that may possibly act on $\mathcal{U}$ -- the operators which commute with the stabilizer group $S:=\sigma(\mathcal{C}_G)$ -- may be larger than the group algebra $\mathbb{C}[M_E]$. In other words the code space may encode more than just parity preserving fermionic states, there may be additional structure.

This additional structure will be indicated by the dimension of the code space. The dimension of a fermionic system with $M$ modes is $2^M$. Fixing parity reduces the dimension by half, ie $2^{M-1}$. The dimension of the code space $\dim(\mathcal{U})$ will depend on the dimension of the original Hilbert space, and the size of the minimal set of generators of $S$, and may diverge from this value.


For a group $X$ we denote the minimum size of a set of generators -- the \textit{rank} -- of $X$ by $D(X)$. For an encoding employing $N$ qubits, the dimension of the code space is $\dim(\mathcal{U}) = 2^{N-D(S)}$. For our purposes it is most useful to consider how the degrees of freedom in the encoded space differ from the usual degrees of freedom of the fermionic space:
\begin{equation}
\dim(\mathcal{U}) = 2^{M + \Delta}
\end{equation}
\begin{equation}\label{eq:disparity}
\Delta := N-M-D(S)
\end{equation}
We call $\Delta$ the \textit{disparity}. When the disparity is $-1$ the code space encodes only the even fermionic states. When the disparity is $0$ the code space encodes the full fermionic Hilbert space. When the disparity is positive, the code space encodes $\Delta$ additional qubit degrees of freedom.

Here and throughout an important subgroup of $M_G$ is the group of all cycles that are mapped to the identity under $\sigma$, ie $ker(\sigma\vert_{\mathcal{C}_G})$, where $\sigma \vert_{\mathcal{C}_G}$ is the restriction of the representation $\sigma$ to the subgroup $\mathcal{C}_G$. For notational ease we define
\begin{equation}\label{eq:kernel}
K:= ker(\sigma\vert_{\mathcal{C}_G})
\end{equation}

\begin{theorem}
Given a fermionic encoding $\sigma$ for a connected fermionic graph $G$, the rank of the stabilizer group $S$ is
$D(S) =  D(\mathcal{C}_G) -D(K).$
\end{theorem}

\begin{proof}
$S$ corresponds to a faithful representation of the quotient group $\mathcal{C}_G/K$, and so $\vert S \vert = \vert \mathcal{C}_G/K \vert$. Furthermore, by Lagrange's theorem $$\vert \mathcal{C}_G/K \vert = \vert \mathcal{C}_G \vert / \vert K \vert.$$ Because $a^2=I\;,\; \forall a \in \mathcal{C}_G$, the elements of $\mathcal{C}_G$ correspond to the elements of the vector space $\mathbb{Z}_2^{D(\mathcal{C}_G)}$ so that $\vert \mathcal{C}_G \vert =  \vert \mathbb{Z}_2^{D(\mathcal{C}_G)}\vert = 2^{D(\mathcal{C}_G)}$. Similarly for $S$ and $K$:  $\vert S \vert = 2^{D(S)}$ and $\vert K \vert = 2^{D(K)}$. Thus
\begin{equation} \label{eq:stabilizerRank}
D(S)=D(\mathcal{C}_G) -D(K)
\end{equation}
\end{proof}

\begin{corollary}
\begin{equation} \label{eq:disparity2}
\Delta = (N-M)-(D(\mathcal{C}_G)-D(K))
\end{equation}
\end{corollary}

\begin{prop}
$D(\mathcal{C}_G) = \vert \textbf{E}\vert - \vert  \textbf{V} \vert+1$
\end{prop}
\begin{proof}
Recalling that $\mathcal{C}_G$ is isomorphic to the cycle space of $G$, $D(\mathcal{C}_G)$ is equal to the circuit rank of $G$, which satisfies:
\begin{equation} \label{eq:circuitRank}
D(\mathcal{C}_G)= \vert \textbf{E}\vert - \vert  \textbf{V} \vert+\beta_0
\end{equation}
$\beta_0$ is the zeroth Betti number, ie the number of connected components of the graph, in this case $1$. 
\end{proof}

\section{A Review of the Compact Encoding}\label{LW encoding}

The Low Weight encoding introduced in \cite{Derby20} can be understood through the concepts explained in the previous sections. This encoding uses the fermionic graph formalism introduced in \cref{subsec:Encodings_on_graphs}, where the graph is a square lattice, with each vertex corresponding to a fermionic mode. It defines qubit representations of the edge and vertex operators from \cref{eq:EVops1} such that any local interaction term has a Pauli weight no greater than 3 and does this with a qubit to mode ratio of $<1.5$. A review of the encoding follows. 


For each vertex $j$ assign a vertex qubit indexed by $j$. Label the faces of the lattice even and odd in a checkerboard pattern and for each odd face assign a face qubit as shown in \cref{fig:Scheme}. Give an orientation to each edge on the lattice such that they circulate around even faces, clockwise or anticlockwise, alternating every row of faces, also illustrated in \cref{fig:Scheme}.

\begin{figure}[ht]
	\begin{center}
		\begin{tikzpicture}[scale=1.18,>=stealth,thick]
		\foreach \x in {0,1,2,3,4}{
			\foreach \y in {0,1,2,3}{
				\node at (\x,\y)[circle,fill=black,scale=0.5]{};
			}
		}
		
		\foreach \x in {0.5,2.5}{
			\node at (\x+1,1.5)[circle,fill=black,scale=0.5]{};
			\foreach \y in {0.5,2.5}{
				\node at (\x,\y)[circle,fill=black,scale=0.5]{};
			}
		}
		
		\foreach \x in {0,1,2,3,4}{
			\foreach \y in {0,1,2}{
				\draw[->] (\x,{\y+0.5-pow(-1,\x)*0.3})--(\x,{\y+0.5+pow(-1,\x)*0.3});
			}
		}
		
		\foreach \x in {0,1,2,3}{
			\foreach \y in {0,1,2,3}{
				\draw[->] ({\y+0.5-pow(-1,\x)*0.3},\x)--({\y+0.5+pow(-1,\x)*0.3},\x);
			}
		}
		
		\foreach \x in {1,2,3,4,5}{
			\node at (\x-1.2,3.2)[scale=0.7]{$\x$};
		}
		\foreach \x in {6,7,8,9,10}{
			\node at (\x-6.2,2.2)[scale=0.7]{$\x$};
		}
		\foreach \x in {11,12,13,14,15}{
			\node at (\x-11.2,1.2)[scale=0.7]{$\x$};
		}
		\foreach \x in {16,17,18,19,20}{
			\node at (\x-16.2,0.2)[scale=0.7]{$\x$};
		}
		
		\begin{scope}[shift={(5.5,0)},faint/.style={opacity=0.2}]
		\foreach \x in {0,1,2,3,4}{
			\foreach \y in {0,1,2}{
				\draw[->,faint] (\x,{\y+0.5-pow(-1,\x)*0.3})--(\x,{\y+0.5+pow(-1,\x)*0.3});
			}
		}
		
		\foreach \x in {0.5,2.5}{
			\node at (\x+1,1.5)[circle,fill=black,scale=0.5]{};
			\foreach \y in {0.5,2.5}{
				\node at (\x,\y)[circle,fill=black,scale=0.5]{};
			}
		}
		
		\foreach \x in {0,1,2,3}{
			\foreach \y in {0,1,2,3}{
				\draw[->,faint] ({\y+0.5-pow(-1,\x)*0.3},\x)--({\y+0.5+pow(-1,\x)*0.3},\x);
			}
		}

		\foreach \x in {0,1,2,3,4}{
			\foreach \y in {0,1,2,3}{
				\node at (\x,\y)[circle,fill=black,scale=0.5]{};
			}
		}
		
		\node at (-.2,3)[scale=0.7]{$Z$};
		\node at (0,3+0.25)[scale=0.8]{$V_1$};
		
		\node at (0.5,0.5)[circle,fill=black,scale=0.5]{};
		\node at (1.5,1.5)[circle,fill=black,scale=0.5]{};
		
		\draw[->,white] (3-0.2,3)--(2.2,3)[];    
		\node at (2.5,3.2)[scale=0.8]{$E_{4,3}$};
		\node at (2.18,3)[scale=0.7]{$Y$};
		\node at (3-0.18,3)[scale=0.7]{$X$};
		\node at (2.5,2.7)[scale=0.7]{$Y$};
		\draw[->] (3-0.25,3)--(2.25,3)[];
		\draw[] (2.5,3)--(2.5,2.8)[];
		
		\draw[<-,white] (1,0.2)--(1,1-0.2)[];
		\node at (1.45,0.55)[scale=0.8]{$E_{12,17}$};
		\node at (1,0.2)[scale=0.7]{$X$};
		\node at (1,1-0.2)[scale=0.7]{$Y$};
		\node at (0.68,0.5)[scale=0.7]{$X$};
		\draw[<-] (1,0.3)--(1,1-0.3)[];
		\draw[] (1,0.5)--(0.75,0.5)[];
		
		\draw[<-,white] (4-0.2,0)--(3.2,0)[];    
		\node at (3.5,-0.3)[scale=0.8]{$E_{19,20}$};
		\node at (3.18,0)[scale=0.7]{$X$};
		\node at (4-0.18,0)[scale=0.7]{$Y$};
		\draw[<-] (4-0.25,0)--(3.25,0)[];
		
		\node[white] at (4.2,3.2){1};
		\end{scope}
		\end{tikzpicture}
	\end{center}
	\caption{(Left) Edge orientation for the encoding on a $4\times 5$ square lattice. (Right) Examples of encoded edge and vertex operators based on this layout.}
	\label{fig:Scheme}
\end{figure}
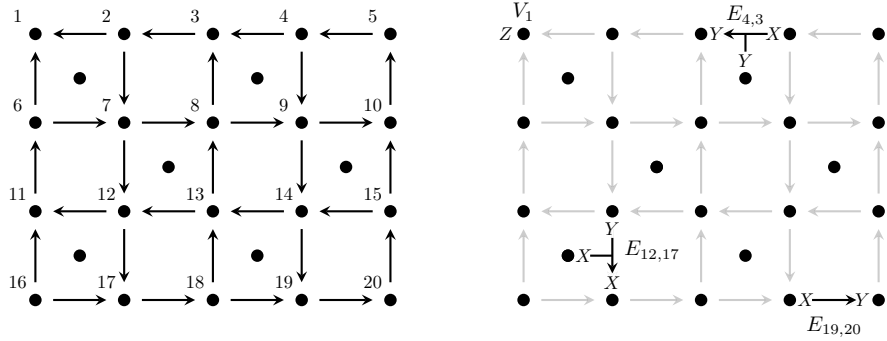

Let $f(i,j)$ index the unique odd face adjacent to edge $(i,j)$.  
For every edge $(i,j)$, with $i$ pointing to $j$, define the following encoded edge operators\footnote{The difference in sign introduce between the vertical up and down arrows is merely to ensure that cycles around odd faces are equal to $1$ and not $-1$.}. 
\begin{equation}
\tilde{e}_{ij} :=  \left\{ \begin{array}{rl}  X_i Y_j X_{f(i,j)} & \textrm{ if $(i,j)$ is oriented downwards} \\
-X_i Y_j X_{f(i,j)} & \textrm{ if $(i,j)$ is oriented upwards}  \\
X_i Y_j Y_{f(i,j)} & \textrm{ if $(i,j)$ is horizontal} \end{array} \right.
\end{equation}
\begin{equation}
\tilde{e}_{ji} := - \tilde{e}_{ij}.
\end{equation}
For those edges on the boundary which are not adjacent to an odd face, we omit the third Pauli operator which is meant to be acting on the non-existent face qubit. For every vertex $j$ define the encoded vertex operators
\begin{equation}
\tilde{v}_{j}:= Z_j
\end{equation}
This specifies all encoded vertex and edge operators. These operators satisfy the conditions in \ref{eq:MGconds}. This encoding is illustrated in \cref{fig:Scheme}.

\begin{figure}
	\centering
	\begin{tikzpicture}[scale=1.2,>=stealth,thick]
	\foreach \x in {0,1}{
		\foreach \y in {0,1}{
			\node at (\x,\y)[circle,fill=black,scale=0.5]{};
			\node at (0.5+\x-\y,-0.5+\x+\y)[circle,fill=black,scale=0.5]{};
		}
	}
	
	\node at (0.18,0)[scale=0.8]{$Y$};
	\node at (1-.18,0)[scale=0.8]{$X$};
	\node at (0.5,-.32)[scale=0.8]{$Y$};
	\draw[->] (1-0.25,0)--(0.25,0)[];
	\draw[] (.5,0)--(.5,-.2)[];
	
	\begin{scope}[rotate=180,shift={(-1,-1)}]
	\node at (0.18,0)[scale=0.8]{$Y$};
	\node at (1-.18,0)[scale=0.8]{$X$};
	\node at (0.5,-.32)[scale=0.8]{$Y$};
	\draw[->] (1-0.25,0)--(0.25,0)[];
	\draw[] (.5,0)--(.5,-.2)[];
	\end{scope}
	
	\node at (0,.18)[scale=0.8]{$X$};
	\node at (0,1-.18)[scale=0.8]{$Y$};
	\node at (0.68-1,0.5)[scale=0.8]{$X$};
	\draw[->] (0,0.3)--(0,1-0.3)[];
	\draw[] (0,0.5)--(-.25,0.5)[];
	
	\begin{scope}[rotate=180,shift={(-1,-1)}]
	\node at (0,.18)[scale=0.8]{$X$};
	\node at (0,1-.18)[scale=0.8]{$Y$};
	\node at (0.68-1,0.5)[scale=0.8]{$X$};
	\draw[->] (0,0.3)--(0,1-0.3)[];
	\draw[] (0,0.5)--(-.25,0.5)[];
	\end{scope}
	
	\node at (2,0.5)[]{$=$};
	
	\begin{scope}[shift={(3,0)}]
	\draw (0,0)--(1,0)--(1,1)--(0,1)--cycle;
	
	\foreach \x in {0,1}{
		\foreach \y in {0,1}{
			\node at (\x,\y)[circle,fill=black,scale=0.5]{};
			\node at (0.5+\x-\y,-0.5+\x+\y)[circle,fill=black,scale=0.5]{};
			\node at (0.15+0.7*\x,0.15+0.7*\y)[scale=0.8]{$Z$};
			
			\draw[dashed] (-0.5+1.5*\x,\y)--(1.5*\x,\y);
			\draw[dashed] (\y,-0.5+1.5*\x)--(\y,1.5*\x);
		}
	}
	
	\node at (0.5,1.7-2)[scale=0.8]{$Y$};
	\node at (0.5,1.5-0.2)[scale=0.8]{$Y$};
	\node at (-0.3,0.5)[scale=0.8]{$X$};
	\node at (1.3,0.5)[scale=0.8]{$X$};
	
	\end{scope}
	
	\begin{scope}[shift={(0,-2.5)}]
	\node at (.5,.5)[circle,fill=black,scale=0.5]{};
	\foreach \x in {0,1}{
		\foreach \y in {0,1}{
			\node at (\x,\y)[circle,fill=black,scale=0.5]{};
		}
	}
	
	\node at (0.18,0)[scale=0.8]{$Y$};
	\node at (1-.18,0)[scale=0.8]{$X$};
	\node at (0.5,+.32)[scale=0.8]{$Y$};
	\draw[->] (1-0.25,0)--(0.25,0)[];
	\draw[] (.5,0)--(.5,+.2)[];
	
	\begin{scope}[rotate=180,shift={(-1,-1)}]
	\node at (0.18,0)[scale=0.8]{$Y$};
	\node at (1-.18,0)[scale=0.8]{$X$};
	\node at (0.5,.32)[scale=0.8]{$Y$};
	\draw[->] (1-0.25,0)--(0.25,0)[];
	\draw[] (.5,0)--(.5,.2)[];
	\end{scope}
	
	\node at (0,.18)[scale=0.8]{$Y$};
	\node at (0,1-.18)[scale=0.8]{$X$};
	\node at (-.68+1,0.5)[scale=0.8]{$X$};
	\draw[<-] (0,0.3)--(0,1-0.3)[];
	\draw[] (0,0.5)--(.25,0.5)[];
	
	\begin{scope}[rotate=180,shift={(-1,-1)}]
	\node at (0,.18)[scale=0.8]{$Y$};
	\node at (0,1-.18)[scale=0.8]{$X$};
	\node at (-.68+1,0.5)[scale=0.8]{$X$};
	\draw[<-] (0,0.3)--(0,1-0.3)[];
	\draw[] (0,0.5)--(.25,0.5)[];
	\end{scope}
	
	\node at (2,0.5)[]{$=$};
	
	\begin{scope}[shift={(3,0)}]
	\draw (0,0)--(1,0)--(1,1)--(0,1)--cycle;
	\node at (.5,.5)[circle,fill=black,scale=0.5]{};
	\foreach \x in {0,1}{
		\foreach \y in {0,1}{
			\node at (\x,\y)[circle,fill=black,scale=0.5]{};
			\node at (0.15+0.7*\x,0.15+0.7*\y)[scale=0.8]{$\mathbb{I}$};
			
			\draw[dashed] (-0.5+1.5*\x,\y)--(1.5*\x,\y);
			\draw[dashed] (\y,-0.5+1.5*\x)--(\y,1.5*\x);
		}
	}
	
	\node at (.65,.5)[scale=0.8]{$\mathbb{I}$};
	
	\end{scope}
	\end{scope}
	\end{tikzpicture}
	\caption{Loops of edge operators around faces on the square lattice. Note that loops around even faces are non-trivial Pauli operators and loops around odd faces cancel out to identity. Phases have been omitted.}\label{fig:2D stabilizer}
\end{figure}
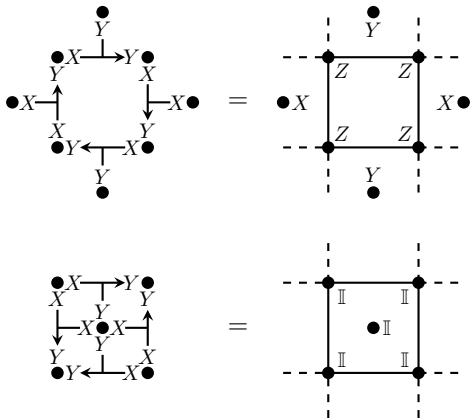 
 The cycle group $\mathcal{C}_G$ of a planar graph is generated by the cycles around faces. However the cycles around odd faces are mapped to the identity under $\sigma$. Thus $K$ is minimally generated by the cycles around odd faces, and so


\begin{equation}\label{eq:sq lattice ker}
D(K)=OF
\end{equation}
where $OF$ denotes the number of odd faces on the lattice. Since the cycles on any planar graph are minimally generated by all face cycles, the stabilizer group is then generated by the even face cycles (see \Cref{thm:planar kernel face cycles} in \Cref{sec:General planar codes}).

From \cref{eq:sq lattice ker,eq:disparity2} and the fact that $N-M=OF$ and $D(\mathcal{C}_G)=EF+OF$, where $EF$ is the number of even faces, we have that
\begin{equation}
\Delta = OF-EF
\end{equation}
for the encoding on a square lattice.

\section{Particle Species on Fermionic Encodings}\label{sec:species}

Different values of the disparity $\Delta$ cause an encoding to have different properties. In cases where $\Delta>0$ one can define what we call \textit{distinct particle species}. These find use in defining stabilizers to restrict excess space in these encodings and to detect errors. To illustrate this we will first demonstrate the effects of different disparities on the square lattice encoding and then present a result on particle species for general values of $\Delta$.

On a square lattice with a checkerboard pattern there are either equal numbers of even and odd faces or one more even/odd face so the disparity $\Delta$ may take the values -1, 0, 1. One may wonder what these different values imply, in particular the case where the encoding represents a larger space than the fermionic system it simulates.

In the $\Delta=-1$ case the codespace has dimension $2^{M-1}$, the even fermionic operators on $M$ modes provide a faithful representation of the set of linear operators on a space of this size so no other operators may be defined.
The fact that only even fermionic operators are represented implies that this is a fixed fermion parity sector. The operator $\prod_j V_j$ is equal to identity up to stabilizers in this case, identifying the parity sector as even. The odd parity sector may also be simulated by flipping the sign of one vertex operator (or indeed any odd number).

In the $\Delta=0$ case the codespace has dimension $2^{M}$. The even fermionic operators do not fully represent all linear operators over this space, the representation is only complete with the inclusion of odd (parity violating) fermionic operators. 
It suffices to define a representation of a single Majorana operator, since all other subsequent odd operators may be generated by applying edge and vertex operators. A single Majorana operator $\gamma_{j}$ satisfies the relations
\begin{equation}
\begin{gathered}
\{\gamma_i, \gamma_j \} = 0 \;,\quad\{V_i, \gamma_i \} = 0 \;,\quad [V_i, \gamma_j ]=0 \quad j\neq i \\
\{E_{ij}, \gamma_i \}=0 \;,\quad [E_{ij}, \gamma_k]=0 \quad k\neq j\quad k \neq i \\
\gamma_i^2=I \;,\quad E_{ij}\gamma_j = \gamma_i \;,\quad V_i\gamma_i=\ii \overline{\gamma}_i. \\
\end{gathered}
\end{equation}

Consider the corners of the lattice associated with an odd face. Such a corner $j$ either has arrows pointing into it or pointing away from it. If the arrows point into the corner then define the encoded Majorana operator $\tilde{\gamma}_j := X_j$, otherwise define it to be $\tilde{\gamma}_j:=Y_j$ (the encoded Majorana of the second kind, $\tilde{\overline{\gamma}}_j$, is given, up to a phase, by multiplying by a vertex operator). These single Majorana operators may be moved around the lattice by multiplication with strings of edge operators as shown in \cref{fig:SingleParticles}. From this it is clear that encoded single Majorana operators $\{\tilde{\gamma}_i,\tilde{\overline{\gamma}}_i\}$ take the form of strings of Paulis anchored at one end to their corner of origin, we denote this corner the \textit{injection point} of the Majorana operators. On the $\Delta=0$ square lattice, there are two equally suitable injection points from which single Majorana operators may be defined. If one is chosen to be the injection point for single Majorana operators then the single particle operators injected at the other corner will correspond to encoded Majorana hole operators $\{\tilde{h}_i,\tilde{\overline{h}}_i\}$, where $h_i:=\gamma_i \prod_j V_j$ and $\overline{h}_i:=\overline{\gamma}_i \prod_j V_j$.

\begin{figure}
	\centering
	\begin{tikzpicture}[scale=1.2,>=stealth,thick,faint/.style={opacity=0.2},line cap=round]
	\foreach \x in {0,1,2,3,4}{
		\foreach \y in {0,1,2,3,4}{
			\node at (\x,\y)[circle,fill=black,scale=0.5]{};
		}
	}
	
	\foreach \x in {0,1,2,3,4}{
		\foreach \y in {0,1,2,3}{
			\draw[<-,faint] (\x,{\y+0.5-pow(-1,\x)*0.3})--(\x,{\y+0.5+pow(-1,\x)*0.3});
			\draw[->,faint] ({\y+0.5-pow(-1,\x)*0.3},\x)--({\y+0.5+pow(-1,\x)*0.3},\x);
		}
	}
	
	\node at (2.5,1.5)[circle,fill=black,scale=0.5]{};
	\node at (2.5,3.5)[circle,fill=black,scale=0.5]{};
	\node at (0.5,1.5)[circle,fill=black,scale=0.5]{};
	\node at (0.5,3.5)[circle,fill=black,scale=0.5]{};
	\node at (1.5,2.5)[circle,fill=black,scale=0.5]{};
	\node at (1.5,.5)[circle,fill=black,scale=0.5]{};
	\node at (3.5,.5)[circle,fill=black,scale=0.5]{};
	\node at (3.5,2.5)[circle,fill=black,scale=0.5]{};
	
	\begin{scope}[shift={(0,1)}]
	\draw[cbBLUE,->] (0,2.8)--(0,2.2);
	\draw[cbBLUE] (0,2.5)--(0.3,2.5);
	\draw[cbBLUE,->] (0,1.8)--(0,1.2);
	\draw[cbBLUE,->] (0.2,1)--(0.8,1);
	\draw[cbBLUE] (0.5,1)--(0.5,.7);
	
	\node at (0,3.25)[scale=.8]{\textbf{\textit{Y}}};
	\node at (-.25,3)[cbBLUE,scale=.8]{$Z$};
	\node at (-.2,2)[cbBLUE,scale=.8]{$Z$};
	\node at (-.2,1)[cbBLUE,scale=.8]{$Z$};
	\node at (.5,2.7)[cbBLUE,scale=.8]{$X$};
	\node at (.3,.5)[cbBLUE,scale=.8]{$Y$};
	\node at (.8,1.2)[cbBLUE,scale=.8]{$Y$};
	\node at (0,3)[scale=.8,circle,draw=black,fill=none]{};
	\node at (1,1)[scale=.8,circle,draw=cbBLUE,fill=none]{};
	\end{scope}
	
	\begin{scope}[shift={(1,0)}]
	\draw[cbORANGE,<-] (2.8,0)--(2.2,0);
	\draw[cbORANGE] (2.5,0)--(2.5,.3);
	\draw[cbORANGE,->] (2,.2)--(2,.8);
	\draw[cbORANGE] (2,.5)--(2.3,.5);
	\draw[cbORANGE,->] (2,1.2)--(2,1.8);
	\draw[cbORANGE] (2,1.5)--(1.7,1.5);
	
	\node at (3,-.2)[scale=.8,cbORANGE]{$Z$};
	\node at (2,-.2)[scale=.8,cbORANGE]{$Z$};
	\node at (2.2,1.2)[scale=.8,cbORANGE]{$Z$};
	\node at (2.2,2.2)[scale=.8,cbORANGE]{$X$};
	\node at (2.5,.7)[scale=.8,cbORANGE]{$Z$};
	\node at (1.5,1.7)[scale=.8,cbORANGE]{$X$};
	\node at (2,2)[scale=.8,circle,draw=cbORANGE,fill=none]{};
	\end{scope}
	
	\end{tikzpicture}
	
	\caption{Single Majorana and hole operators at the circled vertices on a $\Delta=0$ lattice. \textbf{Bold black}: A single particle operator at an odd corner, we choose this to be the source of single Majorana operators $\gamma_j$. \textcolor{cbBLUE}{Blue}: A single Majorana operator in the bulk of the, transported from the top left odd corner by edge operators. \textcolor{cbORANGE}{Orange}: A Majorana hole operator in the bulk of the, transported from the other odd corner by edge operators.}
	\label{fig:SingleParticles}
\end{figure}
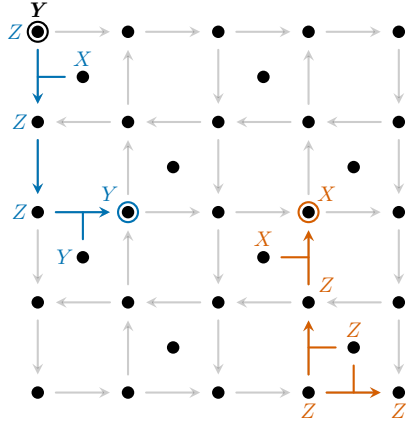

The encoded single Majorana operators $\{\tilde{\gamma}_i,\tilde{\overline{\gamma}}_i\}$ and $\{\tilde{h}_i,\tilde{\overline{h}}_i\}$ are examples of what we refer to as distinct Particle species.
\begin{definition}[Particle Species]\label{def:Maj Species}
	Given a fermionic encoding, a set of Pauli operators $\mathcal{M}=\{\mathcal{M}_i,\overline{\mathcal{M}}_i\}$, indexed over vertices is called a Particle Species if it satisfies the algebraic relations of single Majorana operators with respect to the encoded fermionic operators, i.e.
	\begin{equation}
	\begin{gathered}
	\{\mathcal{M}_i, \mathcal{M}_j \} = 0 \;,\quad\{\tilde{V}_i, \mathcal{M}_i \} = 0 \;,\quad [\tilde{V}_i, \mathcal{M}_j ]=0 \quad j\neq i \\
	\{\tilde{E}_{ij}, \mathcal{M}_i \}=0 \;,\quad [\tilde{E}_{ij}, \mathcal{M}_k]=0 \quad k\neq j\quad k \neq i \\
	\mathcal{M}_i^2=I \;,\quad \tilde{E}_{ij}\mathcal{M}_j = \mathcal{M}_i \;,\quad \tilde{V}_i\mathcal{M}_i=\ii \overline{\mathcal{M}}_i. \\
	\end{gathered}
	\end{equation}
\end{definition}
\begin{definition}[Distinct Particle Species]\label{def:Distinct Species}
	Particle Species $\mathcal{M}$ and $\mathcal{M}^\prime$ are said to be distinct if
	\begin{equation}
	\{\mathcal{M}_i,\mathcal{M}^\prime_i\}=0 \;,\quad [\mathcal{M}_i,\mathcal{M}^\prime_j]=0 \;,\quad j\neq i.
	\end{equation}
\end{definition}


%

In the $\Delta=1$ case the codespace is larger than the fermionic system with dimension $2^{M+1}$, we interpret this as simulating the full fermionic system and an extra logical qubit degree of freedom, i.e. $\mathcal{F}\otimes \mathbb{C}^2$. As with $\Delta=0$, single particle operators may be ``injected'' at the odd corners but instead there are four choices, meaning that four distinct particle species may be simultaneously defined.  One may always choose one of these species to be the encoded single Majorana operators of the fermionic system $\{\tilde{\gamma}_i,\tilde{\overline{\gamma}}_i\}$. The other three Majorana species correspond to $h_i\otimes X$, $h_i\otimes Y$, $h_i\otimes Z$ where the Paulis act on the extra logical qubit degree of freedom. The Paulis may be assigned to species arbitrarily so long as they form an anticommuting set. These Pauli operators can be isolated by fusing pairs of these Majorana species on the same vertex such that the fermionic part cancels to identity. As shown in \cref{fig:nonlocal XY}, the Paulis have a non-local string like form across the lattice. Having isolated these non-local Pauli operators, one of them may be taken as a stabilizer such that only the fermionic space is encoded.

\begin{figure}[ht]
	\centering
	\begin{tikzpicture}[scale=1.2,>=stealth,thick]
	\foreach \x in {0,1,2,3}{
		\foreach \y in {0,1,2,3}{
			\node at (\x,\y)[circle,fill=black,scale=0.5]{};
		}
	}

	\foreach \x in {0,1,2}{
		\foreach \y in {0,2}{
			\draw[<-] (0.2+\x,\y)--(0.8+\x,\y);
			\draw[<-] (\y,0.2+\x)--(\y,0.8+\x);
			\draw[->] (0.2+\x,1+\y)--(0.8+\x,1+\y);
			\draw[->] (1+\y,0.2+\x)--(1+\y,0.8+\x);
		}
	}
	
	\node at (2.5,0.5)[circle,fill=black,scale=0.5]{};
	\node at (2.5,2.5)[circle,fill=black,scale=0.5]{};
	\node at (0.5,0.5)[circle,fill=black,scale=0.5]{};
	\node at (0.5,2.5)[circle,fill=black,scale=0.5]{};
	\node at (1.5,1.5)[circle,fill=black,scale=0.5]{};

	\foreach \y in {0,1,2,3}{
		\node at (3+.2,\y+.2)[scale=.8]{$Z$};
	}
	
	\node at (2.5+.2,.5+.2)[scale=.8]{$X$};
	\node at (2.5+.2,2.5+.2)[scale=.8]{$X$};
	
	\node at (1.5,-0.4)[]{$\tilde{X}$};
	
	\begin{scope}[shift={(4.5,0)}]
	\foreach \x in {0,1,2,3}{
		\foreach \y in {0,1,2,3}{
			\node at (\x,\y)[circle,fill=black,scale=0.5]{};
		}
	}
	
	\foreach \x in {0,1,2}{
		\foreach \y in {0,2}{
			\draw[<-] (0.2+\x,\y)--(0.8+\x,\y);
			\draw[<-] (\y,0.2+\x)--(\y,0.8+\x);
			\draw[->] (0.2+\x,1+\y)--(0.8+\x,1+\y);
			\draw[->] (1+\y,0.2+\x)--(1+\y,0.8+\x);
		}
	}
	
	\node at (2.5,0.5)[circle,fill=black,scale=0.5]{};
	\node at (2.5,2.5)[circle,fill=black,scale=0.5]{};
	\node at (0.5,0.5)[circle,fill=black,scale=0.5]{};
	\node at (0.5,2.5)[circle,fill=black,scale=0.5]{};
	\node at (1.5,1.5)[circle,fill=black,scale=0.5]{};
	
	\foreach \x in {0,1,2,3}{
		\node at (\x+.2,0.2)[scale=.8]{$Z$};
	}
	
	\node at (2.5+.2,.5+.2)[scale=.8]{$Y$};
	\node at (.5+.2,.5+.2)[scale=.8]{$Y$};
	
	\node at (1.5,-0.4)[]{$\tilde{Y}$};
	\end{scope}

	\end{tikzpicture}
	\caption{Non-local $\tilde{X}$ and $\tilde{Y}$ operators formed by fusing different particle species. The pictured operators are formed by fusing particles from the right two corners and the bottom two corners respectively. The string operator may be deformed by applying stabilizers but will remain anchored at these corners, meaning they will always commute with single $X$ or $Y$ errors that may occur there.}\label{fig:nonlocal XY}
\end{figure}
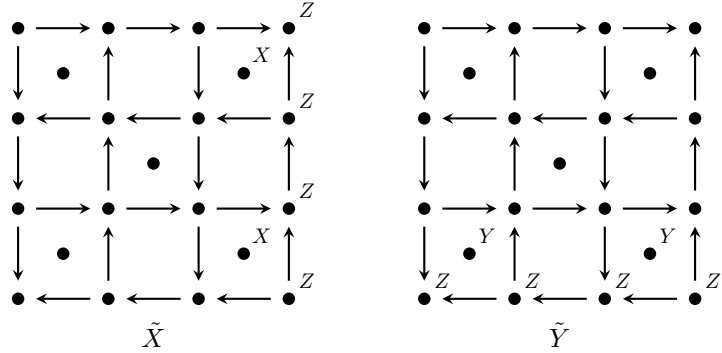

It may be desirable to restrict this excess Hilbert space via this stabilizer as it reduces the number of logical operators that act nontrivially on the fermionic space and therefore reduces the number of possible errors. In particular, restricting via one of these stabilizers on the $\Delta=1$ square lattice means that single $X$ and $Y$ errors may only occur on two of the corners rather than 4.

Other encodings introduced in this paper have $\Delta>1$. This greater disparity precipitates more particle species.

%

\begin{theorem}\label{thm:Delta Species relation}
	For any fermionic encoding, the number $m$ of distinct Majorana species is bounded from above by
	\begin{equation}\label{eq:Delta Species relation}
	m \leq 2\Delta+2,
	\end{equation}
\end{theorem}

\begin{proof}

	The earlier arguments for the $\Delta=-1,0$ cases on the square lattice apply generally and the number of particles in each case is consistent with \cref{eq:Delta Species relation}.
	
	Consider the $\Delta\geq 1$ case, in which the encoding represents a Hilbert space of dimension $2^{M+\Delta}$ which may be taken to be a fermionic system of $M$ modes composed with $\Delta$ qubits, $\mathcal{F}\otimes(\mathbb{C}^2)^{\otimes\Delta}$.  
	
	Consider a set of $m$ different simultaneously defined species $\{\mathcal{M}^{(k)}_i\}_{k=1}^{m}$. Any one of these may be chosen to be the encoded single Majoranas, choose $\mathcal{M}^{(1)}=\{\tilde{\gamma}_i,\tilde{\overline{\gamma}}_i\}$. For any given vertex, consider the set of operators 
	\begin{equation}\label{eq:particle fusion}
	\{P^i_k\}_{k=2}^{m}=\{\ii\;\mathcal{M}^{(1)}_i\mathcal{M}^{(k)}_i\prod_j\tilde{V}_j\}_{k=2}^{m}=\{-\ii\;\tilde{h}_i\;\mathcal{M}^{(k)}_i\}_{k=2}^{m}.
	\end{equation}
	Note that for any two vertices $i$ and $j$, $P^i_k$ and $P^j_k$ are related by a loop of edge operators and are therefore in fact the same operator, accordingly the vertex index is now redundant and will be dropped. 
	
	By the relations in \cref{def:Maj Species,def:Distinct Species}, each operator $P_k$ commutes with every fermionic operator, including our chosen single Majoranas $\mathcal{M}^{(1)}$ 
	and so each element only acts non-trivially on the excess $(\mathbb{C}^2)^{\otimes\Delta}$ space. All $P_k$ are Pauli operators as they are the product of Pauli operators (all single particle operators must be Paulis otherwise they could not be mapped to Majorana operators). Finally all $P_k$ are mutually anticommuting.
	
	Combining the above properties we see that $\{P_k\}_{k=2}^{m}$ is a set of mutually anticommuting Paulis acting only on the $(\mathbb{C}^2)^{\otimes\Delta}$ with an element for each defined single particle species except $\mathcal{M}^{(1)}$. The maximum size of such a set is $2\Delta+1$ by Lemma 4.5 from \cite{sarkar2019sets} and by Corollary 4.6 from the same paper, a set of this size exists so there can therefore be a maximum of $2\Delta+1$ $P_k$ operators and the maximum number of single particle species is $2\Delta+2$.
	
%
\end{proof}

From the above proof we see that particle species not chosen to be encoded Majoranas constitute Pauli operators on the excess space which can be isolated by fusion of different species on the same vertex. Thus in the case where a maximal set of these particle species can be identified 
they may be used to define stabilizers which restrict the code space to one with $\Delta =0$. 

\section{Generalizing the Compact Encoding}\label{sec:General planar codes}

The essential form of the compact encoding is as follows. A graph is supplied, and every edge of the graph is assigned an orientation. A vertex qubit is associated with every vertex of the graph. The vertex operators are defined to be $Z$ operators on their associated vertex qubits, and the edge operators are tentatively defined to act with an $X$ operator on the vertex which the edge is pointing away from, and a $Y$ operator on the vertex which the edge is pointing towards. Finally, auxilliary qubits are introduced, and the edge operators are made to act additionally on these auxilliary qubits in order to resolve any instances where pairs of edge operators sharing a vertex do not yet commute. The essential feature which makes the compact encoding compact is that the vertex operators are weight $1$, and the tentative edge operators are weight $2$, with an increase in weight bounded by the number of remaining anti-commutation relations needing to be resolved. An important additional feature which appears in the particular implementation of this encoding procedure is that the auxilliary qubits may be used to resolve many anti-commutation relations, thus significantly reducing the number of auxilliary qubits required.

In the square lattice encoding, the auxilliary qubits are confined to faces adjacent to the edges, and edges are only permitted to act on those adjacent auxilliary qubits. However in principle there is no reason why this must be the case. Auxilliary qubits could be used to resolve anti-commutation relations between edges not sharing a face. Furthermore, in the square lattice encoding, each face is associated with at most one auxilliary qubit. This is also not essential. However, one valuable consequence of imposing these kinds of constraints is that it makes the design and analysis of the encoding simpler. This motivates introducing a particular subclass of the compact encoding which lends itself well to analysis, and to which the square and hexagonal cases belong. This subclass aims to preserve the geometric locality and extremely low Pauli weight of the edge and vertex operators.


\begin{definition}[Weight 3 Planar Encoding]
\label{def:planar W3 encoding}
	A fermionic encoding on a planar fermionic graph where:
	\begin{itemize}
		\item Every vertex has one vertex qubit assigned;
		\item Every auxiliary qubit is associated with a unique face;
		\item Every face is associated with at most one auxiliary qubit;
		\item Every vertex operator is a single Pauli $Z$  on the assigned vertex qubit;
		\item Every edge operator is composed of a Pauli $X$ or $Y$ on each incident vertex qubit and a Pauli $X$, $Y$ or $Z$ on at most one face qubit from an adjacent face.
	\end{itemize}
\end{definition}


The following theorem allows the disparity of codes that fit \cref{def:planar W3 encoding} to be easily determined by counting faces. We note here that for conceptual simplicity we exclude the ``unbounded'' face surrounding the graph when considering faces, but it may be included with minor modifications.

\begin{theorem}\label{thm:planar kernel face cycles}
	Given a Weight 3 planar encoding, $K$ (as defined in eq. \ref{eq:kernel}) is minimally generated by face cycles. 
\end{theorem}

\begin{proof}
	The full set of face cycles on a planar graph form an independent basis of the cycle space so any subset of these will also be independent. It suffices now to prove that any element of $K$ can be reduced to identity by application of face cycles in the kernel.

	
For a planar graph, every cycle has a unique decomposition into a product of face cycles. First we argue that for a cycle in $K$, at least one of the face cycles in its unique decomposition must be in $K$. Consider a cycle $c \in K$, for any vertex in the cycle, the edge operators in $\sigma(c)$ must cancel to identity on the corresponding vertex qubit. Let one of the edge operators act on this qubit with $X$, label this operator $\tilde{E}_1$. Since edge operators may not apply $Z$ to vertex qubits the only way to cancel this is with another edge operator that applies $X$, label this $\tilde{E}_2$. These edge operators must anticommute so they must share a face and each act on the qubit with different Paulis ($X$ and $Y$). The product $E_1 E_2$ then applies a $Z$ to this face qubit which can only be cancelled by other edge operators around the associated face.
	
	The $Z$ may be cancelled by a single edge operator acting with a $Z$ or two acting with $X$ and $Y$. Let $E_Z$ be an edge operator applying $Z$ to the face qubit. $E_Z$ must anticommute with $E_1$ and $E_2$ so it must share vertex qubits with both, forming a 3-edge loop (to commute with either of them it would also need to share vertex qubits but this would require it to anticommute anyway), furthermore it would need to act with the same Paulis as $E_1$ and $E_2$ on their respective shared vertex qubits. The product $\ii^3 E_1E_2E_Z$ (with the appropriate orientation) must be identity, meaning the associated face cycle is in $K$ (see \cref{fig:planar kernel proof}).
	
	For the other case, let $E_X$ and $E_Y$ be edge operators around the face that act with $X$ and $Y$ on the face qubit respectively. For similar reasons to above, $E_Y$ ($E_X$) must anticommute with $E_1$ ($E_2$) meaning they will share a vertex qubit, as well as that, $E_X$ and $E_Y$ must anticommute and share a vertex themselves with the edges forming a 4-edged loop. As above, the edges that share vertex qubits will act on them with the same Paulis and the product $\ii^4E_1E_2E_XE_Z$ (with appropriate orientation) will be identity and the associated face cycle is in $K$ (see \cref{fig:planar kernel proof}).
	
Consider a cycle $c \in K$, which admits a decomposition into a set of face cycles $F \not \subseteq K$
\begin{equation}
c=\prod_{f \in F} f.
\end{equation} We may construct a new cycle $c'$ by removing  all cycles in $F \cap K$
\begin{equation}
c' = \left(\prod_{f \in F \cap K} f \right)c
\end{equation}
However $c'$ must be in $K$, which is a contradiction since it includes no face cycles in $K$.

\end{proof}
%


\begin{figure}
	\centering
	\begin{tikzpicture}[thick, line cap = round]
	\node at (0,0)[circle,fill=black,scale=.4]{};
	\node at (-.18,0)[scale=.6]{$X$};
	\draw (-.27,0)--(-.7,0);
	\draw[dotted] (-.7,0)--(-1,0);
	\node at (-.5,.2)[scale=.7]{$E_1$};
	
	\draw[line cap=butt,thin,-implies,double equal sign distance] (0.2,-.5) -- (.8,-.5);
	
	\begin{scope}[shift={(2,0)}]
	\node at (-.5,.2)[scale=.7]{$E_1$};
	\node at (.2,-.5)[scale=.7]{$E_2$};
	\node at (0,0)[circle,fill=black,scale=.4]{};
	\node at (-.18,0)[scale=.6]{$X$};
	\draw (-.27,0)--(-.7,0);
	\draw[dotted] (-.7,0)--(-1,0);
	\node at (0,-.2)[scale=.6]{$X$};
	\draw (0,-.34)--(0,-.7);
	\draw[dotted] (0,-.7)--(0,-1);
	\node at (-.5,-.5)[circle,fill=black,scale=.4]{};
	\node at (-.5+.18,-.5)[scale=.6]{$Y$};
	\node at (-.5,-.5+.2)[scale=.6]{$X$};
	\draw (-.5,0)--(-.5,-.16);
	\draw (0,-.5)--(-.23,-.5);
	
	\begin{scope}[shift={(0.3,0)}]
	\draw[line cap=butt,thin,-implies,double equal sign distance] (0.2,-.5) .. controls (1.2,-.5) and (.2,-1.5) .. (1.2,-1.5);
	\draw[line cap=butt,thin,-implies,double equal sign distance] (0.2,-.5) .. controls (1.2,-.5) and (.2,.5) .. (1.2,.5);
	\end{scope}
	
	\begin{scope}[shift={(3.1,1)}]
	\node at (-.5,.2)[scale=.7]{$E_1$};
	\node at (.2,-.5)[scale=.7]{$E_2$};
	\node at (0,0)[circle,fill=black,scale=.4]{};
	\node at (-1,0)[circle,fill=black,scale=.4]{};
	\node at (0,-1)[circle,fill=black,scale=.4]{};
	\node at (-.18,0)[scale=.6]{$X$};
	\node at (-1+.18,0)[scale=.6]{$X$};
	\node at (-.18,-1)[scale=.6]{$X$};
	\draw (-.27,0)--(-.73,0);
	\node at (0,-.2)[scale=.6]{$X$};
	\node at (0,-1+.2)[scale=.6]{$X$};
	\node at (-1,-.2)[scale=.6]{$X$};
	\draw (0,-.34)--(0,-.66);
	\node at (-.5,-.5)[circle,fill=black,scale=.4]{};
	\node at (-.5+.18,-.5)[scale=.6]{$Y$};
	\node at (-.5,-.5+.2)[scale=.6]{$X$};
	\draw (-.5,0)--(-.5,-.16);
	\draw (0,-.5)--(-.23,-.5);
	\node at (-.5-.15,-.5-.15)[scale=.6]{$Z$};
	\draw (-1,-.34) .. controls (-.9,-.9) .. (-.27,-1);
	\draw (-.83,-.83)--(-.75,-.75);
	\node at (-1,-1)[scale=.7]{$E_Z$};
	\end{scope}
	
	\begin{scope}[shift={(3.1,-1)}]
	\node at (-.5,.2)[scale=.7]{$E_1$};
	\node at (.2,-.5)[scale=.7]{$E_2$};
	\node at (-.5,-1.2)[scale=.7]{$E_X$};
	\node at (-1.2,-.5)[scale=.7]{$E_Y$};
	\node at (0,0)[circle,fill=black,scale=.4]{};
	\node at (-1,0)[circle,fill=black,scale=.4]{};
	\node at (0,-1)[circle,fill=black,scale=.4]{};
	\node at (-1,-1)[circle,fill=black,scale=.4]{};
	\node at (-.18,0)[scale=.6]{$X$};
	\node at (-1+.18,0)[scale=.6]{$X$};
	\node at (-.18,-1)[scale=.6]{$X$};
	\node at (-1+.18,-1)[scale=.6]{$X$};
	\draw (-.27,0)--(-.73,0);
	\draw (-.27,-1)--(-.73,-1);
	\node at (0,-.2)[scale=.6]{$X$};
	\node at (0,-1+.2)[scale=.6]{$X$};
	\node at (-1,-.2)[scale=.6]{$X$};
	\node at (-1,-1+.2)[scale=.6]{$X$};
	\draw (0,-.34)--(0,-.66);
	\draw (-1,-.34)--(-1,-.66);
	\node at (-.5,-.5)[circle,fill=black,scale=.4]{};
	\node at (-.5+.18,-.5)[scale=.6]{$Y$};
	\node at (-.5,-.5+.2)[scale=.6]{$X$};
	\node at (-.5-.18,-.5)[scale=.6]{$Y$};
	\node at (-.5,-.5-.2)[scale=.6]{$X$};
	\draw (-.5,0)--(-.5,-.16);
	\draw (0,-.5)--(-.23,-.5);
	\draw (-.5,-1)--(-.5,-1+.16);
	\draw (-1,-.5)--(-1+.23,-.5);
	\end{scope}
	
	\end{scope}

	\end{tikzpicture}
	\caption{Graphical representation of the proof of \cref{thm:planar kernel face cycles}. }\label{fig:planar kernel proof}
\end{figure}

The fact that face qubits may only be cancelled to identity in two ways yields the following corollary.
\begin{corollary}
The	face cycles in $K$ of a Weight 3 planar encoding may only be around 3 sided or 4 sided faces.
\end{corollary}

\begin{corollary}\label{cor:w3disp}
The disparity of a Weight-3 Planar Encoding is given by:
\begin{equation}
\Delta = F_K - EF,
\end{equation}
$F_K$ denoting the number of face cycles in $K$, and $EF$ denoting the faces without a qubit. 
\end{corollary}

\begin{proof}
Recalling that $\Delta = N-M - ( D(\mathcal{C}_G) - D(K) )$ we note that $N-M$ is the number of faces with a qubit, and $D(\mathcal{C}_G) - D(K) = F - F_K$ where $F$ is the number of faces. Thus $D(\mathcal{C}_G) - D(K) = EF + N-M - F_K$, and so $\Delta = F_K - EF$.
\end{proof}

\section{Examples of Weight-3 Planar Encodings}\label{sec:other codes}





In this section we present weight 3 planar encodings for every possible uniform tiling of degree $\leq 4$, except for the square and hexagonal tilings which are presented in \cite{Derby20}. With each encoding we also give the qubit to mode ratio and how their disparity $\Delta$ scales with lattice size and shape.

\subsection{Diagram Notation}

For each lattice structure we include a diagram showing the lattice connectivity and a unit cell of the encoding which shows the form of the edge operators. To to simplify diagrams, we use the following shorthand for the unit cell diagrams.

An edge incidence on a vertex will either be an arrow or a blank line. As before, an arrow denotes that the corresponding edge operator will act on that qubit with a $Y$ and a blank line denotes an $X$.
\begin{center}
\begin{tikzpicture}[thick,>=stealth,line cap=round,scale=1.3]
\node at (0,0)[circle,fill=black,scale=.4]{};
\draw (-.15,0)--(-.7,0);
\draw[dotted] (-.7,0)--(-1,0);

\node at (.5,0){=};

\begin{scope}[shift={(2,0)}]
\node at (0,0)[circle,fill=black,scale=.4]{};
\node at (-.18,0)[scale=.7]{$X$};
\draw (-.27,0)--(-.7,0);
\draw[dotted] (-.7,0)--(-1,0);
\end{scope}

\begin{scope}[shift={(0,-1)}]
\node at (0,0)[circle,fill=black,scale=.4]{};
\draw[<-] (-.15,0)--(-.7,0);
\draw[dotted] (-.7,0)--(-1,0);

\node at (.5,0){=};

\begin{scope}[shift={(2,0)}]
\node at (0,0)[circle,fill=black,scale=.4]{};
\node at (-.18,0)[scale=.7]{$Y$};
\draw[<-] (-.27,0)--(-.7,0);
\draw[dotted] (-.7,0)--(-1,0);
\end{scope}
\end{scope}
\end{tikzpicture}
\end{center}
With these choices of Pauli acting on the vertex qubits, the vertex operators $V_j$ on these codes are represented by $Z$ on the corresponding qubit as before.

All the encodings in the following sections have odd faces (faces with a qubit assigned) of only two forms. We use the following shorthands to denote how their surrounding edge operators act on the face qubit
\begin{equation}\label{eq:face_notation}
\begin{tikzpicture}[thick,>=stealth,line cap=round,scale=1.6]
\foreach \x in {0,1}{
	\foreach \y in {0,1}{
		\node at (\x,\y)[circle,fill=black,scale=.4]{};
	}
}

\node at (.5,.5)[circle,fill=black,scale=.4]{};

\draw[->] (.25,0)--(.75,0);
\draw[<-] (.25,1)--(.75,1);
\draw[->] (0,.25)--(0,.75);
\draw[<-] (1,.25)--(1,.75);

\draw (.5,0)--(.5,.3);
\draw (.5,1)--(.5,.7);
\draw (0,.5)--(.3,.5);
\draw (1,.5)--(.7,.5);

\node at (1.5,.5){=};

\begin{scope}[shift={(2,0)}]
\foreach \x in {0,1}{
	\foreach \y in {0,1}{
		\node at (\x,\y)[circle,fill=black,scale=.4]{};
	}
}

\node at (.5,.5)[circle,fill=black,scale=.4]{};

\draw[->] (.25,0)--(.75,0);
\node at (.15,0)[scale=.7]{$X$};
\node at (.85,0)[scale=.7]{$Y$};
\draw[<-] (.25,1)--(.75,1);
\node at (.15,1)[scale=.7]{$Y$};
\node at (.85,1)[scale=.7]{$X$};
\draw[->] (0,.25)--(0,.75);
\node at (0,.15)[scale=.7]{$X$};
\node at (0,.85)[scale=.7]{$Y$};
\draw[<-] (1,.25)--(1,.75);
\node at (1,.15)[scale=.7]{$Y$};
\node at (1,.85)[scale=.7]{$X$};

\draw (.5,0)--(.5,.2);
\draw (.5,1)--(.5,.8);
\draw (0,.5)--(.2,.5);
\draw (1,.5)--(.8,.5);

\node at (.5,.33)[scale=.7]{$X$};
\node at (.5,.67)[scale=.7]{$X$};
\node at (.33,.5)[scale=.7]{$Y$};
\node at (.67,.5)[scale=.7]{$Y$};
\end{scope}

\begin{scope}[shift={(4,0)}]
\node at (0,0)[circle,fill=black,scale=.4]{};
\node at (1,0)[circle,fill=black,scale=.4]{};
\node at (60:1)[circle,fill=black,scale=.4]{};
\node at (.5,{0.5*tan(30)})[circle,fill=black,scale=.4]{};

\draw ($(0,0)!.15!(1,0)$)--($(0,0)!.85!(1,0)$);
\draw ($(0,0)!.15!(60:1)$)--($(0,0)!.85!(60:1)$);
\draw ($(1,0)!.15!(60:1)$)--($(1,0)!.85!(60:1)$);

\begin{scope}[shift={(.5,{0.5*tan(30)})}]
\draw ($(0,0)!.4!(30:{0.5*tan(30)})$)--(30:{0.5*tan(30)});
\draw ($(0,0)!.4!(150:{0.5*tan(30)})$)--(150:{0.5*tan(30)});
\draw ($(0,0)!.4!(-90:{0.5*tan(30)})$)--(-90:{0.5*tan(30)});
\end{scope}

\node at (1.5,.5){=};

\begin{scope}[shift={(2,0)}]

\node at (0,0)[circle,fill=black,scale=.4]{};
\node at (1,0)[circle,fill=black,scale=.4]{};
\node at (60:1)[circle,fill=black,scale=.4]{};
\node at (.5,{0.5*tan(30)})[circle,fill=black,scale=.4]{};

\draw ($(0,0)!.25!(1,0)$)--($(0,0)!.75!(1,0)$);
\node at ($(0,0)!.15!(1,0)$)[scale=.7]{$X$};
\node at ($(0,0)!.85!(1,0)$)[scale=.7]{$X$};
\draw ($(0,0)!.25!(60:1)$)--($(0,0)!.75!(60:1)$);
\node at ($(0,0)!.15!(60:1)$)[scale=.7]{$X$};
\node at ($(0,0)!.85!(60:1)$)[scale=.7]{$X$};
\draw ($(1,0)!.25!(60:1)$)--($(1,0)!.75!(60:1)$);
\node at ($(1,0)!.15!(60:1)$)[scale=.7]{$X$};
\node at ($(1,0)!.85!(60:1)$)[scale=.7]{$X$};

\begin{scope}[shift={(.5,{0.5*tan(30)})}]
\draw ($(0,0)!.8!(30:{0.5*tan(30)})$)--(30:{0.5*tan(30)});
\draw ($(0,0)!.8!(150:{0.5*tan(30)})$)--(150:{0.5*tan(30)});
\draw ($(0,0)!.8!(-90:{0.5*tan(30)})$)--(-90:{0.5*tan(30)});

\node at ($(0,0)!.5!(30:{0.5*tan(30)})$)[scale=.7]{$X$};
\node at ($(0,0)!.5!(-90:{0.5*tan(30)})$)[scale=.7]{$Z$};
\node at ($(0,0)!.5!(150:{0.5*tan(30)})$)[scale=.7]{$Y$};
\end{scope}
\end{scope}
\end{scope}
\end{tikzpicture}
\end{equation}
Edge incidences may be switched arbitrarily provided that incidences on the same vertex qubit commute and anticommute in the same manner. The Paulis acting on face qubits may also be changed provided that opposite edges act with the same Pauli in the square case and that all three are different in the triangular case.

\subsection{The 4.8.8 Uniform Tiling}
We begin with the $4.8.8$ uniform tiling, illustrated in Figure \ref{fig:488}. This tiling has the notable property that it may be readily used to represent a spinful fermionic system on a square lattice, as illustrated in Figure \ref{fig:spinful}. Using the compact encoding on a $4.8.8$ uniform tiling, a spinful fermi-hubbard model on a square lattice may be represented on a planar hardware interaction architecture, with weight-2 spin-spin interactions, and weight-4 hopping terms.

\begin{figure}[ht]
	\centering
	\begin{tikzpicture}[scale=.8,thick,>=stealth,line cap=round,every node/.style={circle,fill=black,scale=.4}]
	\foreach \x in {22.5,67.5,112.5,157.5,202.5,247.5,292.5,337.5}{
		\node at (\x:{0.5/sin(22.5)}){};
		\draw[->] ($(\x:{0.5/sin(22.5)})!0.15!(\x+45:{0.5/sin(22.5)})$)--($(\x:{0.5/sin(22.5)})!0.85!(\x+45:{0.5/sin(22.5)})$);
	}
	
	\begin{scope}[shift={(45:{1/tan(22.5)})}]
	\foreach \x in {22.5,67.5,112.5,157.5,202.5,247.5,292.5,337.5}{
		\node at (\x:{0.5/sin(22.5)}){};
		\draw[<-] ($(\x:{0.5/sin(22.5)})!0.15!(\x+45:{0.5/sin(22.5)})$)--($(\x:{0.5/sin(22.5)})!0.85!(\x+45:{0.5/sin(22.5)})$);
	}
	\end{scope}
	
	\begin{scope}[shift={(112.5:{0.5/sin(22.5)})}]
	\node at (.5,.5){};
	\node at (0,1){};
	\draw[<-] (0,.15)--(0,.85);
	\draw[->] (.15,1)--(.85,1);
	\draw (.5,0)--(.5,.25);
	\draw (.5,1)--(.5,.75);
	\draw (0,.5)--(.25,.5);
	\draw (1,.5)--(.75,.5);
	\end{scope}
	
	\begin{scope}[shift={(-22.5:{0.5/sin(22.5)})}]
	\node at (.5,.5){};
	\node at (1,0){};
	\draw[<-] (1,.15)--(1,.85);
	\draw[->] (.15,0)--(.85,0);
	\draw (.5,0)--(.5,.25);
	\draw (.5,1)--(.5,.75);
	\draw (0,.5)--(.25,.5);
	\draw (1,.5)--(.75,.5);
	\end{scope}
	
	\begin{scope}[scale=.5,shift={($(-16,1)+(45:1)$)}]
	\eightfour[(0,0)][0]
	\eightfour[($(2,0)+(45:1)+(-45:1)$)][0]
	\eightfour[($(4,0)+(45:2)+(-45:2)$)][90]
	\eightfour[($(1,-1)+(-45:1)$)][0]
	\eightfour[($(3,-1)+(-45:2)+(45:1)$)][0]
	\eightfour[($(1,-1)+(-45:1)$)][90]
	\end{scope}
	\end{tikzpicture}
	\caption{The 4.8.8 Uniform Tiling and the unit cell of its encoding.}
	\label{fig:488}
\end{figure}
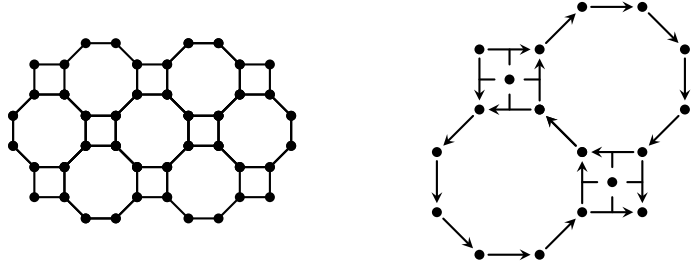

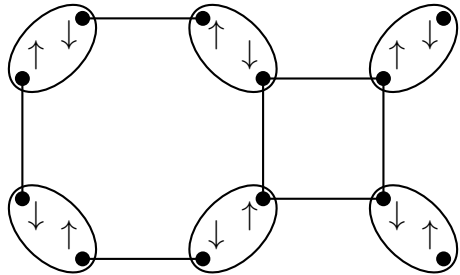
\begin{figure}
\centering
	\begin{tikzpicture}[scale=.8,thick,>=stealth,line cap=round,every node/.style={circle,fill=black,scale=.6}]
	\node[label={[shift={(-.3,-.2)}]\huge $\uparrow$}] at (1,0){};
	\node[label={[shift={(.3,-.2)}]\huge $\downarrow$}] at (3,0){};
	\node[label={[shift={(.3,-1.1)}]\huge $\downarrow$}] at (0,1){};
	\node [label={[shift={(-.3,-1.1)}]\huge $\uparrow$}]at (4,1){};
	\node[label={[shift={(.3,-.2)}]\huge $\uparrow$}] at (0,3){};
	\node[label={[shift={(-.3,-1.1)}]\huge $\downarrow$}] at (1,4){};
	\node[label={[shift={(.3,-1.1)}]\huge $\uparrow$}] at (3,4){};
	\node[label={[shift={(-.3,-.2)}]\huge $\downarrow$}] at (4,3){};
	\node[label={[shift={(-.3,-.2)}]\huge $\uparrow$}] at (7,0){};
	\node [label={[shift={(.3,-1.1)}]\huge $\downarrow$}]at (6,1){};
	\node[label={[shift={(.3,-.2)}]\huge $\uparrow$}] at (6,3){};
	\node[label={[shift={(-.3,-1.1)}]\huge $\downarrow$}] at (7,4){};

	\draw[rotate around={45:(.5,.5)}] (.5,.5) ellipse (15pt and 25pt);
	\draw[rotate around={-45:(.5,3.5)}] (.5,3.5) ellipse (15pt and 25pt);
	
	\draw[rotate around={-45:(3.5,.5)}] (3.5,.5) ellipse (15pt and 25pt);
	\draw[rotate around={45:(3.5,3.5)}] (3.5,3.5) ellipse (15pt and 25pt);
	
		\draw[rotate around={45:(6.5,.5)}] (6.5,.5) ellipse (15pt and 25pt);
	\draw[rotate around={-45:(6.5,3.5)}] (6.5,3.5) ellipse (15pt and 25pt);
	
	\draw (1,0)--(3,0);
	\draw (0,1)--(0,3);
	\draw (1,4)--(3,4);
	\draw (4,3)--(4,1);
	\draw (4,3)--(6,3);
	\draw (6,3)--(6,1);
	\draw (6,1)--(4,1);

	\end{tikzpicture}
	\caption{Layout of a spinful fermionic system on a square lattice, embedded in the $4.8.8$ uniform tiling.}\label{fig:spinful}
\end{figure}

Loops around octagonal faces are non-trivial Paulis and generate the stabilizer, loops around square faces are identity. The qubit to mode ratio is $<1.25$.

This tiling pattern has the same even/odd face pattern as the square lattice with octagons in place of the even square faces and so the disparity follows the same rules. That is, for a rectangular shaped lattice such as the one shown in \cref{fig:488} the disparity is given by
\begin{equation}
\Delta=OF-EF
\end{equation}
and may take values of -1, 0 or 1.

\subsection{The 6.4.3.4 Uniform Tiling}

\begin{figure}[ht]
	\centering
	\begin{tikzpicture}[thick,>=stealth,line cap=round,every node/.style={circle,fill=black,scale=.4}]
	\foreach \x in {0,60,120,180,240,300}{
		\node at (\x:1){};
		\draw[->] ($(\x:1)!0.15!(\x+60:1)$)--($(\x:1)!0.85!(\x+60:1)$);
	}
	
	\begin{scope}[shift={(120:1)}]
	\node at (0,1){};
	\node at (1,1){};
	\node at (.5,.5){};
	\draw[<-] (0,.15)--(0,.85);
	\draw[->] (1,.15)--(1,.85);
	\draw[->] (.15,1)--(.85,1);
	\draw[->] ($(1,1)+(-30:.15)$)--($(1,1)+(-30:.85)$);
	\draw (.5,0)--(.5,.25);
	\draw (.5,1)--(.5,.75);
	\draw (0,.5)--(.25,.5);
	\draw (1,.5)--(.75,.5);
	\end{scope}
	
	\begin{scope}[rotate=-60,shift={(120:1)}]
	\node at (0,1){};
	\node at (1,1){};
	\node at (.5,.5){};
	\draw[<-] (0,.15)--(0,.85);
	\draw[->] (1,.15)--(1,.85);
	\draw[->] (.15,1)--(.85,1);
	\draw[->] ($(1,1)+(-30:.15)$)--($(1,1)+(-30:.85)$);
	\draw (.5,0)--(.5,.25);
	\draw (.5,1)--(.5,.75);
	\draw (0,.5)--(.25,.5);
	\draw (1,.5)--(.75,.5);
	\end{scope}
	
	\begin{scope}[rotate=-120,shift={(120:1)}]
	\node at (0,1){};
	\node at (1,1){};
	\node at (.5,.5){};
	\draw[<-] (0,.15)--(0,.85);
	\draw[->] (1,.15)--(1,.85);
	\draw[->] (.15,1)--(.85,1);
	\draw (.5,0)--(.5,.25);
	\draw (.5,1)--(.5,.75);
	\draw (0,.5)--(.25,.5);
	\draw (1,.5)--(.75,.5);
	\end{scope}
	
	\begin{scope}[scale=.5,shift={($(-12,-.5)+(0,{-1*cos(30)})$)}]
	\sixfourthreefour[(0,0)][0]
	\sixfourthreefour[($(0,1)+(0,{2*sin(60)})$)][0]
	\sixfourthreefour[($(60:1)+(1,0)+(30:1)$)][0]
	\sixfourthreefour[($(-60:1)+(1,0)+(-30:1)$)][0]
	\sixfourthreefour[($(-60:1)+(1,0)+(-30:1)$)][-120]
	\sixfourthreefour[($(60:1)+(1,1)+(30:1)+(0,{2*sin(60)})$)][0]
	\sixfourthreefour[($(60:1)+(1,1)+(30:1)+(0,{2*sin(60)})$)][120]
	\sixfourthreefour[($(3,0)+(30:1)+(-30:1)$)][0]
	\sixfourthreefour[($(3,0)+(30:1)+(-30:1)$)][-120]
	\sixfourthreefour[($(3,1)+(30:1)+(-30:1)+(0,{2*sin(60)})$)][0]
	\sixfourthreefour[(0,0)][240]
	\sixfourthreefour[(0,0)][120]
	\sixfourthreefour[($(0,1)+(0,{2*sin(60)})$)][120]
	\end{scope}
	
	\end{tikzpicture}
	\caption{The 6.4.3.4 Uniform Tiling and the unit cell of its encoding.}
	\label{fig:6434}
\end{figure}
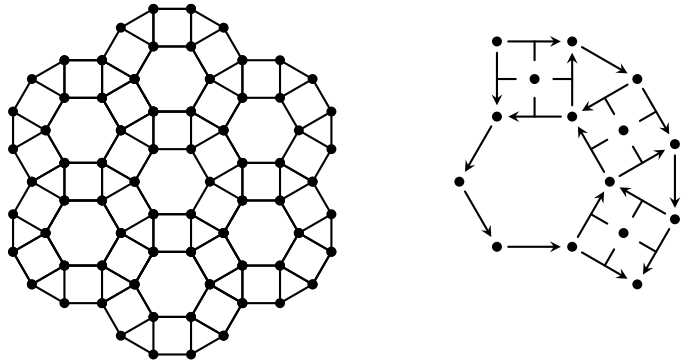

See \cref{fig:6434} for the lattice structure and the unit cell of the encoding. Loops around triangular and hexagonal faces are non-trivial Paulis and generate the stabilizer, loops around square faces are identity. The qubit to mode ratio is $<1.5$.

\begin{prop}\label{prop:6434 delta}
	For a connected 6.4.3.4 lattice without holes and with all hexagonal faces fully surrounded by square and triangular faces (e.g. \cref{fig:6434}), the disparity is
	$$
	\Delta=-1.
	$$
\end{prop}

\begin{proof}
	Consider a 6.4.3.4 lattice with a single fully surrounded hexagonal face. Clearly the disparity is -1 in this case as there are seven even faces and six odd faces. Consider constructing a lattice by adding fully surrounded hexagonal faces. Each hexagonal face added will amount to adding one of the following combinations of faces:
	
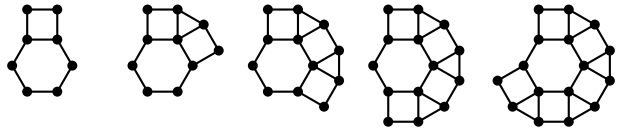
\begin{center}
		\begin{tikzpicture}[thick,>=stealth,line cap=round,every node/.style={circle,fill=black,scale=.4},scale=.4]
		\Kagome[(0,0)][0]*
		\draw (60:1)--($(60:1)+(0,1)$)--($(120:1)+(0,1)$)--(120:1);
		\node at ($(60:1)+(0,1)$){};
		\node at ($(120:1)+(0,1)$){};
		
		\begin{scope}[shift={(4,0)},rotate=-60]
		\Kagome[(0,0)][0]*
		\draw (60:1)--($(60:1)+(0,1)$)--($(120:1)+(0,1)$)--(120:1);
		\node at ($(60:1)+(0,1)$){};
		\node at ($(120:1)+(0,1)$){};
		\begin{scope}[rotate=60]
		\draw ($(60:1)+(0,1)$)--($(60:1)+(30:1)$);
		\draw (60:1)--($(60:1)+(0,1)$)--($(120:1)+(0,1)$)--(120:1);
		\node at ($(60:1)+(0,1)$){};
		\node at ($(120:1)+(0,1)$){};
		\end{scope}
		\end{scope}
		
		\begin{scope}[shift={(8,0)},rotate=-120]
		\Kagome[(0,0)][0]*
		\draw (60:1)--($(60:1)+(0,1)$)--($(120:1)+(0,1)$)--(120:1);
		\node at ($(60:1)+(0,1)$){};
		\node at ($(120:1)+(0,1)$){};
		\begin{scope}[rotate=60]
		\draw ($(60:1)+(0,1)$)--($(60:1)+(30:1)$);
		\draw (60:1)--($(60:1)+(0,1)$)--($(120:1)+(0,1)$)--(120:1);
		\node at ($(60:1)+(0,1)$){};
		\node at ($(120:1)+(0,1)$){};
		\begin{scope}[rotate=60]
		\draw ($(60:1)+(0,1)$)--($(60:1)+(30:1)$);
		\draw (60:1)--($(60:1)+(0,1)$)--($(120:1)+(0,1)$)--(120:1);
		\node at ($(60:1)+(0,1)$){};
		\node at ($(120:1)+(0,1)$){};
		\end{scope}
		\end{scope}
		\end{scope}
		
		\begin{scope}[shift={(12,0)},rotate=-180]
		\Kagome[(0,0)][0]*
		\draw (60:1)--($(60:1)+(0,1)$)--($(120:1)+(0,1)$)--(120:1);
		\node at ($(60:1)+(0,1)$){};
		\node at ($(120:1)+(0,1)$){};
		\begin{scope}[rotate=60]
		\draw ($(60:1)+(0,1)$)--($(60:1)+(30:1)$);
		\draw (60:1)--($(60:1)+(0,1)$)--($(120:1)+(0,1)$)--(120:1);
		\node at ($(60:1)+(0,1)$){};
		\node at ($(120:1)+(0,1)$){};
		\begin{scope}[rotate=60]
		\draw ($(60:1)+(0,1)$)--($(60:1)+(30:1)$);
		\draw (60:1)--($(60:1)+(0,1)$)--($(120:1)+(0,1)$)--(120:1);
		\node at ($(60:1)+(0,1)$){};
		\node at ($(120:1)+(0,1)$){};
		\begin{scope}[rotate=60]
		\draw ($(60:1)+(0,1)$)--($(60:1)+(30:1)$);
		\draw (60:1)--($(60:1)+(0,1)$)--($(120:1)+(0,1)$)--(120:1);
		\node at ($(60:1)+(0,1)$){};
		\node at ($(120:1)+(0,1)$){};
		\end{scope}
		\end{scope}
		\end{scope}
		\end{scope}
		
		\begin{scope}[shift={(17,0)},rotate=-240]
		\Kagome[(0,0)][0]*
		\draw (60:1)--($(60:1)+(0,1)$)--($(120:1)+(0,1)$)--(120:1);
		\node at ($(60:1)+(0,1)$){};
		\node at ($(120:1)+(0,1)$){};
		\begin{scope}[rotate=60]
		\draw ($(60:1)+(0,1)$)--($(60:1)+(30:1)$);
		\draw (60:1)--($(60:1)+(0,1)$)--($(120:1)+(0,1)$)--(120:1);
		\node at ($(60:1)+(0,1)$){};
		\node at ($(120:1)+(0,1)$){};
		\begin{scope}[rotate=60]
		\draw ($(60:1)+(0,1)$)--($(60:1)+(30:1)$);
		\draw (60:1)--($(60:1)+(0,1)$)--($(120:1)+(0,1)$)--(120:1);
		\node at ($(60:1)+(0,1)$){};
		\node at ($(120:1)+(0,1)$){};
		\begin{scope}[rotate=60]
		\draw ($(60:1)+(0,1)$)--($(60:1)+(30:1)$);
		\draw (60:1)--($(60:1)+(0,1)$)--($(120:1)+(0,1)$)--(120:1);
		\node at ($(60:1)+(0,1)$){};
		\node at ($(120:1)+(0,1)$){};
		\begin{scope}[rotate=60]
		\draw ($(60:1)+(0,1)$)--($(60:1)+(30:1)$);
		\draw (60:1)--($(60:1)+(0,1)$)--($(120:1)+(0,1)$)--(120:1);
		\node at ($(60:1)+(0,1)$){};
		\node at ($(120:1)+(0,1)$){};
		\end{scope}
		\end{scope}
		\end{scope}
		\end{scope}
		\end{scope}
		\end{tikzpicture}
	\end{center}
	These combinations all have the same number of even and odd faces and will not change the disparity from -1.
\end{proof}

The lattice can be made to simulate the full fermionic algebra by adding a single square face to the outer edge where single particle operators may be injected at its corners.


\subsection{The 4.6.12 Uniform Tiling}

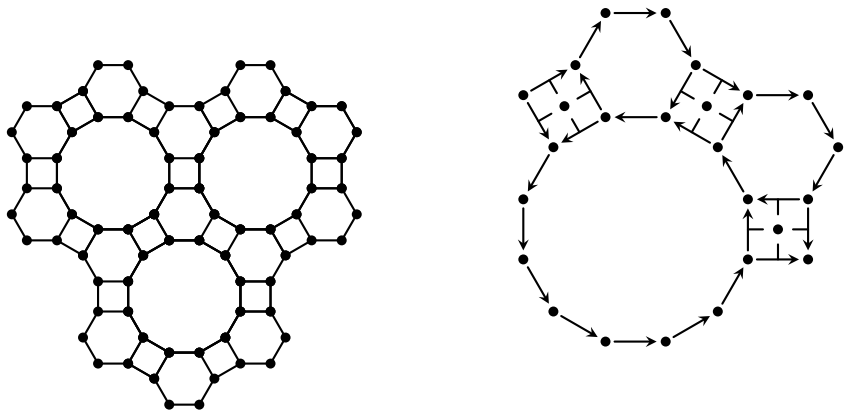
\begin{figure}[h]
	\centering
	\begin{tikzpicture}[thick,>=stealth,line cap=round,every node/.style={circle,fill=black,scale=.4}]
	\begin{scope}[scale=.8]
	\foreach \x in {15,45,75,105,135,165,195,225,255,285,315,345}{
		\node at (\x:{0.5/sin(15)}){};
		\draw[->] ($(\x:{0.5/sin(15)})!0.15!(\x+30:{0.5/sin(15)})$)--($(\x:{0.5/sin(15)})!0.85!(\x+30:{0.5/sin(15)})$);
	}
	
	\begin{scope}[rotate=30,shift={(105:{0.5/sin(15)})}]
	\node at (1,1){};
	\node at (0,1){};
	\node at (.5,.5){};
	\draw[<-] (0,.15)--(0,.85);
	\draw[->] (1,.15)--(1,.85);
	\draw[->] (.15,1)--(.85,1);
	\draw (.5,0)--(.5,.25);
	\draw (.5,1)--(.5,.75);
	\draw (0,.5)--(.25,.5);
	\draw (1,.5)--(.75,.5);
	\draw[->] ($(1,1)+(30:.15)$)--($(1,1)+(30:.85)$);
	\node at ($(1,1)+(30:1)$){};
	\draw[->] ($(1,1)+(30:1)+(-30:.15)$)--($(1,1)+(30:1)+(-30:.85)$);
	\node at ($(1,1)+(30:1)+(-30:1)$){};
	\draw[->] ($(1,1)+(30:1)+(-30:1)+(-90:.15)$)--($(1,1)+(30:1)+(-30:1)+(-90:.85)$);
	\end{scope}
	
	\begin{scope}[rotate=-30,shift={(105:{0.5/sin(15)})}]
	\node at (1,1){};
	\node at (0,1){};
	\node at (.5,.5){};
	\draw[<-] (0,.15)--(0,.85);
	\draw[->] (1,.15)--(1,.85);
	\draw[->] (.15,1)--(.85,1);
	\draw (.5,0)--(.5,.25);
	\draw (.5,1)--(.5,.75);
	\draw (0,.5)--(.25,.5);
	\draw (1,.5)--(.75,.5);
	\draw[->] ($(1,1)+(30:.15)$)--($(1,1)+(30:.85)$);
	\node at ($(1,1)+(30:1)$){};
	\draw[->] ($(1,1)+(30:1)+(-30:.15)$)--($(1,1)+(30:1)+(-30:.85)$);
	\node at ($(1,1)+(30:1)+(-30:1)$){};
	\draw[->] ($(1,1)+(30:1)+(-30:1)+(-90:.15)$)--($(1,1)+(30:1)+(-30:1)+(-90:.85)$);
	\end{scope}
	
	\begin{scope}[rotate=-90,shift={(105:{0.5/sin(15)})}]
	\node at (1,1){};
	\node at (0,1){};
	\node at (.5,.5){};
	\draw[<-] (0,.15)--(0,.85);
	\draw[->] (1,.15)--(1,.85);
	\draw[->] (.15,1)--(.85,1);
	\draw (.5,0)--(.5,.25);
	\draw (.5,1)--(.5,.75);
	\draw (0,.5)--(.25,.5);
	\draw (1,.5)--(.75,.5);
	\end{scope}
	\end{scope}
	
	\begin{scope}[scale=.4,shift={($(-16,.5)+(150:1)+(120:1)$)}]
	\twelvesixfour[(0,0)][0]
	\twelvesixfour[(0,0)][120]
	\twelvesixfour[($(2,0)+(30:1)+(-30:1)+(60:1)+(-60:1)$)][0]
	\twelvesixfour[($(2,0)+(30:1)+(-30:1)+(60:1)+(-60:1)$)][-60]
	\twelvesixfour[($(1,-1)+(-60:2)+(-120:1)+(-30:1)$)][0]
	\twelvesixfour[($(1,-1)+(-60:2)+(-120:1)+(-30:1)$)][120]
	\twelvesixfour[($(1,-1)+(-60:2)+(-120:1)+(-30:1)$)][-120]
	\end{scope}
	\end{tikzpicture}
	\caption{The 4.6.12 Uniform Tiling and the unit cell of its encoding.}
	\label{fig:4.6.12}
\end{figure}

See \cref{fig:4.6.12} for the lattice structure and the unit cell of the encoding. Loops around hexagonal and dodecagonal faces are non-trivial Paulis and generate the stabilizer, loops around square faces are identity. The qubit to mode ratio is $<1.25$.

This tiling has the same odd/even face pattern as the 6.4.3.4 tiling and so its disparity follows a similar rule in that a 4.6.12 lattice with no holes and fully surrounded dodecagonal faces will have disparity
$$
\Delta=-1.
$$

\subsection{The Kagome Lattice (3.6.3.6 Uniform Tiling)}
So far we have presented Weight-3 Planar encodings whose disparity does not grow with the bulk, this is because the unit cells have as many faces in $K$ as they do faces with no auxilliary qubits. However the Weight-3 planar encoding we have found for the Kagome lattice does not have this property.

\begin{figure}[h]
	\centering
	\begin{tikzpicture}[thick,>=stealth,line cap=round,every node/.style={circle,fill=black,scale=.4}]
	\begin{scope}[scale=.5,shift={(-12,0)}]
	\foreach \x in {0,2,4}{
		\Kagome[(\x,0)][0]
		\Kagome[($(\x,0)+(120:2)$)][0]
		\Kagome[($(\x,0)+(-120:2)$)][0]
	}
	\Kagome[(0,0)][180]
	\Kagome[(6,0)][180]
	\Kagome[($(6,0)+(120:2)$)][0]
	\Kagome[($(6,0)+(-120:2)$)][0]
	\end{scope}
	
	\foreach \x in {0,60,120,180,240,300}{
		\node at (\x:1){};
	}
	\foreach \x in {0,120,240}{
		\draw[<->] ($(\x:1)!0.15!(\x+60:1)$)--($(\x:1)!0.85!(\x+60:1)$);
		\draw[] ($(\x+60:1)!0.15!(\x+120:1)$)--($(\x+60:1)!0.85!(\x+120:1)$);
	}
	\begin{scope}[shift={(1,0)}]
	\node at (60:1) {};
	\node at (-60:1) {};
	\node at (0,{0.5/cos(30)}){};
	\node at (0,{-0.5/cos(30)}){};
	\draw[<->] ($(0,0)!0.15!(60:1)$)--($(0,0)!0.85!(60:1)$);
	\draw[] ($(0,0)!0.15!(-60:1)$)--($(0,0)!0.85!(-60:1)$);
	\draw[<->] ($(60:1)!0.15!(120:1)$)--($(60:1)!0.85!(120:1)$);
	\draw[] ($(-60:1)!0.15!(-120:1)$)--($(-60:1)!0.85!(-120:1)$);
	
	\begin{scope}[shift={((0,{0.5/cos(30)}))}]
	\foreach \x in {-30,90,210}{
		\draw ($(0,0)!0.4!(\x:{0.5*tan(30)})$)--($(0,0)!1!(\x:{0.5*tan(30)})$);
	}
	\end{scope}
	
	\begin{scope}[shift={((0,-{0.5/cos(30)}))}]
	\foreach \x in {30,150,270}{
		\draw ($(0,0)!0.4!(\x:{0.5*tan(30)})$)--($(0,0)!1!(\x:{0.5*tan(30)})$);
	}
	\end{scope}
	\end{scope}
	
	\end{tikzpicture}
	\caption{(Left) A Kagome lattice with two triangular corners. (Right) The unit cell of the encoding showing all possible edge operators and faces.}
	\label{fig:Kagome}
\end{figure}
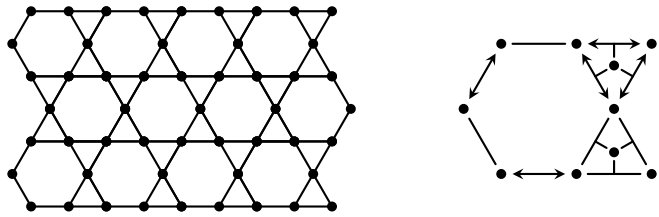
	
See \cref{fig:Kagome} for the lattice structure and the unit cell of the encoding. Loops around hexagonal faces are non-trivial Paulis and generate the stabilizer, loops around triangular faces are identity. The qubit to mode ratio is $<1.67$.

\begin{prop}\label{prop:kagome delta}
	The disparity of the encoding for a connected Kagome lattice of any shape without holes is given by
	\begin{equation}\label{eq:Kagome Delta Formulas}
	\Delta = HF + TC - 2 = \frac{1}{2}(TF+TC-2)
	\end{equation}
	where $HF$ is the number of hexagonal faces, $TF$ is the number of triangular faces and $TC$ is the number of triangular corners, that is, the number of vertices on the lattice boundary which belong only to a triangular face.
\end{prop}

\begin{proof}
Consider a triangular face, the product of any two encoded edge operators around the face gives the third one up to some phase, therefore if an encoding on a lattice includes two edge operators around such a face (e.g. on the lattice boundary), then the third is included in the lattice automatically (illustrated below).
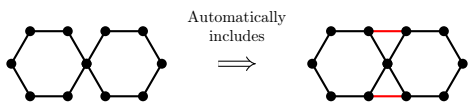
\begin{center}
\begin{tikzpicture}[thick,>=stealth,line cap=round,every node/.style={circle,fill=black,scale=.4},scale=.5]
\Kagome[(0,0)][0]*
\Kagome[(2,0)][0]*

\node at (5,1)[fill=white,rectangle,scale=1.5,align=center]{Automatically \\ includes};
\draw[line cap=butt,thin,-implies,double equal sign distance] (4.5,0) -- (5.5,0);

\begin{scope}[shift={(8,0)}]
\draw[red] (60:1)--($(60:1)+(1,0)$);
\draw[red] (-60:1)--($(-60:1)+(1,0)$);
\Kagome[(0,0)][0]*
\Kagome[(2,0)][0]*
\end{scope}
\end{tikzpicture}
\end{center}
From this we can see that for a lattice of any shape, the triangular faces between hexagonal faces are automatically included in the encoding.

Now consider adding a Hexagonal face to the boundary of an existing lattice with no triangular corners. No matter where it is added, it will create two ``slots''  between itself and other hexagonal faces which induce a triangular face as shown above. All contexts in which a Hexagonal face can be added are shown below to illustrate this (the fifth is captured in the previous diagram), the extra vertices and edges forming the new hexagon are shown in blue and the third edge of the induced triangle face is shown in red.
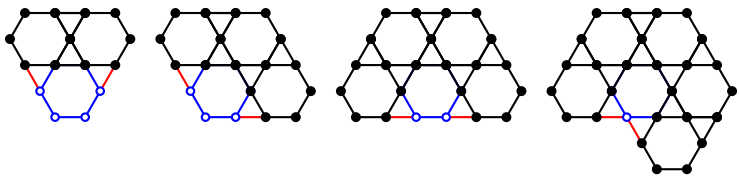
\begin{center}
\begin{tikzpicture}[thick,>=stealth,line cap=round,every node/.style={circle,fill=black,scale=.4},scale=.4]
\begin{scope}[blue,every node/.style={circle,fill=white,draw=blue,scale=.3}]
\draw[red]($(-60:1)+(-120:1)$)--(-120:1);
\draw[red]($(-60:1)+(-120:1)+(2,0)$)--($(-60:1)+(2,0)$);
\Kagome[(-60:2)][0]*
\end{scope}
\Kagome[(0,0)][0]
\Kagome[(2,0)][180]

\begin{scope}[shift={(5,0)}]
\begin{scope}[blue,every node/.style={circle,fill=white,draw=blue,scale=.3}]
\draw[red]($(-60:1)+(-120:1)$)--(-120:1);
\draw[red]($(-60:3)$)--($(-60:3)+(1,0)$);
\Kagome[(-60:2)][0]*
\end{scope}
\Kagome[(0,0)][0]
\Kagome[(2,0)][180]
\Kagome[($(2,0)+(-60:2)$)][120]
\end{scope}

\begin{scope}[shift={(12,0)}]
\begin{scope}[blue,every node/.style={circle,fill=white,draw=blue,scale=.3}]
\draw[red]($(-60:3)$)--($(-60:3)+(1,0)$);
\draw[red]($(-60:3)+(-2,0)$)--($(-60:3)+(-1,0)$);
\Kagome[(-60:2)][0]*
\end{scope}
\Kagome[(0,0)][0]
\Kagome[(2,0)][180]
\Kagome[($(2,0)+(-60:2)$)][120]
\Kagome[(-120:2)][60]
\end{scope}

\begin{scope}[shift={(19,0)}]
\begin{scope}[blue,every node/.style={circle,fill=white,draw=blue,scale=.3}]
\draw[red]($(-60:3)+(-2,0)$)--($(-60:3)+(-1,0)$);
\draw[red]($(-60:3)+(-1,0)$)--($(-60:4)+(-1,0)$);
\Kagome[(-60:2)][0]*
\end{scope}
\Kagome[(0,0)][0]
\Kagome[(2,0)][180]
\Kagome[($(2,0)+(-60:2)$)][120]
\Kagome[(-120:2)][60]
\Kagome[(-60:4)][60]
\end{scope}
\end{tikzpicture}
\end{center}
Any Kagome lattice without triangular corners may be constructed this way, beginning with a single hexagonal face and adding further hexagons one at a time. 

From \Cref{cor:w3disp} the disparity of this encoding is
\begin{equation}
\Delta=TF-HF
\end{equation}
where $TF$ is the number of triangular faces and $HF$ is the number of hexagonal faces. A single hexagonal face has $\Delta=-1$. A hexagonal face added to it will change the delta by +1 as it will bring two triangular faces with it as shown earlier. From this it is clear that $\Delta=HF-2$ for a lattice with no triangular corners.

A lattice with triangular corners may be constructed by simply adding triangular faces to the boundary of a lattice with no such corners. Each added triangular corner will increase $\Delta$ by 1 so the general formula is therefore $\Delta=HF+TC-2$.

For the equivalent expression in terms of $TF$ rather than $HF$ consider the lattice with no triangular corners once more. When constructed as before it is clear that every Hexagonal face added also adds two triangular faces except for the first one so in this case:
\begin{equation}
TF=2HF-2.
\end{equation}
If triangular corners are added to such a lattice then they will each necessarily add another triangular face so in general the relationship is
\begin{equation}
TF=2HF+TC-2.
\end{equation}
Solving for $HF$ and substituting into the first expression in \cref{eq:Kagome Delta Formulas} yields the second.
\end{proof}

From the above we can see that the disparity of the encoding on a Kagome lattice actually grows with the lattice size and so, by \cref{thm:Delta Species relation}, does the number of distinct single particle species one can simultaneously define. As discussed in \cref{sec:species} the appropriate fusions of these species can produce any Pauli operator on the extra $(\mathbb{C}^2)^{\otimes\Delta}$ space attached to the encoding. To reduce errors, this space can be restricted to a single state by finding a set of commuting Paulis to act as a stabilizer, all that remains is to find a maximal set of particle species and to work out how to produce a stabilizer from them.


As with the square lattice encoding, vertices on which edge operators act with only one type of Pauli operator can serve as injection points for distinct species. These are only found on triangular corners on the lattice boundary and some lattices may not even have these so there must be other ways to define single particle species.

These other injection points take a more involved form. Consider a triangular face, the three edge operators around a triangular face will act on their vertex qubits with Pauli, $X$ and each of them will act on the face qubit with a different Pauli (see (\ref{eq:face_notation})). This argument also applies to triangular faces where edge operators act with $Y$ on vertex qubits, simply substitute $X$ with $Y$ where appropriate. Label these edge operators $E_X$, $E_Y$ and $E_Z$ according to their face qubit support. Also label the vertex qubits $q_X$, $q_Y$ and $q_Z$ according to the edge operator opposite, so $q_X$ is opposite $E_X$ etc. Now define an operator which acts with an $X$ on $q_Z$ and with a $Z$ on the face qubit. This operator will anticommute with the vertex operator on $q_Z$, all edge operators incident on $q_Z$ and will commute with all other edge and vertex operators. By \cref{def:Maj Species} this is a member of a single particle species. Consider a single operator defined similarly on $q_X$ ($X$ on $q_X$ and $X$ on the face qubit), it is also a single particle operator but it anticommutes with the particle on $q_Z$, by \cref{def:Distinct Species} this makes them members of the same species. This is confirmed by seeing that the two are related by the edge operator $E_Y$. See \cref{fig:Triangle injection} for illustration.

\begin{figure}[ht]
	\centering
	\begin{tikzpicture}[thick,>=stealth,line cap=round,scale=1.6]
	\node at (0,0)[circle,fill=black,scale=.4]{};
	\node at (1,0)[circle,fill=black,scale=.4]{};
	\node at (60:1)[circle,fill=black,scale=.4]{};
	\node at (.5,{0.5*tan(30)})[circle,fill=black,scale=.4]{};
	
	\draw ($(0,0)!.25!(1,0)$)--($(0,0)!.75!(1,0)$);
	\node at ($(0,0)!.15!(1,0)$)[scale=.7]{$X$};
	\node at ($(0,0)!.85!(1,0)$)[scale=.7]{$X$};
	\draw ($(0,0)!.25!(60:1)$)--($(0,0)!.75!(60:1)$);
	\node at ($(0,0)!.15!(60:1)$)[scale=.7]{$X$};
	\node at ($(0,0)!.85!(60:1)$)[scale=.7]{$X$};
	\draw ($(1,0)!.25!(60:1)$)--($(1,0)!.75!(60:1)$);
	\node at ($(1,0)!.15!(60:1)$)[scale=.7]{$X$};
	\node at ($(1,0)!.85!(60:1)$)[scale=.7]{$X$};

	\begin{scope}[shift={(.5,{0.5*tan(30)})}]
	\draw ($(0,0)!.8!(30:{0.5*tan(30)})$)--(30:{0.5*tan(30)});
	\draw ($(0,0)!.8!(150:{0.5*tan(30)})$)--(150:{0.5*tan(30)});
	\draw ($(0,0)!.8!(-90:{0.5*tan(30)})$)--(-90:{0.5*tan(30)});
	
	\node at ($(0,0)!.5!(30:{0.5*tan(30)})$)[scale=.7]{$X$};
	\node at ($(0,0)!.5!(-90:{0.5*tan(30)})$)[scale=.7]{$Z$};
	\node at ($(0,0)!.5!(150:{0.5*tan(30)})$)[scale=.7]{$Y$};
	
	\node at (30:.55)[scale=.9]{$E_X$};
	\node at (30:-.75)[scale=.9]{$q_X$};
	\node at (-90:.55)[scale=.9]{$E_Z$};
	\node at (-90:.9){$(a)$};
	\node at (-90:-.75)[scale=.9]{$q_Z$};
	\node at (150:.55)[scale=.9]{$E_Y$};
	\node at (150:-.75)[scale=.9]{$q_Y$};
	\end{scope}
	
	\begin{scope}[shift={(2,0)}]
	\node at (0,0)[circle,fill=black,scale=.4]{};
	\node at (1,0)[circle,fill=black,scale=.4]{};
	\node at (60:1)[circle,fill=black,scale=.4]{};
	\node at (.5,{0.5*tan(30)})[circle,fill=black,scale=.4]{};
	
	\node at (60:1)[scale=.8,circle,draw=black,fill=none]{};
	
	\draw ($(0,0)!.15!(1,0)$)--($(0,0)!.85!(1,0)$);
	\draw ($(0,0)!.15!(60:1)$)--($(0,0)!.85!(60:1)$);
	\draw ($(1,0)!.15!(60:1)$)--($(1,0)!.85!(60:1)$);
	
	\begin{scope}[shift={(.5,{0.5*tan(30)})}]
	\draw ($(0,0)!.4!(30:{0.5*tan(30)})$)--(30:{0.5*tan(30)});
	\draw ($(0,0)!.4!(150:{0.5*tan(30)})$)--(150:{0.5*tan(30)});
	\draw ($(0,0)!.4!(-90:{0.5*tan(30)})$)--(-90:{0.5*tan(30)});
	\node at (0,.15)[scale=.7]{$Z$};
	\node at (-90:-.77)[scale=.7]{$X$};
	\end{scope}
	\end{scope}
	
	\begin{scope}[shift={(3.5,0)}]
	\node at (0,0)[circle,fill=black,scale=.4]{};
	\node at (1,0)[circle,fill=black,scale=.4]{};
	\node at (60:1)[circle,fill=black,scale=.4]{};
	\node at (.5,{0.5*tan(30)})[circle,fill=black,scale=.4]{};
	
	\node at (0,0)[scale=.8,circle,draw=black,fill=none]{};
	
	\draw ($(0,0)!.15!(1,0)$)--($(0,0)!.85!(1,0)$);
	\draw ($(0,0)!.15!(60:1)$)--($(0,0)!.85!(60:1)$);
	\draw ($(1,0)!.15!(60:1)$)--($(1,0)!.85!(60:1)$);
	
	\begin{scope}[shift={(.5,{0.5*tan(30)})}]
	\draw ($(0,0)!.4!(30:{0.5*tan(30)})$)--(30:{0.5*tan(30)});
	\draw ($(0,0)!.4!(150:{0.5*tan(30)})$)--(150:{0.5*tan(30)});
	\draw ($(0,0)!.4!(-90:{0.5*tan(30)})$)--(-90:{0.5*tan(30)});
	\node at (0,.15)[scale=.7]{$X$};
	\node at (30:-.77)[scale=.7]{$X$};
	\node at (-90:.9){$(b)$};
	\end{scope}
	\end{scope}
	
	\begin{scope}[shift={(5,0)}]
	\node at (0,0)[circle,fill=black,scale=.4]{};
	\node at (1,0)[circle,fill=black,scale=.4]{};
	\node at (60:1)[circle,fill=black,scale=.4]{};
	\node at (.5,{0.5*tan(30)})[circle,fill=black,scale=.4]{};
	
	\node at (1,0)[scale=.8,circle,draw=black,fill=none]{};
	
	\draw ($(0,0)!.15!(1,0)$)--($(0,0)!.85!(1,0)$);
	\draw ($(0,0)!.15!(60:1)$)--($(0,0)!.85!(60:1)$);
	\draw ($(1,0)!.15!(60:1)$)--($(1,0)!.85!(60:1)$);
	
	\begin{scope}[shift={(.5,{0.5*tan(30)})}]
	\draw ($(0,0)!.4!(30:{0.5*tan(30)})$)--(30:{0.5*tan(30)});
	\draw ($(0,0)!.4!(150:{0.5*tan(30)})$)--(150:{0.5*tan(30)});
	\draw ($(0,0)!.4!(-90:{0.5*tan(30)})$)--(-90:{0.5*tan(30)});
	\node at (0,.15)[scale=.7]{$Y$};
	\node at (150:-.77)[scale=.7]{$X$};
	\end{scope}
	\end{scope}
	\end{tikzpicture}
	\caption{(a) Labelling of edge operators and qubits around a triangular face on the Kagome Lattice encoding. (b) Single particle operators on the circled vertices. These are all of the same species and may be transformed into each other by edge operators. On an equivalent face where the edge operators act on vertices with $Y$, the particle operators also act on vertices with $Y$.}\label{fig:Triangle injection}
\end{figure}
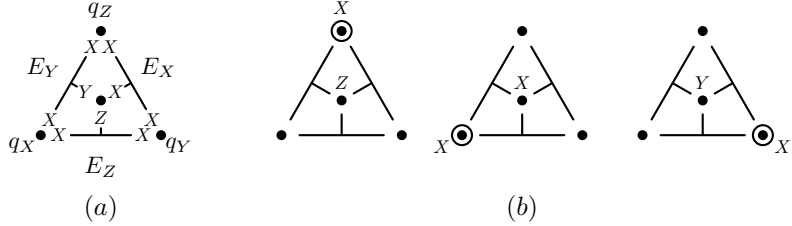

From this we see that, as well as triangular corners, distinct single particle species may be associated with any triangular face and injected at any of its vertices (including triangular faces that form triangular corners). Particles injected at different triangular faces and corners satisfy \cref{def:Distinct Species} as distinct species so we may have a set of simultaneously defined distinct species for every triangular face and triangular corner. Substituting the second expression in \cref{eq:Kagome Delta Formulas} into \cref{eq:Delta Species relation} reveals that this set is maximal and so the species may be used to define any operator on the excess Hilbert space. See \cref{fig:Kagome particles} for examples of single particle operators on a Kagome lattice encoding.

\begin{figure}[ht]
	\centering
	\begin{tikzpicture}[scale=1,thick,line cap=round,>=stealth,faint/.style={opacity=0.2}]
	\foreach \x in {0,60,120,180,240,300}{
		\node at (\x:1)[circle,fill=black,scale=.4]{};
	}
	\foreach \x in {0,120,240}{
		\draw[<->,faint] ($(\x:1)!0.2!(\x+60:1)$)--($(\x:1)!0.8!(\x+60:1)$);
		\draw[faint] ($(\x+60:1)!0.2!(\x+120:1)$)--($(\x+60:1)!0.8!(\x+120:1)$);
	}
	\begin{scope}[shift={(1,0)}]
	\node at (60:1)[circle,fill=black,scale=.4]{};
	\node at (-60:1)[circle,fill=black,scale=.4]{};
	\node at (0,{0.5/cos(30)})[circle,fill=black,scale=.4]{};
	\node at (0,{-0.5/cos(30)})[circle,fill=black,scale=.4]{};
	\draw[<->,faint] ($(60:1)!0.2!(120:1)$)--($(60:1)!0.8!(120:1)$);
	\draw[faint] ($(-60:1)!0.2!(-120:1)$)--($(-60:1)!0.8!(-120:1)$);
	
	\begin{scope}[shift={((0,{0.5/cos(30)}))}]
	\foreach \x in {-30,90,210}{
		\draw[faint] ($(0,0)!0.4!(\x:{0.5*tan(30)})$)--($(0,0)!1!(\x:{0.5*tan(30)})$);
	}
	\end{scope}
	
	\begin{scope}[shift={((0,-{0.5/cos(30)}))}]
	\foreach \x in {30,150,270}{
		\draw[faint] ($(0,0)!0.4!(\x:{0.5*tan(30)})$)--($(0,0)!1!(\x:{0.5*tan(30)})$);
	}
	\end{scope}
	\end{scope}
	
	\begin{scope}[shift={(2,0)}]
	\foreach \x in {0,60,120,180,240,300}{
		\node at (\x:1)[circle,fill=black,scale=.4]{};
	}
	\foreach \x in {0,120,240}{
		\draw[<->,faint] ($(\x:1)!0.2!(\x+60:1)$)--($(\x:1)!0.8!(\x+60:1)$);
		\draw[faint] ($(\x+60:1)!0.2!(\x+120:1)$)--($(\x+60:1)!0.8!(\x+120:1)$);
	}
	\begin{scope}[shift={(1,0)}]
	\node at (60:1)[circle,fill=black,scale=.4]{};
	\node at (-60:1)[circle,fill=black,scale=.4]{};
	\node at (0,{0.5/cos(30)})[circle,fill=black,scale=.4]{};
	\node at (0,{-0.5/cos(30)})[circle,fill=black,scale=.4]{};
	\draw[<->,faint] ($(0,0)!0.2!(60:1)$)--($(0,0)!0.8!(60:1)$);
	\draw[faint] ($(0,0)!0.2!(-60:1)$)--($(0,0)!0.8!(-60:1)$);
	\draw[<->,faint] ($(60:1)!0.2!(120:1)$)--($(60:1)!0.8!(120:1)$);
	\draw[faint] ($(-60:1)!0.2!(-120:1)$)--($(-60:1)!0.8!(-120:1)$);
	
	\begin{scope}[shift={((0,{0.5/cos(30)}))}]
	\foreach \x in {-30,90,210}{
		\draw[faint] ($(0,0)!0.4!(\x:{0.5*tan(30)})$)--($(0,0)!1!(\x:{0.5*tan(30)})$);
	}
	\end{scope}
	
	\begin{scope}[shift={((0,-{0.5/cos(30)}))}]
	\foreach \x in {30,150,270}{
		\draw[faint] ($(0,0)!0.4!(\x:{0.5*tan(30)})$)--($(0,0)!1!(\x:{0.5*tan(30)})$);
	}
	\draw[cbORANGE] ($(0,0)!0.4!(270:{0.5*tan(30)})$)--($(0,0)!1!(270:{0.5*tan(30)})$);
	\end{scope}
	\end{scope}
	\end{scope}
	
	\node at (-1,0)[circle,draw=cbBLUE,fill=none,scale=.8]{};
	\draw[cbBLUE] ($(-1,0)!0.2!(-120:1)$)--($(-1,0)!0.8!(-120:1)$);
	\draw[cbBLUE,<->] ($(-120:1)!0.2!(-60:1)$)--($(-120:1)!0.8!(-60:1)$);
	\node at (-1.3,0)[cbBLUE,scale=.7]{$X$};
	\node at ($(0,-.2)+(-120:1)$)[cbBLUE,scale=.7]{$Z$};
	\node at ($(0,-.2)+(-60:1)$)[cbBLUE,scale=.7]{$Z$};
	\node at ($(-60:1)+(30:0.38)$)[cbBLUE,scale=.7]{$X$};
	
	\node at ($(1,0)+(-60:1)$)[circle,draw=cbORANGE,fill=none,scale=.8]{};
	\draw[cbORANGE,<->] ($(-120:1)!0.2!(-60:1)+(2,0)$)--($(-120:1)!0.8!(-60:1)+(2,0)$);
	\draw[cbORANGE] ($(-120:1)!0.2!(-60:1)+(3,0)$)--($(-120:1)!0.8!(-60:1)+(3,0)$);
	\node at ($(1,-.3)+(-60:1)$)[cbORANGE,scale=.7]{$Y$};
	\node at ($(2,-.2)+(-60:1)$)[cbORANGE,scale=.7]{$Z$};
	\node at ($(3,-.2)+(-60:1)$)[cbORANGE,scale=.7]{$Z$};
	\node at ($(2.5,.45)+(-60:1)$)[cbORANGE,scale=.7]{$Z$};
	
	\node at ($(1,0)+(60:1)$)[circle,draw=cbGREEN,fill=none,scale=.8]{};
	\node at ($(1,.3)+(60:1)$)[cbGREEN,scale=.7]{$Y$};
	\node at ($(1,0)+(60:1)+(-150:.38)$)[cbGREEN,scale=.7]{$Y$};
	
	\node at ($(1,0)+(60:1)+(2,0)$)[circle,draw,fill=none,scale=.8]{};
	\node at ($(3,.3)+(60:1)$)[scale=.7]{\textbf{\textit{Y}}};
	
	\end{tikzpicture}
	\caption{Single particle operators on a Kagome lattice at the circled vertices. \textbf{Bold black:} A single particle operator at a triangular corner. \textcolor{cbORANGE}{Orange:} A single particle operator transported from a triangular corner by edge operators. \textcolor{cbGREEN}{Green:} A single particle operator defined at a triangular face as in \cref{fig:Triangle injection}. \textcolor{cbBLUE}{Blue:} A single particle operator transported from a triangular face via edge operators.}
	\label{fig:Kagome particles}
\end{figure}
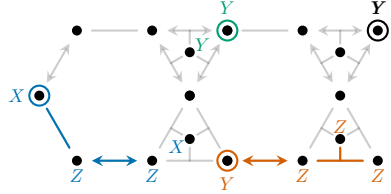

With a maximal set of single particle species we can now define stabilizers to restrict the excess Hilbert space of the encoding. To do this, simply group $2\Delta$ of the particle species into $\Delta$ disjoint pairs and fuse a each pair on any vertex to get $\Delta$ commuting Pauli operators. It is preferable to pair up particles with injection sites that are close to each other, that way their fusion will have a low Pauli weight, see \cref{fig:Kagome pairings} for illustration. As only $2\Delta$ species are involved in the pairings there are 2 left, these are then chosen to represent the single Majorana operators and the hole operators on the remaining encoded fermionic space.

\begin{figure}[ht]
	\centering
	\begin{tikzpicture}[scale=1,thick,line cap=round,>=stealth,faint/.style={opacity=0.2}]
	\foreach \x in {0,60,120,180,240,300}{
		\node at (\x:1)[circle,fill=black,scale=.4]{};
	}
	\foreach \x in {0,120,240}{
		\draw[<->,faint] ($(\x:1)!0.2!(\x+60:1)$)--($(\x:1)!0.8!(\x+60:1)$);
		\draw[faint] ($(\x+60:1)!0.2!(\x+120:1)$)--($(\x+60:1)!0.8!(\x+120:1)$);
	}
	\begin{scope}[shift={(1,0)}]
	\node at (60:1)[circle,fill=black,scale=.4]{};
	\node at (-60:1)[circle,fill=black,scale=.4]{};
	\node at (0,{0.5/cos(30)})[circle,fill=black,scale=.4]{};
	\node at (0,{-0.5/cos(30)})[circle,fill=black,scale=.4]{};
	\draw[<->,faint] ($(60:1)!0.2!(120:1)$)--($(60:1)!0.8!(120:1)$);
	\draw[faint] ($(-60:1)!0.2!(-120:1)$)--($(-60:1)!0.8!(-120:1)$);
	
	\begin{scope}[shift={((0,{0.5/cos(30)}))}]
	\foreach \x in {-30,90,210}{
		\draw[faint] ($(0,0)!0.4!(\x:{0.5*tan(30)})$)--($(0,0)!1!(\x:{0.5*tan(30)})$);
	}
	\end{scope}
	
	\begin{scope}[shift={((0,-{0.5/cos(30)}))}]
	\foreach \x in {30,150,270}{
		\draw[faint] ($(0,0)!0.4!(\x:{0.5*tan(30)})$)--($(0,0)!1!(\x:{0.5*tan(30)})$);
	}
	\end{scope}
	\end{scope}
	
	\begin{scope}[shift={(2,0)}]
	\foreach \x in {0,60,120,180,240,300}{
		\node at (\x:1)[circle,fill=black,scale=.4]{};
	}
	\foreach \x in {0,120,240}{
		\draw[<->,faint] ($(\x:1)!0.2!(\x+60:1)$)--($(\x:1)!0.8!(\x+60:1)$);
		\draw[faint] ($(\x+60:1)!0.2!(\x+120:1)$)--($(\x+60:1)!0.8!(\x+120:1)$);
	}
	\begin{scope}[shift={(1,0)}]
	\node at (60:1)[circle,fill=black,scale=.4]{};
	\node at (-60:1)[circle,fill=black,scale=.4]{};
	\node at (0,{0.5/cos(30)})[circle,fill=black,scale=.4]{};
	\node at (0,{-0.5/cos(30)})[circle,fill=black,scale=.4]{};
	\draw[<->,faint] ($(0,0)!0.2!(60:1)$)--($(0,0)!0.8!(60:1)$);
	\draw[faint] ($(0,0)!0.2!(-60:1)$)--($(0,0)!0.8!(-60:1)$);
	\draw[<->,faint] ($(60:1)!0.2!(120:1)$)--($(60:1)!0.8!(120:1)$);
	\draw[faint] ($(-60:1)!0.2!(-120:1)$)--($(-60:1)!0.8!(-120:1)$);
	
	\begin{scope}[shift={((0,{0.5/cos(30)}))}]
	\foreach \x in {-30,90,210}{
		\draw[faint] ($(0,0)!0.4!(\x:{0.5*tan(30)})$)--($(0,0)!1!(\x:{0.5*tan(30)})$);
	}
	\end{scope}
	
	\begin{scope}[shift={((0,-{0.5/cos(30)}))}]
	\foreach \x in {30,150,270}{
		\draw[faint] ($(0,0)!0.4!(\x:{0.5*tan(30)})$)--($(0,0)!1!(\x:{0.5*tan(30)})$);
	}
	\end{scope}
	\end{scope}
	\end{scope}
	
	\begin{scope}[shift={(-2,0)}]
	\foreach \x in {0,60,120,180,240,300}{
		\node at (\x:1)[circle,fill=black,scale=.4]{};
	}
	\foreach \x in {0,120,240}{
		\draw[<->,faint] ($(\x:1)!0.2!(\x+60:1)$)--($(\x:1)!0.8!(\x+60:1)$);
		\draw[faint] ($(\x+60:1)!0.2!(\x+120:1)$)--($(\x+60:1)!0.8!(\x+120:1)$);
	}
	\begin{scope}[shift={(1,0)}]
	\node at (60:1)[circle,fill=black,scale=.4]{};
	\node at (-60:1)[circle,fill=black,scale=.4]{};
	\node at (0,{0.5/cos(30)})[circle,fill=black,scale=.4]{};
	\node at (0,{-0.5/cos(30)})[circle,fill=black,scale=.4]{};
	\draw[<->,faint] ($(60:1)!0.2!(120:1)$)--($(60:1)!0.8!(120:1)$);
	\draw[faint] ($(-60:1)!0.2!(-120:1)$)--($(-60:1)!0.8!(-120:1)$);
	
	\begin{scope}[shift={((0,{0.5/cos(30)}))}]
	\foreach \x in {-30,90,210}{
		\draw[faint] ($(0,0)!0.4!(\x:{0.5*tan(30)})$)--($(0,0)!1!(\x:{0.5*tan(30)})$);
	}
	\end{scope}
	
	\begin{scope}[shift={((0,-{0.5/cos(30)}))}]
	\foreach \x in {30,150,270}{
		\draw[faint] ($(0,0)!0.4!(\x:{0.5*tan(30)})$)--($(0,0)!1!(\x:{0.5*tan(30)})$);
	}
	\end{scope}
	\end{scope}
	\end{scope}
	
	\node at (1,0)[circle,draw=cbBLUE,fill=none,scale=.8]{};
	\node at (1.3,0)[cbBLUE,scale=.7]{$Z$};
	\node at ($(.5,.45)+(-60:1)$)[cbBLUE,scale=.7]{$Z$};
	\node at ($(.5,-.48)+(60:1)$)[cbBLUE,scale=.7]{$Z$};
	
	\node at ($(3,0)+(-60:1)$)[circle,draw=cbORANGE,fill=none,scale=.8]{};
	\node at ($(3,-.3)+(-60:1)$)[cbORANGE,scale=.7]{$Z$};
	\node at ($(2.5,.45)+(-60:1)$)[cbORANGE,scale=.7]{$Y$};
	
	\node at ($(3,0)+(60:1)$)[circle,draw=cbGREEN,fill=none,scale=.8]{};
	\node at ($(3,.3)+(60:1)$)[cbGREEN,scale=.7]{$Z$};
	\node at ($(2.5,-.48)+(60:1)$)[cbGREEN,scale=.7]{$Y$};
	
	\end{tikzpicture}
	\caption{Possible stabilizers to restrict the excess Hilert space on a Kagome lattice encoding, formed by fusing particle species on the circled vertices. \textcolor{cbBLUE}{Blue:} Fusion of two species injected at adjacent triangular faces. \textcolor{cbGREEN}{Green:} Fusion of one species injected at a triangular corner and one injected at the same triangular face. \textcolor{cbORANGE}{Orange:} Same as Green. The particle species injected at the remaining triangular faces are then the Majorana and hole operators.}
	\label{fig:Kagome pairings}
\end{figure}
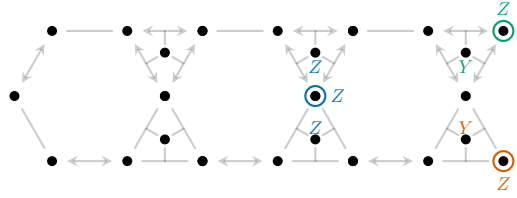



\subsection{The 3.12.12 Uniform Tiling}

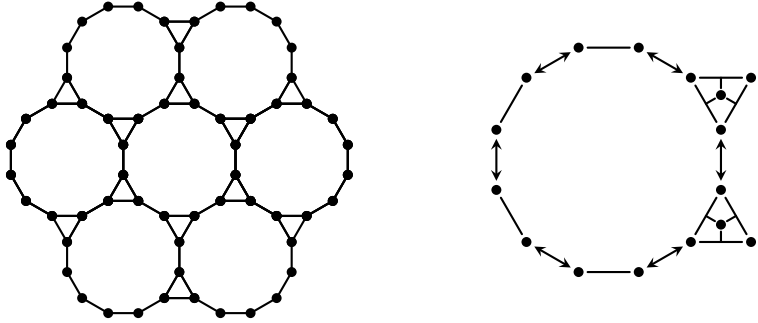
\begin{figure}[ht]
	\centering
	\begin{tikzpicture}[thick,>=stealth,line cap=round,every node/.style={circle,fill=black,scale=.4}]
	\begin{scope}[scale=.8]
	\foreach \x in {15,45,75,105,135,165,195,225,255,285,315,345}{
		\node at (\x:{0.5/sin(15)}){};
	}
	\foreach \x in {15,75,135,195,255,315}{
		\draw[] ($(\x:{0.5/sin(15)})!0.15!(\x+30:{0.5/sin(15)})$)--($(\x:{0.5/sin(15)})!0.85!(\x+30:{0.5/sin(15)})$);
		\draw[<->] ($(\x+30:{0.5/sin(15)})!0.15!(\x+60:{0.5/sin(15)})$)--($(\x+30:{0.5/sin(15)})!0.85!(\x+60:{0.5/sin(15)})$);
	}
	
	\begin{scope}[shift={(15:{0.5/sin(15)})}]
	\node at (60:1){};
	\node at (0,{0.5/cos(30)}) {};
	\draw[] ($(0,0)!0.15!(60:1)$)--($(0,0)!0.85!(60:1)$);
	\draw[] ($(60:1)!0.15!(120:1)$)--($(60:1)!0.85!(120:1)$);
	
	\begin{scope}[shift={((0,{0.5/cos(30)}))}]
	\foreach \x in {-30,90,210}{
		\draw ($(0,0)!0.4!(\x:{0.5*tan(30)})$)--($(0,0)!1!(\x:{0.5*tan(30)})$);
	}
	\end{scope}
	\end{scope}
	
	\begin{scope}[shift={(-15:{0.5/sin(15)})}]
	\node at (-60:1){};
	\node at (0,{-0.5/cos(30)}) {};
	\draw[] ($(0,0)!0.15!(-60:1)$)--($(0,0)!0.85!(-60:1)$);
	\draw[] ($(-60:1)!0.15!(-120:1)$)--($(-60:1)!0.85!(-120:1)$);
	
	\begin{scope}[shift={((0,-{0.5/cos(30)}))}]
	\foreach \x in {30,150,270}{
		\draw ($(0,0)!0.4!(\x:{0.5*tan(30)})$)--($(0,0)!1!(\x:{0.5*tan(30)})$);
	}
	\end{scope}
	\end{scope}
	\end{scope}
	
	\begin{scope}[scale=.4,shift={(-18,0)}]
	\twelvethree[(0,0)][0]
	\twelvethree[(0,0)][60]
	\twelvethree[(0,0)][-60]
	\twelvethree[($(2,0)+(30:1)+(-30:1)$)][0]
	\twelvethree[($(4,0)+(30:2)+(-30:2)$)][180]
	\twelvethree[($(4,0)+(30:2)+(-30:2)$)][120]
	\twelvethree[($(4,0)+(30:2)+(-30:2)$)][240]
	
	\twelvethree[($(1,2)+(120:1)+(60:1)+(-30:1)$)][0]
	\twelvethree[($(3,2)+(120:1)+(60:1)+(-30:2)+(30:1)$)][180]
	\twelvethree[($(1,-2)+(-120:1)+(-60:1)+(30:1)$)][0]
	\twelvethree[($(3,-2)+(-120:1)+(-60:1)+(30:2)+(-30:1)$)][180]
	\end{scope}
	\end{tikzpicture}
	\caption{The 3.12.12 Uniform Tiling and the unit cell of its encoding.}
	\label{fig:3.12.12}
\end{figure}

See \cref{fig:3.12.12} for the lattice structure and the unit cell of the encoding. Loops around dodecagonal faces are non-trivial Paulis and generate the stabilizer, loops around triangular faces are identity. The qubit to mode ratio is $<1.34$.

This tiling has the same even/odd face pattern as the Kagome lattice, with dodecagons instead of hexagons, accordingly the disparity of this encoding is
\begin{equation}
\Delta=DF + TC -2
\end{equation}
where $DF$ is the number of dodecagonal faces and $TC$ is the number of triangular corners. The argument in the proof for \cref{prop:kagome delta} is easily modified for this case.

\section{A Cubic Encoding}\label{sec:cubic}

We now introduce a specific 3d fermionic encoding on a cubic lattice. This encoding is a generalization of the 2d square encoding reviewed in Section \ref{LW encoding} in the following sense. We wish to define an orientation for each edge in the lattice as in the 2d case, with the condition that any 2d slice of the lattice has an edge orientation identical to a 2d encoding. There are a number of ways to do this.

Consider the $(0,0,0)$ corner of the cubic lattice. This corner has three edges extending away from it in the $\hat{x}$, $\hat{y}$ and  $\hat{z}$ direction. Specifying the orientation of those three edges completely specifies the orientation of all other edges, given the condition that any 2d slice looks identical to the edge orientation of a 2d encoding. Since we do not care about the case where all orientations are inverted, there are four different ways to orientate these three edges. In one case, all arrows point away from the corner, and in the three other cases, two arrows point away from the corner, and a third arrow points into the corner. This is illustrated in \cref{fig:fourCubes}. Thus there are four possible ways to give an orientation to the edges of the lattice.

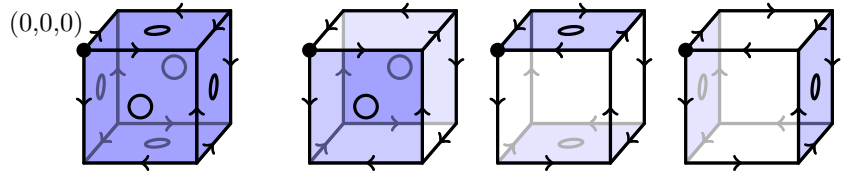
\begin{figure}[ht]
\begin{center}
\begin{tikzpicture}[scale=1.5]

\rotateReset{1.5}
\oddface[0][0]* 
\rotateSide{0}{0}{1.5}
\oddface*[-1][0]*
\rotateFlat{0}{0}{1.5}
\oddface*[0][-1]*
\rotateReset{1.5}
\oddface*[-.3][-.35]
\rotateSide{1}{0}{1.5}
\oddface[-1][0]
\rotateFlat{0}{1}{1.5}
\oddface[-1][-1]

\pgftransformreset

\begin{scope}[shift={(3,0)}]
\rotateReset{1.5}
\oddface[0][0]* 
\rotateSide{0}{0}{1.5}
\evenface*[-1][0]*
\rotateFlat{0}{0}{1.5}
\evenface[0][-1]*
\rotateReset{1.5}
\oddface*[-.3][-.35]
\rotateSide{1}{0}{1.5}
\evenface[-1][0]
\rotateFlat{0}{1}{1.5}
\evenface*[-1][-1]
\end{scope}

\begin{scope}[shift={(5.5,0)}]
\rotateReset{1.5}
\evenface[0][0]* 
\rotateSide{0}{0}{1.5}
\evenface*[-1][0]*
\rotateFlat{0}{0}{1.5}
\oddface[0][-1]*
\rotateReset{1.5}
\evenface*[-.3][-.35]
\rotateSide{1}{0}{1.5}
\evenface[-1][0]
\rotateFlat{0}{1}{1.5}
\oddface*[-1][-1]
\end{scope}

\begin{scope}[shift={(8,0)}]
\rotateReset{1.5}
\evenface[0][0]* 
\rotateSide{0}{0}{1.5}
\oddface*[-1][0]*
\rotateFlat{0}{0}{1.5}
\evenface*[0][-1]*
\rotateReset{1.5}
\evenface*[-.3][-.35]
\rotateSide{1}{0}{1.5}
\oddface[-1][0]
\rotateFlat{0}{1}{1.5}
\evenface[-1][-1]
\end{scope}

\node at (-.45,.975)[circle,fill=black,scale=0.6]{};
\node at (-.45-.5,.975+.35)[]{(0,0,0)};
\node at (3-.45,.975)[circle,fill=black,scale=0.6]{};
\node at (5.5-.45,.975)[circle,fill=black,scale=0.6]{};
\node at (8-.45,.975)[circle,fill=black,scale=0.6]{};
\end{tikzpicture}
\end{center}

\caption{Four possible orientations to the edges of the cubic lattice. Odd faces are coloured blue and have a circle in the center to denote the extra qubit. The rightmost cell is odd and the remaining three are even.}\label{fig:fourCubes}
\end{figure}

In the case of the 2d encoding, the orientations of the edges are in accordance with a checkerboard labelling of the faces, so that an even face has oriented edges circulating around it, and an odd face has all edges around it touching head to head or tail to tail. A similar checkerboard labelling is induced by ones choice of edge orientation for the corner $(0,0,0)$, with all three faces surrounding that corner being odd in the case where all edges point away from the corner, and with only one face being odd in the other three choices of edge orientation. 

Regardless of ones choice of orientation, a given cubic cell in the lattice will either have all six faces odd, which we call an \textit{odd cell}, or exactly two opposite faces odd, which we will call an \textit{even cell} (see \cref{fig:fourCubes}). There are no other possibilities. This is most easily seen by checking the four possible edge orientations of the $(0,0,0)$ corner in the case of a $2\times 2\times 2$ cubic lattice. The cells will be oriented in the cubic lattice such that every odd cell shares its faces only with even cells, and every even cell shares its odd faces only with odd cells, as illustrated in \cref{fig:cells}. Thus every odd face can either be associated with a unique odd cell, or else it is a face on the boundary of the cubic lattice. We call those odd faces on the boundary of the cubic lattice that are not also faces of an odd cell \textit{isolated odd faces}.

\begin{figure}
\centering
\begin{tikzpicture}
\rotateReset{1.5}
\oddface[0][0]* 
\rotateSide{0}{0}{1.5}
\oddface*[-1][0]*
\rotateFlat{0}{0}{1.5}
\oddface*[0][-1]*
\rotateReset{1.5}
\oddface*[-.3][-.35]*
\rotateSide{1}{0}{1.5}
\oddface[-1][0]*
\rotateFlat{0}{1}{1.5}
\oddface[-1][-1]*

\rotateReset*{1.5}
\begin{scope}[shift={(1,0)}]
\evenface*[0][0]* 
\rotateSide{0}{0}{1.5}
\oddface[-1][0]*
\rotateFlat{0}{0}{1.5}
\evenface[0][-1]*
\rotateReset{1.5}
\evenface[-.3][-.35]*
\rotateSide{1}{0}{1.5}
\oddface*[-1][0]
\rotateFlat{0}{1}{1.5}
\evenface*[-1][-1]*
\end{scope}

\rotateReset*{1.5}
\begin{scope}[shift={(-.3,-.35)}]
\oddface*[0][0]* 
\rotateSide{0}{0}{1.5}
\evenface[-1][0]*
\rotateFlat{0}{0}{1.5}
\evenface*[0][-1]*
\rotateReset{1.5}
\oddface[-.3][-.35]
\rotateSide{1}{0}{1.5}
\evenface*[-1][0]*
\rotateFlat{0}{1}{1.5}
\evenface[-1][-1]*
\end{scope}

\rotateReset*{1.5}
\begin{scope}[shift={(.7,-.35)}]
\evenface[0][0]* 
\rotateSide{0}{0}{1.5}
\evenface*[-1][0]*
\rotateFlat{0}{0}{1.5}
\oddface[0][-1]*
\rotateReset{1.5}
\evenface*[-.3][-.35]
\rotateSide{1}{0}{1.5}
\evenface[-1][0]
\rotateFlat{0}{1}{1.5}
\oddface*[-1][-1]*
\end{scope}

\rotateReset*{1.5}
\begin{scope}[shift={(0,1)}]
\evenface[0][0]* 
\rotateSide{0}{0}{1.5}
\evenface*[-1][0]*
\rotateFlat{0}{0}{1.5}
\oddface[0][-1]*
\rotateReset{1.5}
\evenface*[-.3][-.35]*
\rotateSide{1}{0}{1.5}
\evenface[-1][0]*
\rotateFlat{0}{1}{1.5}
\oddface*[-1][-1]
\end{scope}

\rotateReset*{1.5}
\begin{scope}[shift={(1,1)}]
\oddface*[0][0]* 
\rotateSide{0}{0}{1.5}
\evenface[-1][0]*
\rotateFlat{0}{0}{1.5}
\evenface*[0][-1]*
\rotateReset{1.5}
\oddface[-.3][-.35]*
\rotateSide{1}{0}{1.5}
\evenface*[-1][0]
\rotateFlat{0}{1}{1.5}
\evenface[-1][-1]
\end{scope}

\rotateReset*{1.5}
\begin{scope}[shift={(-.3,.65)}]
\evenface*[0][0]* 
\rotateSide{0}{0}{1.5}
\oddface[-1][0]*
\rotateFlat{0}{0}{1.5}
\evenface[0][-1]*
\rotateReset{1.5}
\evenface[-.3][-.35]
\rotateSide{1}{0}{1.5}
\oddface*[-1][0]*
\rotateFlat{0}{1}{1.5}
\evenface*[-1][-1]
\end{scope}

\rotateReset{1.5}
\begin{scope}[shift={(.7,.65)}]
\oddface[0][0]* 
\rotateSide{0}{0}{1.5}
\oddface*[-1][0]*
\rotateFlat{0}{0}{1.5}
\oddface*[0][-1]*
\rotateReset{1.5}
\oddface*[-.3][-.35]
\rotateSide{1}{0}{1.5}
\oddface[-1][0]
\rotateFlat{0}{1}{1.5}
\oddface[-1][-1]
\end{scope}

\end{tikzpicture}
\caption{The unit cell of the encoding which includes two odd cells (front top right, back bottom left) and six even cells in different orientations.}\label{fig:cells}
\end{figure}
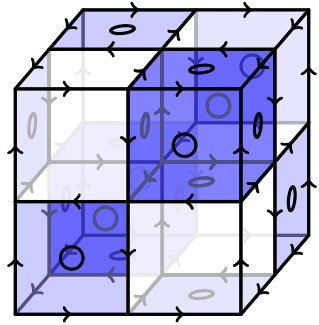

Having specified an orientation of the lattice, and an even/odd labelling of all of the faces and cells, we may now describe the encoding. In fact the encoding is almost exactly the same as in the 2d case. We associate a qubit with every vertex of the lattice, and a qubit with every odd face of the lattice. The vertex operator at vertex $i$ is given by $\tilde{V}_i = Z_i$, and the edge operator for edge $(i,j)$, with $i$ pointing to $j$ is given by $\tilde{E}_{ij}:= X_i Y_j P(i,j)$ ($\tilde{E}_{ji}:=-\tilde{E}_{ij}$) where $P(i,j)$ is the same Pauli operator $P$ on every face qubit adjacent to $(i,j)$, where if the edge $(i,j)$ is aligned to the $\hat{x}$ ($\hat{y}$, $\hat{z}$) direction, then $P =X$ ($Y$, $Z$) respectively. This ensures that any two edges, going in orthogonal directions, which share an adjacent odd face, will act on that face qubit with different Pauli operators. Finally, for every \emph{isolated odd face}, we must choose a single adjacent edge $(i,j)$, with $i$ pointing to $j$, and modify the definition of the operator for that edge to be $\tilde{E}_{ij}:= -X_i Y_j P(i,j)$. This ensures that the product of the edges around that face is equal to $1$ and not $-1$. \cref{fig:cubicOperators} illustrates the operators of this encoding.

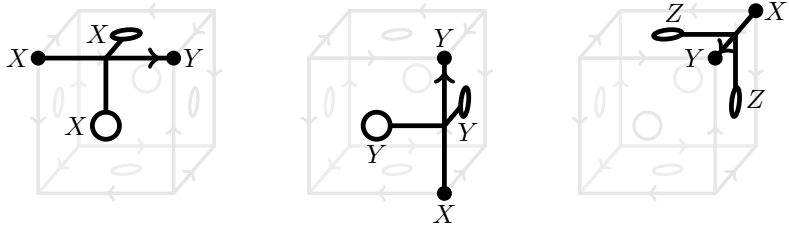
\begin{figure}
\centering
\begin{tikzpicture}[scale=1.8,]
\begin{scope}[white!90!black]
\rotateReset{1.8}
\whiteoddface[0][0] 
\rotateSide{0}{0}{1.8}
\whiteoddface*[-1][0]
\rotateFlat{0}{0}{1.8}
\whiteoddface*[0][-1]
\rotateReset{1.8}
\whiteoddface*[-.3][-.35]
\rotateSide{1}{0}{1.8}
\whiteoddface[-1][0]
\rotateFlat{0}{1}{1.8}
\whiteoddface[-1][-1]
\end{scope}

\begin{scope}[line width=1.8pt,
line cap=round,
decoration={
	markings,
	mark=at position .9 with {\arrow{>}}}]
\draw[postaction=decorate] (-.3,.65)--(.7,.65);
\draw (.2,.25)--(.2,.65)--(.32,.79);
\node at (-.3,.65)[circle,fill=black,scale=0.6]{};
\node at (-.45,.65){$X$};
\node at (.7,.65)[circle,fill=black,scale=0.6]{};
\node at (.85,.65){$Y$};
\node at (-.02,.15){$X$};
\draw[] (.2,.15) circle (.1cm);
\rotateFlat{0}{1}{1.8}
\draw[] (.5,-.5) circle (.1cm);
\node at (.29,-.5){$X$};
\end{scope}

\begin{scope}[shift={(2,0)}]
\begin{scope}[white!90!black]
\rotateReset{1.8}
\whiteoddface[0][0] 
\rotateSide{0}{0}{1.8}
\whiteoddface*[-1][0]
\rotateFlat{0}{0}{1.8}
\whiteoddface*[0][-1]
\rotateReset{1.8}
\whiteoddface*[-.3][-.35]
\rotateSide{1}{0}{1.8}
\whiteoddface[-1][0]
\rotateFlat{0}{1}{1.8}
\whiteoddface[-1][-1]
\end{scope}

\rotateReset{1.8}
\begin{scope}[line width=1.8pt,
line cap=round,
decoration={
	markings,
	mark=at position .9 with {\arrow{>}}}]
\draw[postaction=decorate] (.7,-.35)--(.7,.65);
\draw (.3,.15)--(.7,.15)--(.82,.29);
\node at (.7,-.35)[circle,fill=black,scale=0.6]{};
\node at (.7,-.5){$X$};
\node at (.7,.65)[circle,fill=black,scale=0.6]{};
\node at (.7,.8){$Y$};
\draw[] (.2,.15) circle (.1cm);
\node at (.2,-.06){$Y$};
\node at (.875,.325-.22){$Y$};
\rotateSide{1}{0}{1.8}
\draw[] (-.5,.5) circle (.1cm);

\end{scope}

\end{scope}

\begin{scope}[shift={(4,0)}]
\begin{scope}[white!90!black]
\rotateReset{1.8}
\whiteoddface[0][0] 
\rotateSide{0}{0}{1.8}
\whiteoddface*[-1][0]
\rotateFlat{0}{0}{1.8}
\whiteoddface*[0][-1]
\rotateReset{1.8}
\whiteoddface*[-.3][-.35]
\rotateSide{1}{0}{1.8}
\whiteoddface[-1][0]
\rotateFlat{0}{1}{1.8}
\whiteoddface[-1][-1]
\end{scope}

\rotateReset{1.8}
\begin{scope}[line width=1.8pt,
line cap=round,
decoration={
	markings,
	mark=at position .9 with {\arrow{>}}}]
\draw[postaction=decorate] (1,1)--(.7,.65);
\node at (1.15,1){$X$};
\node at (.55,.65){$Y$};
\node at (.4,.98){$Z$};
\node at (1,.35){$Z$};
\draw (.45,.825)--(.85,.825)--(.85,.425);
\node at (1,1)[circle,fill=black,scale=0.6]{};
\node at (.7,.65)[circle,fill=black,scale=0.6]{};
\rotateFlat{0}{1}{1.8}
\draw[] (.5,-.5) circle (.1cm);
\rotateSide{1}{0}{1.8}
\draw[] (-.5,-.5) circle (.1cm);
\end{scope}
\end{scope}

\end{tikzpicture}
\caption{Edge operators of the cubic encoding along edges aligned in the $x$, $y$ and $z$ directions (shown from left to right). If an edge is part of an isolated odd face then it will only act on one face qubit, if it is only part of two even faces then it will only act on vertex qubits.}\label{fig:cubicOperators}
\end{figure}
The maximum weight of the edge operators in this encoding is $4$. In the bulk the number of qubits per mode is $2.5$.

The disparity of the encoding, $\Delta$, is more complicated to compute than for the planar cases. In particular, $D(K)$ is not simply equal to the number of odd faces, as in the square lattice encoding.

\begin{figure}
	\centering
	\begin{tikzpicture}
	\begin{scope}[white!90!black]
	\rotateReset{1.8}
	\whiteoddface[0][0] 
	\rotateSide{0}{0}{1.8}
	\whiteoddface*[-1][0]
	\rotateFlat{0}{0}{1.8}
	\whiteoddface*[0][-1]
	\rotateReset{1.8}
	\whiteoddface*[-.3][-.35]
	\rotateSide{1}{0}{1.8}
	\whiteoddface[-1][0]
	\rotateFlat{0}{1}{1.8}
	\whiteoddface[-1][-1]
	\end{scope}
	
	\rotateReset{1.8}
	\begin{scope}[line width=1.8pt,
	line cap=round,
	decoration={
		markings,
		mark=at position .9 with {\arrow{>}}}]
	\draw[postaction=decorate] (.7,-.35)--(.7,.65);
	\draw (.7,.15)--(.82,.29);
	\node at (.875,.325-.22){$Y$};
	\rotateSide{1}{0}{1.8}
	\draw[] (-.5,.5) circle (.1cm);
	
	\rotateReset*{1.8}
	\begin{scope}[shift={(-1,0)}]
	\draw[postaction=decorate] (.7,.65)--(.7,-.35);
	\draw (.7,.15)--(.82,.29);
	\node at (.875,.325-.22){$Y$};
	\rotateSide{1}{0}{1.8}
	\draw[] (-.5,.5) circle (.1cm);
	\end{scope}
	
	\rotateReset*{1.8}
	\draw[postaction=decorate] (-.3,.65)--(.7,.65);
	\draw (.2,.65)--(.32,.79);
	\rotateFlat{0}{1}{1.8}
	\draw[] (.5,-.5) circle (.1cm);
	\node at (.29,-.5){$X$};
	
	\rotateReset*{1.8}
	\begin{scope}[shift={(0,-1)}]
	\draw[postaction=decorate] (.7,.65)--(-.3,.65);
	\draw (.2,.65)--(.32,.79);
	\rotateFlat{0}{1}{1.8}
	\draw[] (.5,-.5) circle (.1cm);
	\node at (.29,-.5){$X$};
	\end{scope}
	\end{scope}
	
	\node at (1.5-.15,.5-.175)[]{=};
	\node at (1.5-.15,-.7)[]{$(a)$};
	\node at (4.5-.15,-.7)[]{$(b)$};
	
	\begin{scope}[shift={(2,0)}]
	\begin{scope}[white!90!black]
	\rotateReset{1.8}
	\whiteoddface[0][0] 
	\rotateSide{0}{0}{1.8}
	\whiteoddface*[-1][0]
	\rotateFlat{0}{0}{1.8}
	\whiteoddface*[0][-1]
	\rotateReset{1.8}
	\whiteoddface*[-.3][-.35]
	\rotateSide{1}{0}{1.8}
	\whiteoddface[-1][0]
	\rotateFlat{0}{1}{1.8}
	\whiteoddface[-1][-1]
	\end{scope}
	
	\rotateReset{1.8}
	\begin{scope}[line width=1.8pt,
	line cap=round,
	decoration={
		markings,
		mark=at position .9 with {\arrow{>}}}]
	\draw[postaction=decorate] (1,1)--(1,0);
	\draw (.88,.36)--(1,.5);
	\node at (.875,.325+.22){$Y$};
	\rotateSide{1}{0}{1.8}
	\draw[] (-.5,.5) circle (.1cm);

	\rotateReset{1.8}
	\begin{scope}[shift={(-2,0)}]
	\draw[postaction=decorate] (1,0)--(1,1);
	\draw (.88,.36)--(1,.5);
	\node at (.875,.325+.22){$Y$};
	\rotateSide{1}{0}{1.8}
	\draw[] (-.5,.5) circle (.1cm);
	\end{scope}

	\rotateReset{1.8}
	\begin{scope}[shift={(-1,0)}]
	\draw[postaction=decorate] (1,1)--(0,1);
	\draw (.38,.86)--(.5,1);
	\rotateFlat{0}{1}{1.8}
	\draw[] (.5,-.5) circle (.1cm);
	\node at (.71,-.5){$X$};
	\end{scope}
	
	\rotateReset{1.8}
	\begin{scope}[shift={(-1,-1)}]
	\draw[postaction=decorate] (0,1)--(1,1);
	\draw (.38,.86)--(.5,1);
	\rotateFlat{0}{1}{1.8}
	\draw[] (.5,-.5) circle (.1cm);
	\node at (.71,-.5){$X$};
	\end{scope}
	
	\draw[very thick,white!90!black] (-1.3,.65)--(-.3,.65);
	\draw[very thick,white!90!black] (-.3,-.35)--(-.3,.65);
	
	\begin{scope}[shift={(1.5-.15,-.175)}]
	\draw[postaction=decorate] (0,0)--(0,1);
	\draw[postaction=decorate] (0,1)--(1,1);
	\draw[postaction=decorate] (1,1)--(1,0);
	\draw[postaction=decorate] (1,0)--(0,0);
	\node at (1,1)[circle,fill=black,scale=0.6]{};
	\node at (1.15,1){$Z$};
	\node at (0,1)[circle,fill=black,scale=0.6]{};
	\node at (0,1.15){$Z$};
	\node at (1,0)[circle,fill=black,scale=0.6]{};
	\node at (1,-.15){$Z$};
	\node at (0,0)[circle,fill=black,scale=0.6]{};
	\node at (-.15,0){$Z$};
	\draw (-.5,.5) circle (.1cm);
	\draw (0,.5)--(-.4,.5);
	\node at (-.5,.72){$Y$};
	\draw (1.5,.5) circle (.1cm);
	\draw (1,.5)--(1.4,.5);
	\node at (1.5,.72){$Y$};
	\draw (.5,1.5) circle (.1cm);
	\draw (.5,1)--(.5,1.4);
	\node at (.72,1.5){$X$};
	\draw (.5,-.5) circle (.1cm);
	\draw (.5,0)--(.5,-.4);
	\node at (.72,-.5){$X$};
	
	\rotateSide{0}{0}{1.8}
	\draw[] (-.5,.5) circle (.1cm);
	\draw (0,.5)--(-.4,.5);
	\node at (-1,.65){$Y$};
	\rotateSide{1}{0}{1.8}
	\draw[] (-.5,.5) circle (.1cm);
	\draw (0,.5)--(-.4,.5);
	\node at (-.5,.72){$Y$};
	\rotateFlat{-1}{0}{1.8}
	\draw[] (.5,.5) circle (.1cm);
	\draw (.5,0)--(.5,.4);
	\node at (.28,.5){$X$};
	\rotateFlat{0}{1}{1.8}
	\draw[] (.5,.5) circle (.1cm);
	\draw (.5,0)--(.5,.4);
	\node at (.72,.5){$X$};
	\end{scope}
	\end{scope}
	\end{scope}

	\end{tikzpicture}
	\caption{Loop operators around odd and even faces in the cubic encoding. (a) shows the identical loop operators around odd faces opposite to each other on an odd cell. Loop operators around odd cells aligned in different directions will have a similar form only with different Paulis. Loop operators around isolated odd faces are identity. (b) shows a loop operator around an even face. Similarly, operators oriented in different directions will have the same shape but different Paulis. Loop operators around even faces on the lattice boundary will have ``hanging'' Paulis omitted.}\label{fig:pauliFaceLoop}
\end{figure}
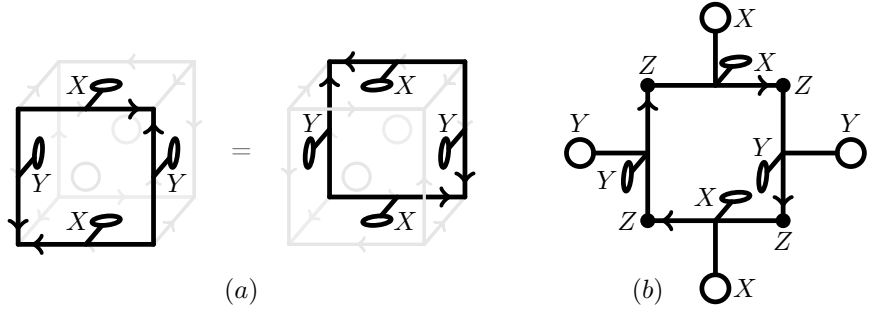

\begin{lemma}
\begin{equation}\label{eq:rankOfKernel}
D(K) = IOF + 2OC.
\end{equation}
Where $IOF$ denotes the number of isolated odd faces, and $OC$ denote the number of odd cells.
\end{lemma}

\begin{proof}

Since $K$ is an abelian finitely generated group, its rank $D(K)$ is the size of a maximal independent subset of elements.

Note that for cycles $a$ around isolated odd faces $\sigma(a) = 1$, and so $a \in K$. Furthermore, for a cycle $a$ around an odd face that is bounding an odd cell, the cycle $b$ on the opposite face of that cell satisfies $\sigma(a)=\sigma(b)$ (see Figure \ref{fig:pauliFaceLoop}), and so the cycle $ab$ around that pair of faces is in $K$.

Consider a subset $g \subseteq K$ consisting of the cycles around:
\begin{itemize}
\item each isolated odd face.
\item one pair of opposite faces for every odd cell.
\item a second different pair of opposite faces for every odd cell
\end{itemize}
We claim that $K = \langle g \rangle$, and that no element of $g$ can be generated by any other elements of $g$. Thus $D(K)=|g| = IOF+2OC$. 

First we argue that every element of $g$ is independent. This can be seen straightforwardly by noting that every odd face only shares edges with odd faces bounding a common odd cell. Thus all the isolated odd faces are independent. Furthermore each cycle $ab$ around a pair of opposite odd faces in $g$ only shares edges with the other pair of opposite odd faces $cd \in g$, and $ab \neq cd$. Thus all elements are independent.

Secondly we argue that $K$ is generated by $g$, via proof by contradiction. Assume there exists an element $a\in K$ that is not in $\langle g \rangle$. We first note that there exists a set of odd face cycles $F$ such that 
\begin{equation}
a = \prod_{b\in F} b
\end{equation}
This can be seen by noting that if a cycle $a$ is in $K$ then for every edge $e_1$ in $a$ pointing into (away from) a vertex $v$ there exists another edge $e_2$ in $a$ which also points into (away from resp.) vertex $v$. Furthermore by inspection $e_1$ and $e_2$ must bound a unique common odd cell. Thus $a$ may be decomposed into a product of cycles, where each of these cycles are confined to the edges bounding a unique odd cell. Such cycles may be generated by the odd face cycles bounding the odd cell. 

For any odd cell $c$, we may define $F_c= \{b | b \in F \textrm{ and $b$ bounds $c$} \}$, and the corresponding product
\begin{align*}
a_c := \prod_{b \in F_c} b .
\end{align*}
W.l.o.g. we may assume to have chosen an $a$ s.t. no $b \in F$ is an isolated odd face. Thus $a = \prod_c a_c$. We note that any two odd face cycles $b$ and $b'$ that bound different odd cells have representations $\sigma(b)$ and $\sigma(b')$ that act on disjoint qubits. It must follow that $a_c \in K$. It can be seen by enumerating cases that all possible forms of $a_c$, ie all possible combinations of products of face cycles bounding a given odd cell whose product yields an element of $K$, may be generated by elements in $g$. Thus $a \in \langle g \rangle$, which is a contradiction.

\end{proof}

Given this expression for $D(K)$, the disparity $\Delta$ may now be determined from counting arguments.

\begin{theorem}
Let an odd corner vertex of a cubic lattice be a corner vertex of the lattice whose associated corner cell is odd. Given a cubic lattice encoding as defined above, the disparity is given by
$$\Delta = \frac{OCV}{2}-1 \in \{-1,0,1,3\} $$ 
Where $OCV$ denotes the number of odd corner vertices.
\end{theorem}

\begin{proof}
Let $V$ denote the number of vertices, $E$ the number of edges, $F$ the number of faces, $C$ the number of cells, $OF/EF$ the number of odd/even faces, $OC/EC$ the number of odd/even cells, $IOF$ the number of isolated odd faces, $N$ the number of qubits, and $M$ the number of fermionic modes.

Note that for this encoding
\begin{equation}
N-M = OF.
\end{equation}
Note also that
\begin{equation}
IOF = OF - 6OC.
\end{equation}
Note also that the Euler characteristic of a cubic lattice is given by 
\begin{equation}
V-E+F-C=1.
\end{equation}
This can be seen most easily by counting the edges vertices faces and cells of a cubic lattice with a single cell, and noting that the Euler characteristic is an invariant of the lattice size, since it is a topological invariant. Thus, using Equation \ref{eq:circuitRank} we retrieve
\begin{equation}
D(\mathcal{C}_G) = F-C
\end{equation}
Substituting these expressions, and \cref{eq:rankOfKernel}, into \cref{eq:disparity2} gives
\begin{equation}
\Delta=OF-EF-3OC+EC
\end{equation}
We need only count odd/even faces and cells. For a given vertex in the lattice, we define:
\begin{equation}
\Delta_v = OF_v/4-EF_v/4 -3OC_v/8+EC_v/8
\end{equation}
where $OF_v, EF_v, OC_v, EC_v$ are the number of odd faces, even faces, odd cells and even cells, respectively, containing the vertex $v$. $\Delta_v$ counts the number of odd/even faces and cells per vertex, with the factors of $1/4$ or $1/8$ corresponding to the number of other vertices sharing respectively a face or a cell. Thus
\begin{equation}
\Delta = \sum_v \Delta_v
\end{equation}
We now proceed by computing $\Delta_v$ for vertices in the bulk ($deg(v)=6$), on the face ($deg(v)=5$), on an edge ($deg(v)=4$), and on a corner ($deg(v)=3$) of the cubic lattice, by inspection of the vertex neighbourhoods illustrated in \cref{fig:vertexNeighbourhoods}.
For a vertex $v$ in the bulk of the lattice
\begin{equation}
\Delta_v = 6/4- 6/4 - 3*2/8 + 6/8 =0 .
\end{equation}
For a vertex $v$ on the face of the lattice
\begin{equation}
\Delta_v = 4/4 - 4/4 -3*1/8+ 3/8 = 0.
\end{equation}
For a vertex $v$ on the edge of the lattice, either
\begin{equation}
\Delta_v = 2/4 - 3/4 -3*0/8+ 2/8 = 0,
\end{equation}
or
\begin{equation}
\Delta_v = 3/4 - 2/4 -3*1/8+ 1/8 = 0.
\end{equation}
Finally, for a vertex $v$ on the corner of the lattice, either 
\begin{equation}
\Delta_v = 3/4-0/4-3*1/8+0/8 = 3/8
\end{equation}
if $v$ is adjacent to an odd cell (an odd corner vertex), or
\begin{equation}
\Delta_v = 1/4-2/4 -3*0/8+1/8 = -1/8
\end{equation}
if $v$ is adjacent to an even cell (an even corner vertex).
\begin{figure}[ht]
\begin{center}
	\begin{tikzpicture}[scale=1.5]
	\begin{scope}[shift={(3,0)}]
	\rotateSide{0}{0}{1.5}
	\evenface*[0][0]*
	\oddface*[0][1]*
	\rotateFlat{0}{1}{1.5}
	\oddface*[0][0]*
	\evenface[-1][0]*
	\rotateReset{1.5}
	\oddface[-1][-1]
	\oddface*[0][0]
	\evenface*[0][-1]
	\evenface[-1][0]
	
	\begin{scope}[very thick,decoration={
		markings,
		mark=at position 0.5 with {\arrow{>}}}]
	\draw[postaction={decorate}] (0,1)--(.3,1.35);
	\draw[postaction={decorate}] (1,0)--(1.3,.35);
	\draw[]  (.3,1.35)--(.3,1);
	\draw[]  (1.3,.35)--(1,.35);
	\end{scope}
	
	\node at (0,0)[circle,fill=black,scale=0.6]{};
	\end{scope}
	
	\rotateSide{0}{0}{1.5}
	\evenface*[0][0]*
	\oddface*[0][1]*
	
	\rotateFlat{0}{1}{1.5}
	\oddface*[0][0]*
	\evenface[-1][0]*
	
	\rotateReset{1.5}
	\oddface[-1][-1]*
	\oddface*[0][0]
	\evenface*[0][-1]*
	\evenface[-1][0]*
	
	\rotateSide{0}{0}{1.5}
	\oddface[-1][-1]*
	
	\rotateFlat{0}{0}{1.5}
	\oddface[-1][-1]*
	
	\rotateSide{0}{0}{1.5}
	\evenface[-1][0]
	\draw[very thick] (-.5,-.5) circle (.1cm);
	
	\rotateFlat{0}{0}{1.5}
	\evenface*[0][-1]
	\draw[very thick] (-.5,-.5) circle (.1cm);
	
	\rotateReset{1.5}
	\draw[very thick] (-.5,-.5) circle (.1cm);
	\begin{scope}[very thick,decoration={
		markings,
		mark=at position 0.5 with {\arrow{>}}}]
	\draw[postaction={decorate}] (0,1)--(.3,1.35);
	\draw[postaction={decorate}] (1,0)--(1.3,.35);
	\draw[postaction={decorate}] (-1,0)--(-1,1);
	\draw[postaction={decorate}] (-1,1)--(0,1);
	\draw[postaction={decorate}] (0,-1)--(1,-1);
	\draw[postaction={decorate}] (1,-1)--(1,0);
	\draw[postaction={decorate}] (-.3,-1.35)--(-.3,-.35);
	\draw[postaction={decorate}] (-.3,-1.35)--(0,-1);
	\draw[postaction={decorate}] (-1.3,-.35)--(-.3,-.35);
	\draw[postaction={decorate}] (-1.3,-.35)--(-1,0);
	\draw[]  (.3,1.35)--(.3,1);
	\draw[]  (1.3,.35)--(1,.35);
	\end{scope}
	
	\begin{scope}[very thick,decoration={
		markings,
		mark=at position 2/7 with {\arrow{>}}}]
	\draw[postaction={decorate}] (-.3,0)--(-1,0);
	\end{scope}
	
	\begin{scope}[very thick,decoration={
		markings,
		mark=at position 5/7 with {\arrow{>}}}]
	\draw[postaction={decorate}] (-1,-1)--(-.3,-1);
	\end{scope}
	
	\begin{scope}[very thick,decoration={
		markings,
		mark=at position 5/6.5 with {\arrow{>}}}]
	\draw[postaction={decorate}] (-1,-1)--(-1,-.35);
	\end{scope}
	
	\begin{scope}[very thick,decoration={
		markings,
		mark=at position 1.5/6.5 with {\arrow{>}}}]
	\draw[postaction={decorate}] (0,-.35)--(0,-1);
	\end{scope}
	
	\node at (0,0)[circle,fill=black,scale=0.6]{};
	
	\rotateReset{1.5}
	
	\begin{scope}[shift={(0,-3)}]
	\rotateSide{0}{0}{1.5}
	\oddface*[-1][0]*
	\rotateReset{1.5}
	\oddface*[-.3][-.35]
	\evenface[-1.3][-.35]
	\rotateFlat{0}{1}{1.5}
	\oddface[0][-1]
	\evenface*[-1][-1]
	\node at (0,-1)[circle,fill=black,scale=0.6]{};
	
	\rotateReset{1.5}
	\end{scope}
	
	\begin{scope}[shift={(3,-3)}]
	\rotateSide{0}{0}{1.5}
	\evenface*[-1][0]*
	\rotateReset{1.5}
	\evenface*[-.3][-.35]
	\oddface[-1.3][-.35]
	\rotateFlat{0}{1}{1.5}
	\oddface*[0][-1]
	\evenface[-1][-1]
	\node at (0,-1)[circle,fill=black,scale=0.6]{};

	\end{scope}
	
	\rotateReset{1.5}
	
	\begin{scope}[shift={(5.5,-.65)}]
	\rotateSide{0}{0}{1.5}
	\oddface*[-1][0]*
	\rotateReset{1.5}
	\oddface*[-.3][-.35]
	\rotateFlat{0}{1}{1.5}
	\oddface[0][-1]
	\node at (0,-1)[circle,fill=black,scale=0.6]{};
	\end{scope}
	
	\rotateReset{1.5}
	
	\begin{scope}[shift={(5.5,-3)}]
	\rotateSide{0}{0}{1.5}
	\evenface*[-1][0]*
	\rotateReset{1.5}
	\evenface*[-.3][-.35]
	\rotateFlat{0}{1}{1.5}
	\oddface*[0][-1]
	\node at (0,-1)[circle,fill=black,scale=0.6]{};
	\end{scope}
	
	\end{tikzpicture}
\end{center}
\caption{The possible neighbourhoods of vertices in the bulk, on a face, on a edge, and on a corner of the cubic lattice. }\label{fig:vertexNeighbourhoods}
\end{figure}
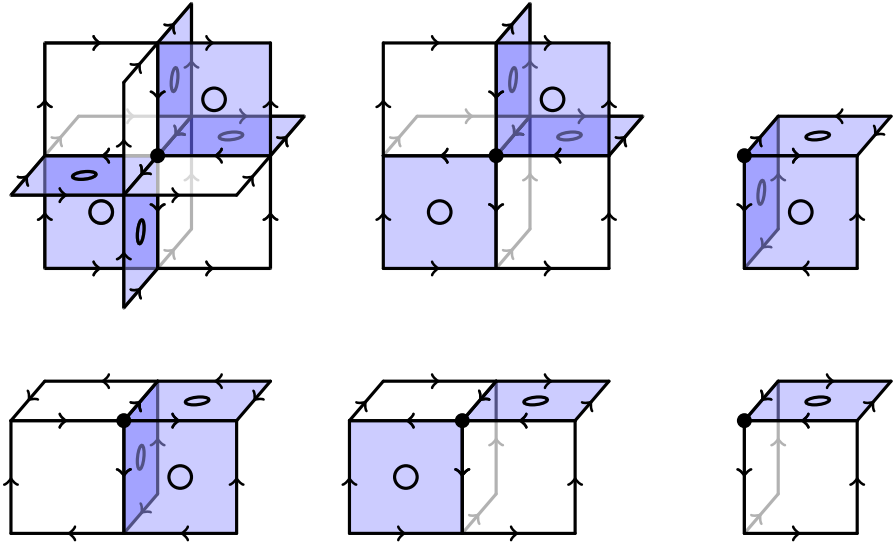
Thus, since there are $8$ corner vertices, which are either even corner vertices or odd corner vertices 
\begin{equation}
\Delta = \frac{3}{8} OCV -\frac{1}{8}(8-OCV) =\frac{1}{2}OCV -1.
\end{equation}
All that remains is to show that $\Delta \in \{-1,0,1,3\}$.

First note that an odd cell only shares faces with even cells, an even cell only shares odd faces with odd cells, and the face opposite an odd face on an even cell is also an odd face. Therefore, if we consider the three columns of cells extending away from a given odd cell in the three cardinal directions, each column will consist of alternating even and odd cells. Thus, \textit{given an odd corner cell, a corner opposite to it in one of the three cardinal directions will be odd if the side length of the lattice in that direction is even, and even if the side length is odd.}

Second note that any cell sharing an even face with an even cell must also be an even cell. Furthermore, the face opposite an even face on an even cell is also an even face. Therefore, if we consider the three columns of cells extending away from a given even cell in the three cardinal directions, only the cells in one of those columns contain odd cells, and the other two columns must only contain even cells. Thus, \textit{for a given even corner cell, of the three corner cells opposite this corner cell along the cardinal directions, at least two of them must be even corner cells, and the other is odd if and only if the side length of the lattice in that direction is odd.}

If a lattice has no odd corners, then the disparity is $\Delta =-1$. If on the other hand a lattice has at least one odd corner vertex, then the observations above completely fix the even or oddness of the remaining corner vertices based solely on the even/oddness of the side lengths of the lattice. This is illustrated in \cref{fig:evenOddCubic}. In these cases the number of corner vertices is $8$, $4$ or $2$, which yields disparities of $3$, $1$ and $0$.





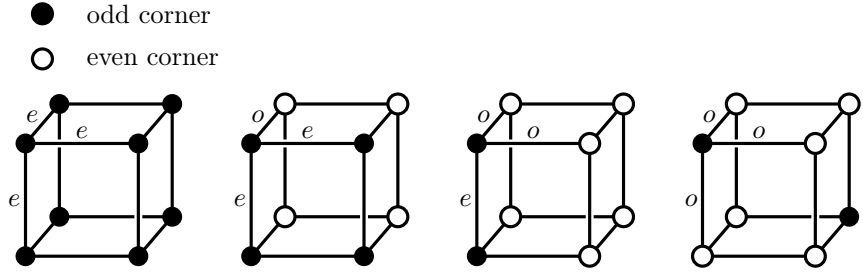
\begin{figure}[ht]
	\centering
	\begin{tikzpicture}[very thick,scale=1.5,line cap=round]
	\node at (0-.15,1.8)[circle,fill=black,scale=.8,draw]{};
	\node at (0-.15,1.4)[circle,fill=white,scale=.8,draw]{};
	
	\node at (0-.15+.3,1.8)[text width=5cm,anchor=west]{odd corner};
	\node at (0-.15+.3,1.4)[text width=5cm,anchor=west]{even corner};
	
	\draw (0,0) rectangle (1,1);
	\draw (0,0)--(-.3,-.35);
	\draw (1,0)--(1-.3,-.35);
	\draw (0,1)--(-.3,1-.35);
	\draw (1,1)--(1-.3,1-.35);
	\node at (0,0)[circle,fill=black,scale=.8]{};
	\node at (1,0)[circle,fill=black,scale=.8]{};
	\node at (0,1)[circle,fill=black,scale=.8]{};
	\node at (1,1)[circle,fill=black,scale=.8]{};
	
	\node at (-.24,1-.11){$e$};
	\node at (-.4,.15){$e$};
	\node at (.5-.3,1-.35+.1){$e$};

	\begin{scope}[shift={(-.3,-.35)}]
	\draw[white,line width=3pt] (0,0) rectangle (1,1);
	\draw (0,0) rectangle (1,1);
	\node at (0,0)[circle,fill=black,scale=.8]{};
	\node at (1,0)[circle,fill=black,scale=.8]{};
	\node at (0,1)[circle,fill=black,scale=.8]{};
	\node at (1,1)[circle,fill=black,scale=.8]{};
	\end{scope}
	
	\begin{scope}[shift={(2,0)}]
	\draw (0,0) rectangle (1,1);
	\draw (0,0)--(-.3,-.35);
	\draw (1,0)--(1-.3,-.35);
	\draw (0,1)--(-.3,1-.35);
	\draw (1,1)--(1-.3,1-.35);
	\node at (0,0)[circle,fill=white,scale=.8,draw]{};
	\node at (1,0)[circle,fill=white,scale=.8,draw]{};
	\node at (0,1)[circle,fill=white,scale=.8,draw]{};
	\node at (1,1)[circle,fill=white,scale=.8,draw]{};
	
	\node at (-.24,1-.11){$o$};
	\node at (-.4,.15){$e$};
	\node at (.5-.3,1-.35+.1){$e$};
	
	\begin{scope}[shift={(-.3,-.35)}]
	\draw[white,line width=3pt] (0,0) rectangle (1,1);
	\draw (0,0) rectangle (1,1);
	\node at (0,0)[circle,fill=black,scale=.8]{};
	\node at (1,0)[circle,fill=black,scale=.8]{};
	\node at (0,1)[circle,fill=black,scale=.8]{};
	\node at (1,1)[circle,fill=black,scale=.8]{};
	\end{scope}
	\end{scope}
	
	\begin{scope}[shift={(4,0)}]
	\draw (0,0) rectangle (1,1);
	\draw (0,0)--(-.3,-.35);
	\draw (1,0)--(1-.3,-.35);
	\draw (0,1)--(-.3,1-.35);
	\draw (1,1)--(1-.3,1-.35);
	\node at (0,0)[circle,fill=white,scale=.8,draw]{};
	\node at (0,1)[circle,fill=white,scale=.8,draw]{};
	\node at (1,0)[circle,fill=white,scale=.8,draw]{};
	\node at (1,1)[circle,fill=white,scale=.8,draw]{};
	
	\node at (-.24,1-.11){$o$};
	\node at (-.4,.15){$e$};
	\node at (.5-.3,1-.35+.1){$o$};
	
	\begin{scope}[shift={(-.3,-.35)}]
	\draw[white,line width=3pt] (0,0) rectangle (1,1);
	\draw (0,0) rectangle (1,1);
	\node at (0,0)[circle,fill=black,scale=.8]{};
	\node at (1,0)[circle,fill=white,scale=.8,draw]{};
	\node at (0,1)[circle,fill=black,scale=.8]{};
	\node at (1,1)[circle,fill=white,scale=.8,draw]{};
	\end{scope}
	\end{scope}
	
	\begin{scope}[shift={(6,0)}]
	\draw (0,0) rectangle (1,1);
	\draw (0,0)--(-.3,-.35);
	\draw (1,0)--(1-.3,-.35);
	\draw (0,1)--(-.3,1-.35);
	\draw (1,1)--(1-.3,1-.35);
	\node at (0,0)[circle,fill=white,scale=.8,draw]{};
	\node at (1,0)[circle,fill=black,scale=.8]{};
	\node at (0,1)[circle,fill=white,scale=.8,draw]{};
	\node at (1,1)[circle,fill=white,scale=.8,draw]{};
	
	\node at (-.24,1-.11){$o$};
	\node at (-.4,.15){$o$};
	\node at (.5-.3,1-.35+.1){$o$};
	
	\begin{scope}[shift={(-.3,-.35)}]
	\draw[white,line width=3pt] (0,0) rectangle (1,1);
	\draw (0,0) rectangle (1,1);
	\node at (0,0)[circle,fill=white,scale=.8,draw]{};
	\node at (1,0)[circle,fill=white,scale=.8,draw]{};
	\node at (0,1)[circle,fill=black,scale=.8]{};
	\node at (1,1)[circle,fill=white,scale=.8,draw]{};
	\end{scope}
	\end{scope}
	
	\end{tikzpicture}
	\caption{The possible configurations of even and odd corners on the cubic encoding, assuming there is at least one odd corner. The letters denote whether a lattice edge has an even or odd number of vertices.}\label{fig:evenOddCubic}
\end{figure}
\end{proof}

One notable feature of this construction is that for a lattice with infinite extent in any direction, there are no corner vertices, and so the disparity is $0$, corresponding to an encoding representing the full fermionic Hilbert space. The 2d square encoding is formally a subcase of the cubic encoding outlined here, and so this also applies to the 2d square encoding.

Since every edge is adjacent to at most $2$ odd faces, the maximum Pauli weight of an edge operator is $4$. In the bulk the number of odd faces per qubit is $6/4 = 1.5$, so the mode to qubit ratio is less than $1.5$.

\subsection{Majorana Species for \texorpdfstring{$\Delta=1$}{Lg} and \texorpdfstring{$\Delta=3$}{Lg}}
If all of the edge operators in a given fermionic encoding act with the same Pauli operator on a vertex qubit, then there exists a weight-1 Pauli operator acting on that vertex qubit which anti-commutes with all fermionic edge and vertex operators on that site. This Pauli operator is a logical Majorana operator. In the square lattice encoding, as well as the cubic encoding, this occurs at odd corner vertices. Majoranas may be ``injected" into the lattice by applying a Majorana operator at a corner vertex, followed by a string of edge operators, translating the Majorana to the desired location.

In the case where the disparity $\Delta=-1$ there are no odd corner vertices where a Majorana may be injected. This corresponds to the fact that only the even fermionic Hilbert space is encoded in this case. When $\Delta=0$, there are two odd corner vertices where a Majorana may be injected. If one takes by convention one odd corner vertex to correspond to the injection of a Majorana, then the other odd corner vertex corresponds to the injection of a Majorana hole, which anti-commutes with the Majorana particle. In this sense both corners correspond to the injection of the same species of Majorana, since they are equivalent up to stabilizers and the action of logical even fermionic operators. 

If the disparity $\Delta =1$, then the lattice has four odd corner vertices, and the additional qubit degree of freedom manifests as four species of mutually anti-commuting Majorana particles, each species distinguished by which of the four odd corner vertices the particle was originally "injected" into the lattice. These particles correspond to logical operations acting jointly on both the encoded fermionic Hilbert space, and the additional qubit degree of freedom. Specifically, one particle can be taken to be the canonical Majorana operator $\gamma$, and the other three particles then correspond to operators acting jointly on the logical qubit degree of freedom and the fermionic Hilbert space as a tensor product of a Pauli operator and a Majorana hole operator:
$$ h \otimes Z \;,\; h \otimes X \;,\; h \otimes Y$$
Importantly, the mutual anti-commutation relations are preserved in this logical representation. This feature appears in both the 2d square lattice encoding, as well as the cubic encoding, when $\Delta =1$.

Unlike the square lattice encoding, the cubic lattice encoding may also have a disparity of $\Delta =3$, corresponding to $8$ odd corner vertices, and $3$ additional qubit degrees of freedom. In this case, there are $8$ species of mutually anti-commuting Majorana particles which may be injected at the different corners. It is no coincidence that the size of the largest set of mutually anti-commuting Pauli operators on 3 qubits is $7$, and so the $8$ Majorana particles may be chosen to correspond to the following $8$ mutually anti-commuting operators on the joint fermionic-qubit logical space:

\begin{align}
\gamma \otimes III, \;\; h \otimes ZXI, \;\; h \otimes ZYI, \;\; h \otimes XIZ, \\
h \otimes YIZ, \;\; h \otimes IZX, \;\; h \otimes IZY, \;\; h \otimes ZZZ \nonumber
\end{align}

Note that all logical Pauli operators acting exclusively on the three qubit space can be generated by even products of the 7 operators associated with holes, for example:
\begin{equation}
I \otimes XII = (h \otimes XIZ)(h \otimes ZXI) (h \otimes ZYI) (h \otimes ZZZ).
\end{equation}
Thus all logical Paulis correspond to pairings of different species of Majorana particles.



\section{Discussion}\label{sec:discussion}

We have illustrated how the compact encoding may be applied to many other lattices, including all uniform tilings of degree less than $4$, and the cubic lattice. Most importantly we have illustrated how to analyse the code spaces of these encodings. Our hope is that these examples may be useful to others in tailoring the compact encoding to their own specific needs.

It would be valuable to have a general procedure for constructing these encodings. In particular a general procedure for choosing the orientations of the edges of the graphs, and assignment of auxilliary qubits, so that the weights of the edge operators are minimal and act locally. The constructions we present here were not too difficult to find, but ultimately emerged from a process of trial and error. We expect the problem can be straightforwardly framed as a combinatoric optimization problem, however this suggests that a generic set of instructions for optimal constructions is unlikely, especially given that ultimately the optimization problem will be dependent on the particular notion of locality furnished by the quantum computing device.


As in \cite{Derby20}, the encodings presented here have a low code distance with a single $Z$ physical error on any vertex qubit inducing an undetectable logical error. In \cite{bausch2020} the error mitigating properties of the compact encoding on the square lattice were discussed. We expect that many of these properties also apply to the encodings presented here.

In \cite{Derby20} it was shown how the compact encoding may be thought of as a method for condensing the particle excitations of the toric code into the low energy subspace. In this case removing the vertex qubits from the stabilizer generators reveals toric code stabilizers on the auxilliary qubits. We would like to understand better how this feature translates to the 3d cubic encoding. The structure of the encoding bears some resemblance to a 3d generalization of the toric code presented in \cite{Chamon2005, BRAVYI2011}, which has qubits living on the faces of cubes, and employs weight $6$ stabilizers on each cube. The cubic encoding also has its auxilliary qubits living on the faces of cubes, however these are more sparsely arranged, and the simplest stabilizers around faces, when considering only auxilliary qubits, have weights $4$ and $8$ (see Figure \ref{fig:pauliFaceLoop}). We have as yet been unable to identify a code in the literature which has stabilizers obviously resembling these \cite{dennis2002topological,levin2003fermions, haah2011local, webster2020fault}. This motivates further study of the structure and properties of this underlying code.

\section*{Acknowldegments}

We would like to thank Tom Scruby and Michael Vasmer for helpful discussions on 3D toric codes, Johannes Bausch for his input on spinful planar encodings and Toby Cubitt for his diligent proof-reading.

\printbibliography

\appendix

\section{Properties of the cycle group \texorpdfstring{$\mathcal{C}_G$}{Lg}} \label{app:CGProperties}

The cycle group $\mathcal{C}_G$ is defined by Equation \ref{eq:cycleGroup}. Here we prove some properties of this group. We say a cycle $c \in \mathcal{C}_G$ contains an edge if its expression as a product of edges contains that edge. We say a cycle $c$ contains a vertex if it contains an edge incident on that vertex.

\begin{prop}
$\mathcal{C}_G$ is in the centralizer of $M_G$.
\end{prop}

\begin{proof}
Consider an element $c \in \mathcal{C}_G$. For every vertex contained in $c$, $c$ contains an even number of edges incident on that vertex. $M_G$ is generated by edge and vertex operators. First we note that $c$ commutes with all vertex operators since it contains an even number of edge operators incident on that vertex. Next we note that $c$ commutes with all edge operators, since for every edge operator $e$, $c$ contains an even number of edge operators not equal to $e$ and incident on common vertices with $e$ (this is true even if $c$ contains $e$). Thus $\mathcal{C}_G$ commutes with $M_G$.
\end{proof}

\begin{corollary}
$\mathcal{C}_G$ is an abelian normal subgroup of $M_G$.\end{corollary}

\begin{definition}[Eulerian Graph]
A graph is eulerian if each of its vertices has an even number of incident edges.
\end{definition}

\begin{definition}[Simple Cycle]
A connected eulerian subgraph in which all vertices have degree two.
\end{definition}

\begin{definition}[Cycle Space $\Gamma$ of $G$]
Given a graph $G=(\bf{E},\bf{V})$, let $\mathcal{E}=P(\bf{E})$ be the power set of edges in the graph, also known as the edge space. The edge space forms an abelian group, with the product operation being the disjunctive union $ab = a \cup b - a \cap b$. The cycle space $\Gamma$ is the set of edge spaces of eulerian subgraphs of $G$, and it is a subgroup under this group operation. Every element of the cycle space can be expressed as the product of simple cycles, thus there exists a basis -- called a cycle basis -- consisting of independent simple cycles which generates $\Gamma$.
\end{definition}

\begin{prop}
$\mathcal{C}_G$ is isomorphic to the cycle space $\Gamma$ of $G$.
\end{prop}

\begin{proof}

Every simple cycle $a \in \Gamma$ is a set of undirected edges which may be given an ordering and directedness such that one may traverse the cycle, passing over each vertex in $a$ exactly once. The choice of ordering is unique up to the choice of starting point and the direction of travel. For every simple cycle $a\in \Gamma$ define the function $f: \Gamma \rightarrow \mathcal{C}_G$ in accordance with an arbitrary choice of starting point and direction of travel
\begin{equation}
f(a) =i^{|a|} \prod_{e_{i,i+1} \in a} e_{i, i+1}.
\end{equation}
We may note that $f(a)$ is invariant under any choice of starting point and direction of travel by noting that a simple cycle never goes over the same vertex twice and so
\begin{equation}
e_{1, 2} e_{2,3}... e_{|a|,1} = e_{2,3}... e_{|a|,1}e_{1, 2}
\end{equation}
and furthermore inverting direction of travel yields
\begin{align}
e_{1, |a|}... e_{3,2}e_{2,1} &= (-1)^{|a|}e_{|a|, 1}... e_{2,3}e_{1,2} \\
&=(-1)^{|a|}e_{1,2} (e_{|a|, 1}... e_{2,3})\\
&=(-1)^{|a|} e_{1,2} e_{2,3} (-1)^1(e_{|a|, 1}... e_{3,4})\\
&=(-1)^{|a|} e_{1,2} e_{2,3}... e_{|a|-1, |a|} (-1)^{|a|-2}e_{|a|, 1}\\
&= e_{1, 2} e_{2,3}... e_{|a|,1}
\end{align}
Thus one may uniquely retrieve $a$ from $f(a)$, and so $f$ is invertible on the simple cycles. 
Choose a cycle basis $B$ of $\Gamma$. Define $f$ on all of $\Gamma$ as the product of its action on $B$, ie if $c = \prod b_i \;, b_i \in B$, then $f(c) = \prod_i  f(b_i)$.
We can see by construction that $f$ is a homomorphism. Furthermore, since $f$ is invertible on all simple cycles, it is invertible on the cycle basis, and is thus an isomorphism.
\end{proof}

\begin{prop}
$M_E$ is isomorphic to $M_G/\mathcal{C}_G$
\end{prop}

\begin{proof} 
It is not difficult to identify the isomorphism. For all vertices $i$ in $G$ define the invertible transformation
\begin{equation}
f: v_i C_G \rightarrow V_i 
\end{equation}
for every edge $\{v_i,v_j\}$ in $G$ define the invertible transformation
\begin{equation}
f: e_{ij} \mathcal{C}_G \rightarrow E_{ij} \;,\; e_{ji}\mathcal{C}_G \rightarrow E_{ji} 
\end{equation}
Since $G$ is connected, all other edge operators $E_{ij}$ for $\{v_i, v_j\}$ not in $G$ may be constructed by composition of defined edge operators, using the recursion relation
\begin{equation}
E_{ij}E_{jk} = - i E_{ik} 
\end{equation}
By this construction Relation \ref{eq:Eloopcond} is satisfied for cyclic paths not in $G$, and it is also satisfied for paths in $G$ since 
\begin{align}
f^{-1}\left(\ii^{|L|}\prod_{x=1}^{|L|-1}E_{i_x,i_{x+1}} \right) &= \ii^{|L|}\prod_{x=1}^{|L|-1}f^{-1}(E_{i_x,i_{x+1}})\\
&= \ii^{|L|}\prod_{x=1}^{|L|-1}e_{i_x,i_{x+1}} \mathcal{C}_G \\
&= \mathbb{I}\mathcal{C}_G 
\end{align}


\end{proof}

\typeout{get arXiv to do 4 passes: Label(s) may have changed. Rerun}
\end{document}

